\documentclass[pra,aps,superscriptaddress,twocolumn,nofootinbib,longbibliography,accepted=2018-11-19]{quantumarticle}

\usepackage{graphicx,color,amsmath,amsfonts,enumerate,amsthm,amssymb,mathtools,enumitem,thmtools,hyperref,subfigure,mathdots,enumitem,centernot,bm,soul,bbm}
\usepackage[capitalise, noabbrev]{cleveref}
\hypersetup{colorlinks=true,linkcolor=blue,citecolor=blue,urlcolor=blue}

\def\E{ {\cal E} }
\def\F{ {\cal F} }

\def\Z{ {\cal Z} }

\def\>{\rangle}
\def\<{\langle}

\newcommand{\ket}[1]{| {#1} \rangle}
\newcommand{\abs}[1]{\left| {#1} \right|} 
\newcommand{\ketbra}[2]{\ensuremath{\left|#1\right\rangle\!\!\left\langle#2\right|}}
\newcommand{\matrixel}[3]{\ensuremath{\left\langle #1 \vphantom{#2#3} \right| #2 \left| #3 \vphantom{#1#2} \right\rangle}}
\newcommand{\tr}[1]{\mathrm{Tr}\left( #1 \right)}
\newcommand{\trr}[2]{\mathrm{Tr}_{#1}\left( #2 \right)}
\newcommand{\iden}{\mathbbm{1}}

\renewcommand{\v}[1]{\ensuremath{\boldsymbol #1}}

\DeclareMathOperator*{\argmin}{arg\,min}
\DeclareMathOperator*{\argmax}{arg\,max}

\definecolor{ppblue}{RGB}{46,117,182}
\definecolor{ppred}{RGB}{197, 90, 17}

\newcommand{\bmtilde}[1]{\bm{\tilde{#1}}}
\newcommand{\bmhat}[1]{\bm{\hat{#1}}}
\newcommand{\norm}[1]{\left\|#1\right\|}
\DeclareMathOperator{\Var}{Var}

\newlength{\figheight}
\newlength{\figwidth}
\theoremstyle{plain}
\newtheorem{thm}{Theorem}
\newtheorem{lem}[thm]{Lemma}
\newtheorem{prop}[thm]{Proposition}
\newtheorem{cor}[thm]{Corollary}
\theoremstyle{definition}
\newtheorem{defn}[thm]{Definition}
\newcommand{\includeTikz}[1]{\includegraphics{#1}}

\newcommand{\includeTikzz}[1]{\includeTikz{#1}}

\begin{document}

\title{Beyond the thermodynamic limit: finite-size corrections to state interconversion rates}

\author{Christopher T.~Chubb}
\email{paper@christopherchubb.com}
\affiliation{Centre for Engineered Quantum Systems, School of Physics, University of Sydney, Sydney NSW 2006, Australia.}
\orcid{0000-0002-2668-1567}

\author{Marco Tomamichel}
\affiliation{Centre for Quantum Software and Information, School of Software, University of Technology Sydney, Sydney NSW 2007, Australia.}
\orcid{0000-0001-5410-3329}

\author{Kamil Korzekwa}
\affiliation{Centre for Engineered Quantum Systems, School of Physics, University of Sydney, Sydney NSW 2006, Australia.}
\orcid{0000-0002-0683-5469}

\begin{abstract}
	Thermodynamics is traditionally constrained to the study of macroscopic systems whose energy fluctuations are negligible compared to their average energy. Here, we push beyond this thermodynamic limit by developing a mathematical framework to rigorously address the problem of thermodynamic transformations of finite-size systems. More formally, we analyse state interconversion under thermal operations and between arbitrary energy-incoherent states. We find precise relations between the optimal rate at which interconversion can take place and the desired infidelity of the final state when the system size is sufficiently large. These so-called second-order asymptotics provide a bridge between the extreme cases of single-shot thermodynamics and the asymptotic limit of infinitely large systems. We illustrate the utility of our results with several examples. We first show how thermodynamic cycles are affected by irreversibility due to finite-size effects. We then provide a precise expression for the gap between the distillable work and work of formation that opens away from the thermodynamic limit. Finally, we explain how the performance of a heat engine gets affected when the working body is small. We find that while perfect work cannot generally be extracted at Carnot efficiency, there are conditions under which these finite-size effects vanish. In deriving our results we also clarify relations between different notions of approximate majorisation.
\end{abstract}

\maketitle

\section{Introduction}

\paragraph*{Background.}

Thermodynamics forms an integral part of contemporary physics, providing us with invaluable rules that govern which transformations between macroscopic states are possible and which are not~\cite{giles1964mathematical}. In modern parlance thermodynamics is an example of a resource theory~\cite{coecke2016mathematical}. These provide us with a general framework to study the question of state interconversion by leveraging the structure imposed by the free states and operations within a given theory. The resource-theoretic approach has recently garnered renewed attention in the field of quantum information (see~\cite{horodecki2013quantumness} for a recent review) and allows us to quantify notions like entanglement~\cite{horodecki2009quantum}, coherence~\cite{baumgratz2014quantifying,marvian2016quantify,winter16} and asymmetry~\cite{marvian2013asymmetry,marvianthesis}. In the study of entanglement, for example, local operations and classical communication constitute the free operations and separable states are treated as free states. In an analogous way, a rigorous resource-theoretic formulation of quantum thermodynamics was provided in Refs.~\cite{janzing2000thermodynamic,horodecki2013fundamental,brandao2013resource}, with thermal Gibbs states being free and the laws of thermodynamics being captured by the restricted set of free operations known as \emph{thermal operations}.

One of the main problems studied within the resource theory of thermodynamics is single-shot state interconversion, i.e., identifying when it is possible to convert a given state to another using only free operations (or, alternatively, identifying what extra resources are necessary to enable such transformations). Although for general quantum states only partial results are known~\cite{lostaglio2015quantum,cwiklinski2015limitations,narasimhachar2015low} (see also Ref.~\cite{gour2017quantum} for the most recent progress), the full solution was found for the restricted case of transforming states with no coherence between distinct energy eigenspaces. For such \emph{energy-incoherent states} the necessary and sufficient conditions for single-shot state interconversion are given by a thermomajorisation relation between the initial and final states~\cite{horodecki2013fundamental}. As a result the allowed transformations are in general irreversible, which is captured by the fact that the amount of work needed to create a given state, the \emph{work of formation}, is larger than the amount of work one can extract from it, the \emph{distillable work}. We note that formally the thermomajorisation condition is strongly related to the majorisation condition appearing within the resource theory of entanglement when studying transformations between pure bipartite states~\cite{nielsen1999conditions}.

An important variant of the interconversion problem, lying on the opposite extreme of the single-shot case, is asymptotic state interconversion. In this case one considers having access to arbitrarily many copies of the initial state, and asks for the maximal conversion rate at which it is possible to transform instances of one state to another with asymptotically vanishing error. It was found that this rate is given by the ratio of non-equilibrium free energies of the initial and target states~\cite{brandao2013resource}. Thus, the asymptotic interconversion rate is directly linked with the amount of useful work that can be extracted on average from a given state. Moreover, in this regime all transformations become fully reversible, as work of formation and distillable work asymptotically coincide. Again, this result closely resembles that obtained in pure state entanglement theory, where the optimal interconversion rate is given by the ratio of the entanglement entropies of the initial and target states~\cite{bennett1996concentrating}.

In this work we study the interconversion problem in an intermediate regime, between the single-shot and asymptotic cases described above. We thus consider transformations of a finite number $n$ of instances of the input state and tolerate a non-zero error, which affects the optimal conversion rate. By developing the notions of approximate majorisation~\cite{horodecki2017approximate} and thermomajorisation, we extend a formal relationship between the resource theories of pure state entanglement and thermodynamics of incoherent states to the approximate case. This allows us to adapt recently developed tools for approximate entanglement transformations in Ref.~\cite{kumagai2017second} to study corrections to the asymptotic rates of thermodynamic transformations, the so-called second-order asymptotics (see, e.g., Refs~\cite{tomamichel12,li12,tomamicheltan14,datta14,datta15,wilde15b,tomamichel16,wildetomamichel17} for other recent studies of second-order asymptotics in quantum information). The crucial technical difference between Ref.~\cite{kumagai2017second} and our work is the reversed direction of the majorisation relation, resulting in free states being given by uniform states rather than pure states. The second-order corrections were known to scale as $1/\sqrt{n}$~\cite{brandao2013resource}, and our main technical contribution is identifying the exact constant (including its dependence on the error) and interpreting its thermodynamic meaning.

\medskip
\paragraph*{Motivation.}

Probably the most famous thermodynamic result concerns the irreversible nature of thermodynamic transformations, and is often captured by the oversimplified statement ``entropy has to grow''. The dynamics of a system interacting with a thermal bath is irreversible since transformations performed at finite speed lead to heat dissipation, resulting in a loss of information about the system. One thus often studies idealised scenarios, when the system undergoes changes so slowly that it stays approximately in thermal equilibrium at all times. In this quasi-static limit one recovers reversible dynamics. However, thermodynamic reversibility actually requires one more assumption that is usually made implicitly. Namely, the thermodynamic description is only valid when applied to systems whose energy fluctuations are much smaller than their average energy. This is true for macroscopic systems composed of $n\rightarrow\infty$ particles, in the so-called thermodynamic limit, and is reflected by the reversibility of the interconversion problem in the asymptotic limit. However, for finite $n$ the macroscopic results do not hold anymore, leading to another source of irreversibility.

In the emerging field of quantum thermodynamics (see Ref.~\cite{goold2016role} and references therein) a focus is placed on possible transformations of small quantum systems interacting with a thermal environment. The necessity to go beyond classical thermodynamics is motivated by the fact that at the nanoscale quantum effects, like coherence~\cite{aberg2014catalytic,lostaglio2015description,lostaglio2015quantum,cwiklinski2015limitations,uzdin2015equivalence,korzekwa2016extraction} and entanglement~\cite{alicki2013entanglement,perarnau2015extractable,brask2015autonomous}, start playing an important role. However, beyond these phenomena, in the quantum regime one also deals with systems composed of a finite number $n$ of particles. Hence, thermodynamic transformations of such systems are affected by the effective irreversibility discussed in the previous paragraph. The results we present in this paper provide a mathematical framework to rigorously address this problem. We thus provide a bridge between the extreme case of single-shot thermodynamics with $n=1$ and the asymptotic limit of $n\rightarrow\infty$, allowing us to study the irreversibility of thermodynamic processes in the intermediate regime of large but finite $n$. 

\medskip
\paragraph*{Main Results.}

In order to state our main results we first need to introduce some concepts that will be defined more formally in Sections~\ref{sec:setting} and~\ref{sec:prelim}. Let us consider a finite-dimensional quantum system, characterised by its Hamiltonian $H$, in the presence of a thermal bath at fixed temperature $T$. The  initial state of the system is in general out of thermal equilibrium, and the bath can be governed by an arbitrary Hamiltonian. Energy-conserving operations that interact the system with a bath in thermal equilibrium are then known as thermal operations. A simple example of a thermal operation just swaps the system with a bath (governed by the same Hamiltonian), replacing the initial state with a thermal Gibbs state. Gibbs states for $H$, denoted by $\gamma$, are thus free states in the resource-theoretic formulation of thermodynamics.

In the following we focus on a system comprised of a finite number $n$ of non-interacting subsystems, each governed by the Hamiltonian $H$. Let us consider a pair of subsystem states $\rho$ and $\sigma$ that both commute with $H$.\footnote{In a particular case of trivial Hamiltonian $H\propto\iden$, $\rho$ and $\sigma$ can be arbitrary quantum states.} Our results will be expressed in terms of two information quantities: the relative entropy~\cite{umegaki62} with the Gibbs state, $D(\cdot\|\gamma)$, and the relative entropy variance~\cite{tomamichel12,li12} with the Gibbs state, $V(\cdot\|\gamma)$. These quantities can also be interpreted thermodynamically: $k_B T \cdot D(\rho\|\gamma)$ is the difference between the generalised free energies of $\rho$ and $\gamma$ (with $k_B$ denoting the Boltzmann constant); and $V(\rho\|\gamma)$ is proportional to a generalised heat capacity of the system. The latter interpretation is justified since for \mbox{$\rho=\gamma'$} being a Gibbs state at temperature $T'\neq T$, the quantity $V(\gamma'\|\gamma)$ is proportional to the heat capacity at $T'$. We also note that $D(\rho\|\gamma)$ vanishes if and only if $\rho=\gamma$, whereas $V(\rho\|\gamma)$ vanishes whenever $\rho$ is proportional to the Gibbs state on the support of $\rho$, e.g., when $\rho$ is pure.

Let us now consider the problem of thermodynamic state interconversion between a finite number of instances of a state and for a fixed inverse temperature $\beta$ of the background bath. Formally, we are looking for the maximal rate $R$ for which there exists a thermal operation $\mathcal{E}^\beta$ such that $\mathcal{E}^\beta(\rho^{\otimes n}) = \tilde{\sigma}$ for some state $\tilde{\sigma}$ on $R n$ subsystems that is sufficiently close to $\sigma^{\otimes R n}$. To measure the proximity of two quantum states we will use infidelity, i.e.\ we require that $F(\sigma^{\otimes Rn}, \tilde{\sigma}) \geq 1 - \epsilon$ for some accuracy parameter \mbox{$\epsilon \in (0,1)$}, where $F(\cdot,\cdot)$ denotes Uhlmann's fidelity~\cite{uhlmann85}. The maximal conversion rate, denoted by $R^*(n, \epsilon)$, depends on both the number of subsystems~$n$ and the accuracy~$\epsilon$. We can assume that neither the initial state $\rho$ nor the target state $\sigma$ are the thermal state $\gamma$, as otherwise the interconversion problem is trivial. We then find the following expansions of $R^*(n, \epsilon)$ in $n$:
\begin{subequations}
\begin{align}
\label{eq:intro/general}	
&R^*(n, \epsilon) \simeq \frac{D(\rho\|\gamma)}{D(\sigma\|\gamma)} \left( 1 \!+\! \sqrt{\frac{V(\rho\|\gamma)}{n\, D(\rho\|\gamma)^2}}\, Z_{1/\nu}^{-1}(\epsilon) \right) \\
&\quad\simeq \frac{D(\rho\|\gamma)}{D(\sigma\|\gamma)} \left( 1 \!+\! \sqrt{\frac{V(\sigma\|\gamma)}{n\, D(\rho\|\gamma) D(\sigma\|\gamma)}}\, Z_{\nu}^{-1}(\epsilon) \right) \!,
\label{eq:intro/creation}
\end{align}
\end{subequations}
where $Z^{-1}_{\nu}$ is the inverse of the cumulative function of Rayleigh-normal distribution $Z_{\nu}$ introduced in Ref.~\cite{kumagai2017second} with $\nu$ given by
\begin{align}
\nu = \frac{V(\rho\|\gamma)/D(\rho\|\gamma)}{V(\sigma\|\gamma)/D(\sigma\|\gamma)}  \,,
\end{align}
and $\simeq$ denotes equality up to terms of order \mbox{$o(1/\sqrt{n})$}. We note that $Z_0 = \Phi$ is the cumulative normal distribution function and $Z_1$ is the cumulative Rayleigh distribution function. The inverse of the cumulative Rayleigh-normal distribution is typically negative for small values of $\epsilon$ (unless $\nu = 1$), and thus the finite-size correction term that scales as $1/\sqrt{n}$ is generally negative. For the special case \mbox{$V(\rho\|\gamma)=V(\sigma\|\gamma)=0$} (when $\nu$ is undefined) we provide an exact formula for \mbox{$R^*(n, \epsilon)$}, up to all orders in $n$.

In deriving our results we also prove an important relation between two different notions of approximate majorisation~\cite{horodecki2017approximate}. More precisely, we show that \emph{pre-} and \emph{post-majorisation}, which hold when the majorisation relation holds up to the smoothing of the major\emph{ising} or major\emph{ised} distribution, are equivalent. We further extend these concepts to thermomajorisation, which allows us to rigorously address the problem of approximate thermodynamic transformations.

\medskip
\paragraph*{Discussion.}

One of the main applications of our result is to the study of thermodynamic irreversibility. In the asymptotic limit, $n\rightarrow\infty$, the optimal conversion rate $R^*$ from $\rho$ to $\sigma$ is equal to the inverse of the conversion rate from $\sigma$ to $\rho$~\cite{brandao2013resource}. We can thus transform $\rho^{\otimes n}$ through $\sigma^{\otimes R^*n}$ back to $\rho^{\otimes n}$, so that the rate of concatenated transformations $R_r^*$ is equal to 1 and the process can be performed reversibly. However, using Eq.~\eqref{eq:intro/general} twice, one finds the correction term to reversibility rate $R_r^*$, which is proportional to $1/\sqrt{n}$. Moreover, if $\nu\neq 1$ this correction term is negative for small errors. In fact, it diverges when the error approaches zero, preventing a perfect reversible cycle. However, pairs of states with equal ratios of relative entropy and relative entropy variance with respect to the thermal state (such that $\nu=1$) are reversibly interconvertible up to second-order asymptotic corrections, mirroring a recent result in entanglement theory~\cite{ito2015asymptotic}. Thus, $\nu$ can be interpreted as the irreversibility parameter that quantifies the amount of infidelity of an approximate cyclic process.

One particular consequence of the discussed irreversibility is the difference between the distillable work, $W_D$, and the work of formation, $W_F$, for a given state $\rho$~\cite{horodecki2013fundamental}. The former is defined as the maximal amount of free energy in the form of pure energy eigenstates $\psi$ that can be obtained per copy of $\rho$; the latter as the minimal amount of free energy in the form of pure energy eigenstates $\psi$ needed per copy to create the target state $\rho$. We note that in the special case when $\psi$ is chosen to be the ground state, the distillation process can be considered as Landauer erasure (resetting to zero energy pure states), whereas the formation process can be seen as the action of a Szilard engine (creating states out of information). In single-shot thermodynamics $W_D$ and $W_F$ were shown to be proportional to max- and min-relative entropies with respect to the thermal state~\cite{horodecki2013fundamental}; while in the asymptotic scenario they are both equal to \mbox{$W = k_B T \cdot D(\rho\|\gamma)$}, the non-equilibrium free energy of a state~\cite{brandao2013resource}. Here, using appropriately modified Eqs.~\eqref{eq:intro/general} and~\eqref{eq:intro/creation}, we show that for large~$n$ the values of distillable work and work of formation per particle lie symmetrically around the asymptotic value, $W_D \simeq W - \Delta W$ and $W_F \simeq W + \Delta W$, and provide the exact expression for the gap $\Delta W$. Moreover, in the special case when the investigated state $\rho$ is itself a thermal state at some temperature different from the background temperature, $\Delta W$ can be directly related to the relative strength of energy fluctuations of the system.

Finally, we investigate how one can investigate the performance of heat engines with finite working bodies using an appropriately chosen interconversion scenario. This allows us to study finite-size corrections to the efficiency of a heat engine and the quality of work it performs. More precisely, we consider a heat engine operating between two thermal baths at temperatures $T_{\mathrm{h}}$ and $T_{\mathrm{c}}$, and a finite working body composed of $n$ particles initially at temperature $T_{\mathrm{c}}$. We show that, unless the irreversibility parameter satisfies $\nu=1$, near-perfect work can be performed only with efficiency lower than the Carnot efficiency $\eta_{\mathrm{C}}$. However, allowing for imperfect work allows one to achieve and even surpass $\eta_{\mathrm{C}}$~\cite{ng2017surpassing}. Moreover, we find that it is possible for a finite working body to have two thermal states at different temperatures, $T_{\mathrm{c}}$ and $T_{\mathrm{c}'}$, such that the irreversibility parameter for them is equal to 1. Thus, in a particularly engineered setting, it is possible to achieve Carnot efficiency and perform perfect work, while the finite working body changes temperature from $T_{\mathrm{c}}$ to $T_{\mathrm{c}'}$.

\medskip
\paragraph*{Overview.}

The remainder of this paper is organised as follows. We first describe the resource-theoretic approach to thermodynamics in Section~\ref{sec:setting} and introduce necessary mathematical concepts used within the paper in Section~\ref{sec:prelim}. In Section~\ref{sec:main_result} we state our main result concerning state interconversion under thermal operations, and discuss its thermodynamic interpretation and possible applications. We then proceed to Section~\ref{sec:aux_results}, where we present auxiliary results concerning approximate majorisation and thermomajorisation, which we believe may be of independent interest. The technical proof of the main result can be found in Section~\ref{sec:proof}. We conclude with an outlook in Section~\ref{sec:outlook}.

\section{Thermodynamic setting}
\label{sec:setting}

\subsection{Thermal operations}

We begin by describing the resource-theoretic approach to the thermodynamics of finite-dimensional quantum systems in the presence of a single heat bath at temperature $T$~\cite{horodecki2013fundamental,brandao2015second}. The investigated system is described by a Hamiltonian \mbox{$H=\sum_i E_i\ketbra{E_i}{E_i}$} and prepared in a general state $\rho$, whereas the bath, with a Hamiltonian $H_B$, is in a thermal equilibrium state,
\begin{equation}
\label{eq:gibbs_state}
\gamma_B=\frac{e^{-\beta H_B}}{\tr{e^{-\beta H_B}}},
\end{equation}
where $\beta=1/k_B T$ is the inverse temperature with $k_B$ denoting the Boltzmann constant.\footnote{Note that within the paper we will refer interchangeably to systems at temperature $T_{\mathrm{x}}$ or inverse temperature $\beta_{\mathrm{x}}$.} The evolution of the joint system is assumed to be closed, so that it is described by a unitary operator $U$, which additionally conserves the total energy,
\begin{equation}
\label{eq:energy_conservation}
[U,H+H_B]=0.
\end{equation} 
The central question now is: what are the possible final states that a given initial state $\rho$ can be transformed into?

More formally, one defines the set of \emph{thermal operations}~\cite{janzing2000thermodynamic}, which describes the free operations of the resource theory of thermodynamics, i.e., all possible transformations of the system that can be performed without the use of additional resources (beyond the single heat bath). These are defined as follows:
\begin{defn}[Thermal operations] 
	\label{def:thermal_ops}
	Given a fixed inverse temperature $\beta$, the set of thermal operations $\{\E^\beta\}$ consists of completely positive trace-preserving (CPTP) maps that act on a system $\rho$ with Hamiltonian $H$ as
\begin{equation}
\label{eq:thermal_ops}
\E^\beta(\rho)=\trr{B}{U\left(\rho\otimes\gamma_B\right)U^{\dagger}},
\end{equation}	
with $U$ satisfying Eq.~\eqref{eq:energy_conservation}, $\gamma_B$ given by Eq.~\eqref{eq:gibbs_state}, and $H_B$ being arbitrary.
\end{defn}
Note that energy conservation condition, Eq.~\eqref{eq:energy_conservation}, can be interpreted as encoding the first law of thermodynamics; whereas the fact that the bath is in thermal equilibrium leads to $\E^\beta(\gamma)=\gamma$, with $\gamma$ being the thermal Gibbs state of the system (i.e., given by Eq.~\eqref{eq:gibbs_state} with $H_B$ replaced by $H$), thus encoding the second law. 

\subsection{Thermodynamic state interconversion}

The thermodynamic interconversion problem is stated as follows: given a system (i.e., fixing $H$) together with initial and target states, $\rho$ and $\sigma$, does there exist a thermal operation $\E^\beta$ (for a fixed $\beta$) such that $\E^\beta(\rho)=\sigma$? The general answer for such a question is not known beyond the simplest qubit case~\cite{lostaglio2015quantum,cwiklinski2015limitations} (however, we note that the problem has very recently been solved for a larger class of free operations given by \emph{generalised thermal processes}~\cite{gour2017quantum} for coherent state interconversion). Nevertheless, for a restricted problem involving only \emph{energy-incoherent states}, i.e., $\rho$ and $\sigma$ commuting with $H$, the set of necessary and sufficient conditions was found~\cite{horodecki2013fundamental}. First, note that within this incoherent subtheory a quantum state can be equivalently represented by a probability distribution. For a non-degenerate Hamiltonian $H$, the initial and target states, $\rho$ and $\sigma$, that commute with $H$ are diagonal in the energy eigenbasis, so we can identify them with probability distributions $\v{p}$ and $\v{q}$, with \mbox{$p_i=\matrixel{E_i}{\rho}{E_i}$} and \mbox{$q_i=\matrixel{E_i}{\sigma}{E_i}$}. For degenerate Hamiltonians we note that unitaries within a degenerate energy subspace are thermal operations, so one can always diagonalise a state within such subspace for free. Therefore, in a general case the components of $\v{p}$ and $\v{q}$ representing $\rho$ and $\sigma$ are simply given by the eigenvalues of $\rho$ and $\sigma$. Next, in Ref.~\cite{janzing2000thermodynamic} (see also Refs.~\cite{horodecki2013fundamental,korzekwa2016coherence} for an expanded discussion) the existence of a thermal operation between incoherent states was linked to the existence of a particular stochastic map via the following theorem.
\begin{thm}[Theorem~5 of Ref.~\cite{janzing2000thermodynamic}]
	\label{thm:TO_GP_equiv}
	Let $\rho$ and $\sigma$ be quantum states commuting with the system Hamiltonian $H$, and $\gamma$ its thermal equilibrium state. Denote their eigenvalues by $\v{p}$, $\v{q}$ and $\v{\gamma}$, respectively. Then there exists a thermal operation $\E^\beta$ such that $\E^\beta(\rho)=\sigma$ if and only if there exists  a stochastic map $\Lambda^\beta$ such that
		\begin{equation}
		\Lambda^\beta\v{\gamma}=\v{\gamma},\quad
		\Lambda^\beta\v{p}=\v{q} \label{eq:GP_conversion} \,.
		\end{equation}
\end{thm}
As a result, studying thermodynamic interconversion problem between energy-incoherent states, one can replace CPTP maps and density matrices with stochastic matrices and probability vectors. We will fully address this simplified problem in Section~\ref{sec:prelim}.

In this paper we study a particular variant of the general interconversion problem: the limit of asymptotically many copies of input and output states. Informally, we want to find the optimal rate $R^*$ allowing one to transform $n$ copies of an energy-incoherent state $\rho$ into $R^*n$ copies of another energy-incoherent state $\sigma$, as $n$ becomes large. Since the dimension of the input and output spaces must be the same\footnote{Strictly speaking this is only true because of our choice of Definition~\ref{def:thermal_ops}. More generally the partial trace in Eq.~\eqref{eq:thermal_ops} may be taken over arbitrary subsystem, not necessarily $B$. However, as it does not affect our analysis, we believe that it is more compelling to keep the dimension of the system under study fixed.}, we note that one can append any number of states in thermal equilibrium, $\gamma^{\otimes m}$, to both the initial state $\rho^{\otimes n}$, and target state $\sigma^{\otimes Rn}$. Physically, it is motivated by the fact that thermal states are free resources; mathematically, it comes from the fact that the bath Hamiltonian $H_B$ is arbitrary (so, in particular, it may contain $m$ copies of the system, effectively adding $\gamma^{\otimes m}$ to the initial state) and that transforming any copy of the system into $\gamma$ is a thermal operation (so that the part of the final state beyond $\sigma^{\otimes R^*n}$ can always be replaced by $\gamma^{\otimes m}$). Therefore, we ask for the maximal value of $R^*$ (in the limit $n\rightarrow\infty$) for which there exists a thermal operation $\E^\beta$ satisfying
\begin{equation}
\label{eq:asymp_conv}
\mathcal{E}^\beta(\rho^{\otimes n}\otimes\gamma^{\otimes R^*n})\approx\sigma^{\otimes R^*n}\otimes\gamma^{\otimes n},
\end{equation}
where $\approx$ denotes closeness in some distance measure, e.g.\ infidelity or trace norm.

As already mentioned in the Introduction, our focus here is on $R^*$ for large but finite $n$, i.e., we look for corrections of order $1/\sqrt{n}$ to the optimal conversion rate coming from the finite number of systems involved in the thermodynamic process.

\section{Mathematical preliminaries}
\label{sec:prelim}

\subsection{Majorisation and embedding}
\label{sec:majorisation}

Unless otherwise stated we consider $d$-dimensional probability distributions and their products. We define the uniform state~$\v{\eta}$ and the thermal Gibbs state $\v{\gamma}$ at inverse temperature~$\beta$ as 
\begin{subequations}
	\begin{eqnarray}
		\v{\eta}&:=&\frac{1}{d}[1,\dots,1],\\
		\v{\gamma}&:=&\frac{1}{\Z}\left[e^{-\beta E_1},\dots,e^{-\beta E_d}\right],
	\end{eqnarray}
\end{subequations}
with $E_i$ denoting the eigenvalues of $H$ and \mbox{$\Z=\sum_i e^{-\beta E_i}$} being the partition function of the system. Moreover, we we call a distribution $\v{f}$ \emph{flat} if all its non-zero entries are equal. Note that in the infinite and zero temperature limits, $\beta\to 0$ and $\beta \to \infty$ respectively, the thermal state $\bm\gamma$ becomes flat. Specifically, in the former case $\bm\gamma\to\bm \eta$, and in the latter \mbox{$\bm\gamma\to\bm s:=[1,0,\dots,0]$} for Hamiltonians with non-degenerate ground spaces.

The most general transformation between two probability distributions is given by a stochastic matrix $\Lambda$ satisfying $\Lambda_{ij}\geq 0$ and $\sum_i \Lambda_{ij} =1$. We denote by $\Lambda^\beta$ a \emph{Gibbs-preserving} stochastic matrix with a thermal fixed point, i.e., $\Lambda^\beta\v{\gamma}=\v{\gamma}$. In particular, a matrix $\Lambda^0$ that preserves the uniform distribution $\v{\eta}$ is also called \emph{bistochastic}. A probability vector $\v{p}$ is said to majorise $\v{q}$, denoted by $\v{p}\succ\v{q}$, if and only if
\begin{equation}
\sum_{i=1}^k p_i^\downarrow\geq \sum_{i=1}^k q_i^\downarrow,
\end{equation}
for all $k\in\{1,\dots,d\}$, with $\v{p}^\downarrow$ denoting a probability vector with entries of $\v{p}$ arranged in a non-increasing order. We then have the following central result that is used in the study of state interconversion: 

\begin{thm}[Theorem~II.1.10 of Ref.~\cite{bhatia2013matrix}]
	\label{thm:major}
	There exists a bistochastic matrix mapping from $\v{p}$ to $\v{q}$ if and only if $\v{p} \succ \v{q}$, i.e.\
	\begin{align}
	\exists \Lambda^0:~\Lambda^0\v{\eta}=\v{\eta}~\mathrm{and}~\Lambda^0\bm p=\bm q\quad \Longleftrightarrow\quad\bm p\succ \bm q.
	\end{align}
\end{thm}

The above theorem can be generalised from bistochastic matrices $\Lambda^0$ to Gibbs-preserving matrices $\Lambda^\beta$ with arbitrary $\beta$. To achieve this we first need to introduce the following embedding map.
\begin{defn}[Embedding map]
	\label{def:embedding}
	Given a thermal distribution $\v{\gamma}$ with rational entries, \mbox{$\gamma_i=D_i/D$} and \mbox{$D_i,D\in\mathbb{N}$}, the embedding map $\Gamma^\beta$ sends a $d$-dimensional probability distribution $\v{p}$ to a $D$-dimensional probability distribution $\hat{\v{p}}:=\Gamma^\beta(\v{p})$ as follow:
	\begin{equation}
	\hat{\v{p}}=\left[\phantom{\frac{i}{i}}\!\!\!\!\right.\underbrace{\frac{p_1}{D_1},\dots,\frac{p_1}{D_1}}_{D_1\mathrm{~times}},
	\,\dots\,,\underbrace{\frac{p_d}{D_d},\dots,\frac{p_d}{D_d}}_{D_d\mathrm{~times}}\left.\phantom{\frac{i}{i}}\!\!\!\!\right]. \label{eq:embedding}
	\end{equation}
	The potentially irrational values of $\gamma_i$ can be approached with arbitrarily high accuracy by choosing $D$ large enough.
\end{defn}
\noindent Note that $\hat{\v{\gamma}}=\v{\eta}^D$, i.e., embedding maps a thermal distribution into a uniform distribution over $D$ entries, and that $\Gamma^\beta$ is injective, implying the existence of a left inverse $(\Gamma^\beta)^{-1}$. The action of $(\Gamma^\beta)^{-1}$ on a $D$-dimensional vector $\v{r}$ is given by summing up all the entries belonging to the same block of $D_i$ entries. Moreover, $\Gamma^\beta(\Gamma^\beta)^{-1}$ is a bistochastic map, that transforms each block of $D_i$ entries into a uniform distribution, i.e., given an index $j$ belonging to a block $D_i$, denoted by $j\in[D_i]$, we have 
\begin{equation}
	[\Gamma^\beta(\Gamma^\beta)^{-1}\v{r}]_j=\frac{\sum_{k\in [D_i]} r_k}{D_i}.
\end{equation}
We can also introduce the embedded version of a matrix~$\Lambda^\beta$,
\begin{equation}
\hat{\Lambda}^\beta:=\Gamma^\beta\Lambda^\beta(\Gamma^\beta)^{-1}.
\end{equation}
Notice that $\hat{\Lambda}^\beta$ is a bistochastic matrix, as it clearly maps the set of $D$-dimensional probability distributions into itself, and it preserves the uniform state $\v{\eta}^D$,
\begin{equation}
\hat{\Lambda}^\beta\v{\eta}^D=\Gamma^\beta\Lambda^\beta(\Gamma^\beta)^{-1}\hat{\v{\gamma}}=\Gamma^\beta\Lambda^\beta\v{\gamma}=\Gamma^\beta\v{\gamma}=\v{\eta}^D.
\end{equation}

Using the notion of embedding we can define \emph{thermomajorisation} relation~\cite{horodecki2013fundamental} (originally introduced in Ref.~\cite{ruch1978mixing} as $d$-majorisation). A probability vector $\v{p}$ is said to thermomajorise $\v{q}$, denoted by $\v{p}\succ^\beta\v{q}$, if and only if the majorisation relation holds between the embedded versions of $\v{p}$ and $\v{q}$, i.e.,
\begin{equation}
\v{p}\succ^\beta\v{q} \quad \Longleftrightarrow\quad \hat{\v{p}}\succ\hat{\v{q}}.
\end{equation}
Note that for $\beta=0$ thermomajorisation becomes standard majorisation, as the embedding map is the identity matrix. Now, we have the following equivalence
\begin{equation}
\Lambda^\beta\v{p}=\v{q} \quad \Longleftrightarrow\quad \hat{\Lambda}^\beta\hat{\v{p}}=\hat{\v{q}},
\end{equation}
which, since $\hat{\Lambda}^\beta$ is bistochastic, allows us to use Theorem~\ref{thm:major} to obtain
\begin{cor}[Thermodynamic interconversion]
	\label{cor:thermo_major}
	There exists a Gibbs-preserving matrix mapping from $\v{p}$ to $\v{q}$ if and only if $\v{p}\succ^\beta \v{q}$, i.e.\
	\begin{align}
	\exists \Lambda^\beta:~\Lambda^\beta\v{\gamma}=\v{\gamma}~\mathrm{and}~\Lambda^\beta\v{p}=\v{q}\quad \Longleftrightarrow\quad \v{p}\succ^\beta \v{q}.
\end{align}
\end{cor}
Due to Theorem~\ref{thm:TO_GP_equiv}, Corollary~\ref{cor:thermo_major} specifies the necessary and sufficient conditions for energy-incoherent state interconversion under thermal operations. In Fig.~\ref{fig:interconversion_equiv} we present the chain of equivalence relations leading to this result.

\begin{figure}
	\centering
	\includeTikz{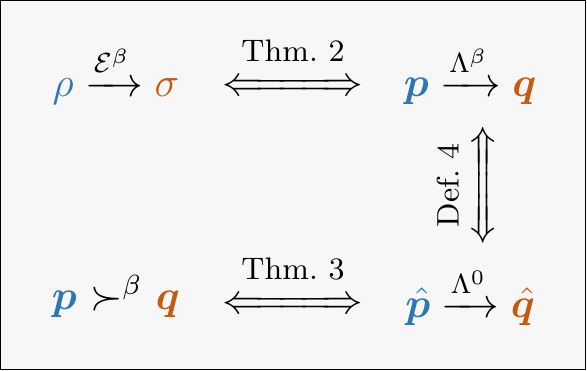}
	\caption{\label{fig:interconversion_equiv} \emph{Interconversion equivalence.} Quantum states $\rho$ and $\sigma$ are energy-incoherent, and their eigenvalues are given by $\v{p}$ and $\v{q}$, respectively. The arrow between states (distributions) symbolises the existence of a given map.}  
\end{figure}

\subsection{Information-theoretic notions and their thermodynamic interpretation}

The \emph{relative entropy} or Kullback-Leibler divergence of a probability distribution $\v{p}$ with $\v{q}$ is defined as
\begin{align}
	D(\v{p} \| \v{q}) := \sum_{i=1}^d p_i \log \frac{p_i}{q_i}  
	\label{eq:defD}
\end{align}
whenever the support of $\v{q}$ contains the support of $\v{p}$ (otherwise the divergence is set to $+\infty$). Denoting the average of a random variable $X$ in a state $\v{p}$ by
\begin{equation}
\langle X\rangle_{\v{p}}=\sum_i p_i X_i,
\end{equation}
we can introduce a random variable $L$ with
\begin{align}
\Pr\left[L=\log \frac{p_i}{q_i}\right] = p_i,
\end{align}
so that the divergence can be interpreted as the expectation value of the log-likelihood ratio, \mbox{$D(\v{p}\|\v{q})=\langle L\rangle_{\v{p}}$}. Similarly, we define the corresponding variance, the \emph{relative entropy variance}, as
\begin{align}
V(\v{p} \| \v{q}) := \Var_{\v{p}}(L),
\label{eq:defV}
\end{align}
where
\begin{align}
	\Var_{\v{p}}(X) := \langle (X-\langle X\rangle_{\v{p}})^2\rangle_{\v{p}}.
	\label{eq:defVar}
\end{align}
The following equalities are an immediate consequence of the embedding map introduced in Eq.~\eqref{eq:embedding}:
\begin{align}
	D(\hat{\v{p}} \| \hat{\v{q}}) = D(\v{p} \| \v{q})  \quad \textrm{and} \quad
	V(\hat{\v{p}} \| \hat{\v{q}}) = V(\v{p} \| \v{q}) \,. \label{eq:embedD}
\end{align}

In this work we will mostly encounter these quantities in the special case when $\v{q} = \v{\gamma}$ is the thermal distribution corresponding to inverse temperature $\beta$,
\begin{equation}
\gamma_i=\frac{e^{-\beta E_i}}{\Z},\quad \Z=\sum_{i=1}^d e^{-\beta E_i}.
\end{equation}
Then, both $D(\v{p}\|\v{\gamma})$ and $V(\v{p}\|\v{\gamma})$ can be interpreted thermodynamically. First note that
\begin{equation}
\label{eq:rel_ent_gibbs}
D(\v{p}\|\v{\gamma})=\beta\langle E\rangle_{\v{p}}-H(\v{p})+\log \Z,
\end{equation}
with $\langle E\rangle_{\v{p}}$ being the average energy and
\begin{equation}
H(\v{p}):=-\sum_{i=1}^dp_i\log p_i
\end{equation}
denoting the Shannon entropy of $\v{p}$ (as a function of a distribution $\v{p}$ it should not to be confused with the Hamiltonian $H$). Now, recall that the classical expression for free energy reads $U-TS$, with $U$ being the average energy of the system, $T$ the background temperature and $S$ the thermodynamic entropy; and that the free energy of the thermal state is $-k_B T \log \Z$. We thus see that $D(\v{p}||\v{\gamma})/\beta$ can be interpreted as a non-equilibrium generalisation of free energy difference between an incoherent state $\rho$ (represented by a probability distribution $\v{p}$) and a thermal state~$\gamma$. 

Now, to interpret $V(\v{p}\|\v{\gamma})$ let us first introduce a covariance matrix $M$ for the log-likelihood $\log p$ and energy in the units of temperature $\beta E$:
\begin{equation}
M=\left[\begin{array}{cc}
\mathrm{Cov}_{\v{p}}(\beta E,\beta E) & \mathrm{Cov}_{\v{p}}(\beta E,\log p)\\
\mathrm{Cov}_{\v{p}}(\log p,\beta E) & \mathrm{Cov}_{\v{p}}(\log p,\log p)
\end{array}
\right],
\end{equation}
where $\mathrm{Cov}_{\v{p}}(X,Y)=\langle XY \rangle_{\v{p}}-\langle X \rangle_{\v{p}}\langle Y \rangle_{\v{p}}$.
The relative entropy variance can then be expressed as
\begin{equation}
V(\v{p}\|\v{\gamma})=\sum_{i,j=1}^2 M_{ij}.
\end{equation}
In a particular case when the distribution $\v{p}$ is a thermal distribution $\v{\gamma}'$ at some different temperature \mbox{$T'\neq T$}, the expression becomes
\begin{equation}
\label{eq:var_capacity}
V(\v{\gamma}'||\v{\gamma})=\left(1-\frac{T'}{T}\right)^2 \cdot\frac{c_{T'}}{k_B},
\end{equation}
where
\begin{equation}
c_{T'}=\frac{\partial \langle E\rangle_{\v{\gamma}'}}{\partial T'}
\end{equation}
is the specific heat capacity of the system in a thermal state at temperature $T'$.

Finally, let us note that one can define quantum generalisations of both the relative entropy, $D(\rho||\sigma)$~\cite{umegaki62}, and the relative entropy variance, $V(\rho||\sigma)$~\cite{tomamichel12,li12}. Moreover, for a quantum state $\rho$ commuting with the Hamiltonian, these relative quantities with respect to the thermal state~$\gamma$ coincide with classical expressions,
\begin{align}
&D(\rho||\gamma)=D(\v{p}||\v{\gamma}),\quad V(\rho||\gamma)=V(\v{p}||\v{\gamma}),
\end{align}
where $\v{p}$ denotes the vector of eigenvalues of $\rho$.

\subsection{Approximate interconversion}
\label{sec:approximate}

As already mentioned in Section~\ref{sec:setting} we will focus on approximate interconversions, allowing the final state $\tilde{\v{q}}$ to differ from the target state $\v{q}$, as long as it is close enough. We measure distance between states using the \emph{infidelity},
\begin{equation}
	\delta(\v{p},\v{q}):=1-F(\v{p},\v{q}) \,,
\end{equation}
where the \emph{fidelity} (or Bhattacharyya coefficient) is
\begin{equation}
F(\v{p},\v{q}):= \left( \sum_{i=1}^d \sqrt{p_i q_i} \right)^2 \,.
\end{equation}
We will also use the fidelity $\F$ between two (continuous) probability density functions $f(x)$ and $g(x)$, defined as
\begin{equation}
\F(f,g):= \left( \int_{-\infty}^{\infty}\sqrt{f(x)g(x)}\, \mathrm{d} x  \right)^2.
\end{equation}
The two important properties of the infidelity that we will use throughout the paper are as follows. First, since fidelity is non-decreasing under stochastic maps 
we have
\begin{equation}
\delta(\Lambda^0\v{p},\Lambda^0\v{q})\leq\delta(\v{p},\v{q}).
\end{equation}
Second, the distance $\delta$ between two probability vectors is the same as between their embedded versions, i.e.,
\begin{equation}
\delta(\v{p},\v{q})=\delta(\hat{\v{p}},\hat{\v{q}}),
\end{equation}
which can be verified by direct calculation.

Although we will be mainly concerned with ``smoothing'' the final distribution (allowing it to differ from the desired target one), it is useful to introduce two dual definitions of approximate majorisation and thermomajorisation.
\begin{defn}[Pre- and post-thermomajorisation]
	\label{def:pre_post}
	A distribution $\v{p}$ \emph{$\epsilon$-pre-thermomajorises} a distribution $\v{q}$, which we denote $\v{p} \prescript{}{\epsilon}{\succ^\beta}~ \v{q}$, if there exists a $\tilde{\v{p}}$ such that 
	\begin{align}
	\tilde{\v{p}}\succ^\beta \v{q}\qquad\text{and}\qquad \delta(\v{p},\tilde{\v{p}})\leq \epsilon.
	\end{align}
	A distribution $\v{p}$ \emph{$\epsilon$-post-thermomajorises} a distribution $\v{q}$, which we denote $\v{p} \succ_\epsilon^\beta \v{q}$, if there exists a $\tilde{\v{q}}$ such that
	\begin{align}
	\v{p}\succ^\beta \tilde{\v{q}}\qquad\text{and}\qquad \delta(\v{q},\tilde{\v{q}})\leq \epsilon.
	\end{align}
	In a particular case of $\beta=0$, when thermomajorisation coincides with majorisation, we will speak of pre- and post-majorisation, denoted by $\prescript{}{\epsilon}{\succ}$ and $\succ_\epsilon$, respectively.
\end{defn}
Let us make a few comments about the above definition. First, notice that due to Corollary~\ref{cor:thermo_major}, $\v{p} \prescript{}{\epsilon}{\succ^\beta}~ \v{q}$ means that in the vicinity of $\v{p}$ there exists a state $\tilde{\v{p}}$ and $\Lambda^\beta$ that maps it to $\v{q}$. Similarly, $\v{p} \succ_\epsilon^\beta \v{q}$ means that there exists $\Lambda^\beta$ that maps $\v{p}$ to $\tilde{\v{q}}$, which lies in the vicinity of $\v{q}$. We illustrate this in Fig.~\ref{fig:pre_post}. Next note that both $\succ^\beta_0$ and $\prescript{}{0}{\succ^\beta}$ reduce to thermomajorisation $\succ^\beta$, specifically $\succ^0_0$ and $\prescript{}{0}{\succ^0}$ are equivalent to the standard majorisation relation~$\succ$. Let us also mention that the concept of majorisation between smoothed distributions has been recently studied in Refs.~\cite{renes2016relative,horodecki2017approximate,van2017smoothed}. Moreover, as with exact thermomajorisation, approximate thermomajorisation specifies the existence of a thermal operation between two energy-incoherent states. More precisely, due to Theorem~\ref{thm:TO_GP_equiv}, Corollary~\ref{cor:thermo_major} and Definition~\ref{def:pre_post}, we have the following:
\begin{cor}
	\label{cor:approx_thermal}
	Let $\rho$ and $\sigma$ be quantum states commuting with the system Hamiltonian $H$. Denote their eigenvalues by $\v{p}$ and $\v{q}$, respectively. Then there exists a thermal operation $\E^\beta$ such that $\delta(\E^\beta(\rho), \sigma) \leq \epsilon$ if and only if \mbox{$\v{p}\succ^\beta_\epsilon\v{q}$}. 
\end{cor}

Finally, let us make an important comment concerning approximate majorisation. Consider two distributions, $\v{p}$ and $\v{q}$, such that $\v{p}\succ_{\epsilon}\v{q}$. By definition there exists $\tilde{\v{q}}$ close to $\v{q}$ that is majorised by $\v{p}$. As majorisation is invariant under permutations, $\v{p}$ also majorises any distribution $\Pi\tilde{\v{q}}$, where $\Pi$ is arbitrary permutation. However, the fidelity between $\v{q}$ and $\Pi\tilde{\v{q}}$ is, in general, permutation-dependent. It is the largest, when the $i$-th largest entries of $\v{q}$ and $\Pi\tilde{\v{q}}$ coincide for all $i$, and so it is equal to \mbox{$F(\v{q}^\downarrow,\tilde{\v{q}}^\downarrow)$}. Therefore, for a given $\tilde{\v{q}}$ satisfying some majorisation relation, we know that for every state $\v{q}$ there exists $\Pi\tilde{\v{q}}$ satisfying the same relation, and with
\begin{align}
F(\Pi\tilde{\v{q}},\v{q})=F(\tilde{\v{q}}^\downarrow,\v{q}^\downarrow).
\end{align}
Thus, in the context of approximate majorisation, while calculating fidelities between any two states we will assume, without loss of generality, that they are ordered.

\begin{figure}
	\centering
	\includeTikz{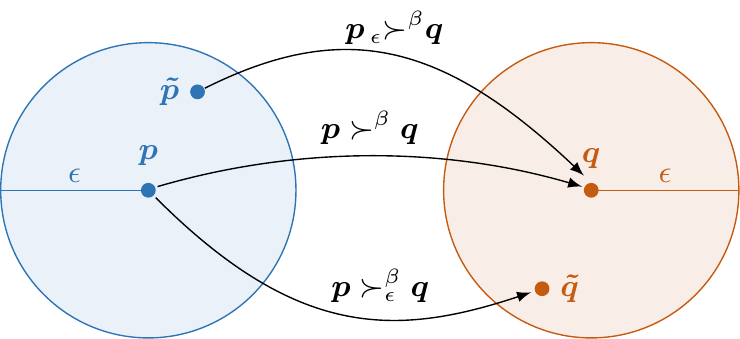}
	\caption{\label{fig:pre_post} \emph{Pre- and post-thermomajorisation.} Arrows depict the existence of Gibbs-preserving maps between corresponding distributions, whereas $\epsilon$-circles represent sets of probability distributions whose distance $\delta$ from $\v{p}$ and from $\v{q}$ is smaller than $\epsilon$.}
\end{figure}

\subsection{Asymptotic notation}

As we will be interested in approximating the optimal rate up to terms of order $O(1/\sqrt{n})$, we will adopt the following asymptotic notation for sequences $\{ a_n \}_n$ and $\{ b_n \}_n$ in $n \in \mathbb{N}$.
\begin{subequations}\begin{align}
	a_n &\simeq b_n & &\Longleftrightarrow & a_n - b_n & = o(1/\sqrt{n})\,, \\
	a_n &\lesssim b_n & &\Longleftrightarrow & a_n-b_n& \leq f_n = o(1/\sqrt{n})\,, \\
	a_n &\gtrsim b_n & &\Longleftrightarrow & a_n-b_n&\geq g_n = o(1/\sqrt{n}) \,,
\end{align}\end{subequations}
where $\{ f_n \}_n$ and $\{ g_n \}_n$ are auxiliary sequences that we usually do not introduce explicitly.

\subsection{Rayleigh-normal distributions}
\label{sec:rayleigh}

The dependence of the finite-size corrections to optimal interconversion rate on the infidelity is given by generalisations of the Gaussian distribution known as the Rayleigh-normal distributions. This family of distributions was first introduced in Ref.~\cite{kumagai2017second} in the context of LOCC entanglement conversion. In order to define it, let us first denote the Gaussian cumulative distribution function, with mean value $\mu$ and variance $\nu$, by $\Phi_{\mu,\nu}$,
\begin{align}
\Phi_{\mu,\nu}(x)=\frac{1}{\sqrt{2\pi\nu}}\int_{-\infty}^{x}e^{-\frac{(t-\mu)^2}{2\nu}}\,\mathrm dt.
\end{align}
As a shorthand notation we will also use $\Phi$ to denote $\Phi_{0,1}$. Following Ref.~\cite{kumagai2017second} we can now define
\begin{defn}[Rayleigh-normal distributions]
	For any $\nu>0$ the Rayleigh-normal distribution is a distribution on $\mathbb R$, whose cumulative function is given by
	\begin{align}
		Z_\nu (\mu):=1-\sup_{A\geq \Phi}\mathcal{F}\left(A',\Phi_{\mu,\nu}'\right)\!,
	\end{align}
	where the supremum is taken over all monotone increasing and continuously differentiable $A:\mathbb{R}\to [0,1]$ such that $A\geq \Phi$ pointwise; and $f'(x)$ denotes the derivative of $f(x)$.
\end{defn}

\begin{figure}
	\centering
	\setlength\figheight{8cm}
	\setlength\figwidth{8cm}
	\def\figscale{.5}
	\hspace{-.5cm}\includeTikz{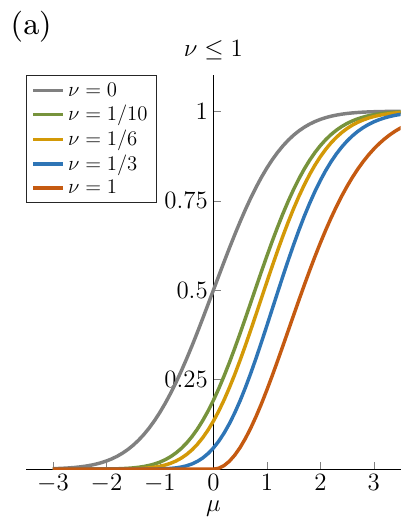}  \includeTikz{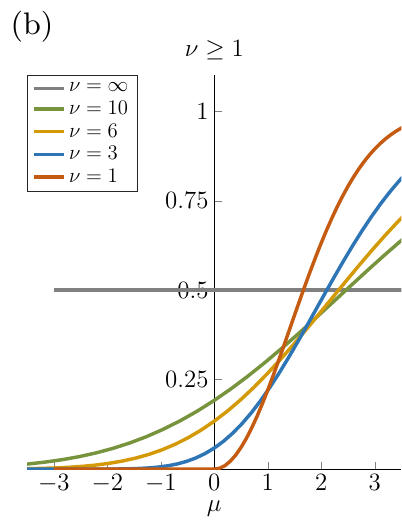}
	\caption{\emph{Rayleigh-normal distributions.} Plots of Rayleigh-normal cumulative probability distributions introduced in Ref.~\cite{kumagai2017second}. Note the difference between the above graphs and those presented in Fig.~1 of Ref.~\cite{kumagai2017second}. The parameter is chosen in the ranges (a) $\nu\in[0,1]$ and (b) $\nu\in[1,\infty]$. Due to the duality property, Eq.~\eqref{eq:duality}, the plots can be directly related to the ones presented in panel (a).
	\label{fig:rayleigh-normal}}
\end{figure}

We now present some relevant properties of the Rayleigh-normal distributions.
\begin{lem}[Section~2 of Ref.~\cite{kumagai2017second}]
	The Rayleigh-normal distributions have the following properties:
	\begin{itemize}
		\item The $\nu \to 0$ case converges in distribution to the normal Gaussian,
		\begin{align}
			\label{eq:rayleigh_limit}
			\lim\limits_{\nu \to 0}Z_\nu(\mu)=\Phi(\mu).
		\end{align}
		\item The $\nu=1$ case reduces to the Rayleigh distribution of scale parameter $\sigma=\sqrt{2}$,
		\begin{align}
			Z_{1}(\mu)=R_{\sqrt{2}}(\mu):=\begin{dcases}
			1-e^{-\mu^2/4} & \mu\geq 0,\\
			0 & \mu\leq 0.
			\end{dcases}
		\end{align}
		\item The Rayleigh-normal distributions possess a duality under inversion of the parameter $\nu$ of the form
		\begin{align}
			\label{eq:duality}
			Z_{1/\nu}(\mu)=Z_\nu (\sqrt{\nu}\mu).
		\end{align}
	\end{itemize}
\end{lem}

As well as these properties, an explicit form for the Rayleigh-normal distribution can be given. If we define $\alpha_{\mu,\nu}$ as the unique solution~\cite[Lemma 3]{kumagai2017second} to
\begin{align}
\frac{\Phi'(x)}{\Phi'_{\mu,\nu}(x)}=\frac{\Phi(x)}{\Phi_{\mu,\nu}(x)},
\end{align}
and let 
\begin{align}
	A_{\mu,\nu}:=\begin{dcases}
	\Phi_{\mu,\nu}(x)\frac{\Phi(\alpha_{\mu,\nu})}{\Phi_{\mu,\nu}(\alpha_{\mu,\nu})} & x\leq \alpha_{\mu,\nu},\\
	\Phi(x) & x\geq \alpha_{\mu,\nu},
	\end{dcases}
\end{align}
then for $\nu>1$ we have
\begin{align}
	\mathop{\mathrm{arg~max}}_{A\geq \Phi}\mathcal F\left(A',\Phi'_{\mu,\nu}\right)=A'_{\mu,\nu},
\end{align}
and so $Z_\nu(\mu)=1-\mathcal F\left(A'_{\mu,\nu},\Phi'_{\mu,\nu}\right)$~\cite[Theorem 4]{kumagai2017second}. Using the duality property, Eq.~\eqref{eq:duality}, a similar expression can be given for $\nu\in(0,1)$. We present plots of Rayleigh-normal distributions for a few selected values of $\nu$ in Fig.~\ref{fig:rayleigh-normal}.

\section{Second-order asymptotics for thermodynamic interconversion}
\label{sec:main_result}

\subsection{Statement of the main result}

We are now ready to state our main result concerning the second-order analysis of the approximate interconversion rates between independent and identically distributed (i.i.d.) states under thermal operations. We focus on initial and target states, $\rho$ and $\sigma$, that commute with the Hamiltonian $H$, so that we can represent them as probability distributions, $\v{p}$ and $\v{q}$, over their eigenvalues. For two fixed distributions $\bm p$ and $\bm q$ we will be interested in the trade-off between three parameters in the asymptotic $n\to\infty$ regime: the rate of conversion $R$, the infidelity $\epsilon$, and the inverse temperature of the bath $\beta$\footnote{Note that for fixed $H$ the inverse temperature $\beta$ fully specifies the thermal Gibbs distribution $\v{\gamma}$.}. Specifically, we will be interested in the triples $(\beta,\epsilon,R)$ for which there exist Gibbs-preserving maps $\Lambda^\beta$ such that
\begin{align}
\delta\Bigl( \Lambda^\beta\left(\bm p^{\otimes n}\otimes \bm \gamma^{\otimes Rn}\right) , \left(\bm q^{\otimes Rn}\otimes \bm \gamma^{\otimes n}\right)\Bigr)\leq \epsilon,
\end{align}
where $\bm \gamma$ denotes the Gibbs state at inverse temperature~$\beta$. By Corollary~\ref{cor:thermo_major} and Definition~\ref{def:pre_post} this condition is equivalent to approximate post-thermomajorisation,
\begin{align}
\bm p^{\otimes n}\otimes \bm \gamma^{\otimes Rn}
\succ_\epsilon^\beta
\bm q^{\otimes Rn}\otimes \bm \gamma^{\otimes n},
\end{align}
and, by Corollary~\ref{cor:approx_thermal}, there exists a thermal operation transforming $\rho^{\otimes n}$ into a state $\epsilon$ away in the infidelity measure from $\sigma^{\otimes R n}$.

We then define the optimal interconversion rate $R^{*}_\beta(n,\epsilon;\bm p,\bm q)$ and the optimal infidelity of interconversion $\epsilon^{*}_\beta(n,R;\bm p,\bm q)$ as
\begin{subequations}
	\begin{align}
	R^{*}_\beta:=&{\max}\left\lbrace R \,\middle|\, \bm p^{\otimes n}\otimes \bm\gamma^{\otimes Rn} \succ_{\epsilon}^\beta
	\bm q^{\otimes Rn}\otimes \bm \gamma^{\otimes n} \right\rbrace,\\
	\epsilon^{*}_\beta:=&\rlap{\,$\min$}\phantom{\max}\left\lbrace \epsilon \,\middle|\, \bm p^{\otimes n}\otimes \bm\gamma^{\otimes Rn} \succ_{\epsilon}^\beta
	\bm q^{\otimes Rn}\otimes \bm \gamma^{\otimes n} \right\rbrace.
	\end{align} 
\end{subequations}	
When it is clear from context we will drop the explicit dependence on $\bm p$ and $\bm q$. Our main result is then given by the following theorem.

\begin{thm}[Second-order asymptotic interconversion rates]
	\label{thm:sec-ord thermal}
	Let $\rho$ and $\sigma$ be energy-incoherent initial and target states with eigenvalues given by $\bm p$ and $\bm q$, respectively. Then, for inverse temperature $\beta$ and infidelity $\epsilon\in(0,1)$, the optimal interconversion rate has the following second-order expansions
	\begin{subequations}
	\begin{align}
	\label{eq:interconversion1}
	&R^*(n, \epsilon) \simeq \frac{D(\v{p}\|\v{\gamma})}{D(\v{q}\|\v{\gamma})} \left( 1 \!+\! \sqrt{\frac{V(\v{p}\|\v{\gamma})}{n\, D(\v{p}\|\v{\gamma})^2}}\, Z_{1/\nu}^{-1}(\epsilon) \right) \\
	& ~~\simeq \frac{D(\v{p}\|\v{\gamma})}{D(\v{q}\|\v{\gamma})} \left( 1 \!+\! \sqrt{\frac{V(\v{q}\|\v{\gamma})}{n\, D(\v{p}\|\v{\gamma}) D(\v{q}\|\v{\gamma})}}\, Z_{\nu}^{-1}(\epsilon) \right) \!,\label{eq:interconversion2}
	\end{align}
	\end{subequations}
where
\begin{align}
\label{eq:nu}
\nu = \frac{V(\v{p}\|\v{\gamma})/D(\v{p}\|\v{\gamma})}{V(\v{q}\|\v{\gamma})/D(\v{q}\|\v{\gamma})}
\end{align}
is the irreversibility parameter.
\end{thm}

The full proof of \cref{thm:sec-ord thermal} can be found in \cref{sec:proof}. Before presenting it, we will discuss some of its consequences and applications in \cref{sec:discussion}, and prove auxiliary results concerning approximate majorisation in \cref{sec:aux_results}. But first, let us make a few technical remarks about the above theorem. Note that Eqs.~\eqref{eq:interconversion1}-\eqref{eq:interconversion2} are simply related by the duality property of Rayleigh-normal distribution, Eq.~\eqref{eq:duality}. The reason to state both formulas is that this way one covers each of the special cases, \mbox{$V(\v{p}\|\v{\gamma})=0$} and \mbox{$V(\v{q}\|\v{\gamma})=0$}, avoiding the use of $Z^{-1}_{\infty}$, which is undefined. The special case when both relative entropy variances vanish is covered separately in \cref{subsec:trivial}, where an exact expression for \mbox{$R^*(n,\epsilon)$} is provided (the asymptotic expansion of which coincides with the appropriate limit of Eq.~\eqref{eq:interconversion1}). 

Furthermore, since all the involved states are energy-incoherent, one can replace probability distributions \mbox{$\v{p},\v{q},\v{\gamma}$} in \cref{thm:sec-ord thermal} with density matrices \mbox{$\rho,\sigma,\gamma$}. This way one can study interconversion between non-commuting states $\rho$ and $\sigma$, as long as they both commute with $H$. For example, if the Hamiltonian is trivial, $H\propto\iden$, $\rho$ and $\sigma$ may be arbitrary states. Thus, \cref{thm:sec-ord thermal} yields a complete second-order analysis of interconversion under noisy operations~\cite{horodecki2003local}, as for trivial Hamiltonians thermal operations coincide with noisy operations. 

Finally, using results originally derived in Ref.~\cite{vidal2000approximate} (see \cref{app:lemma_fidelity}), one can numerically evaluate the optimal interconversion rates. In \cref{app:numerics} we show that this algorithm can be executed with a runtime that is efficient in the system size. Using this, in \cref{fig:numerics2,fig:numerics1} we can compare our second-order expansion to the exact interconversion rates. We find that even for relatively small system sizes, the second-order asymptotic expansion gives a remarkably good approximation to the optimal interconversion rates, especially when compared to the first-order asymptotics.

\begin{figure*}
	\centering
	\setlength\figheight{7cm}
	\setlength\figwidth{22cm}
	\def\figscale{.8}
	\includeTikzz{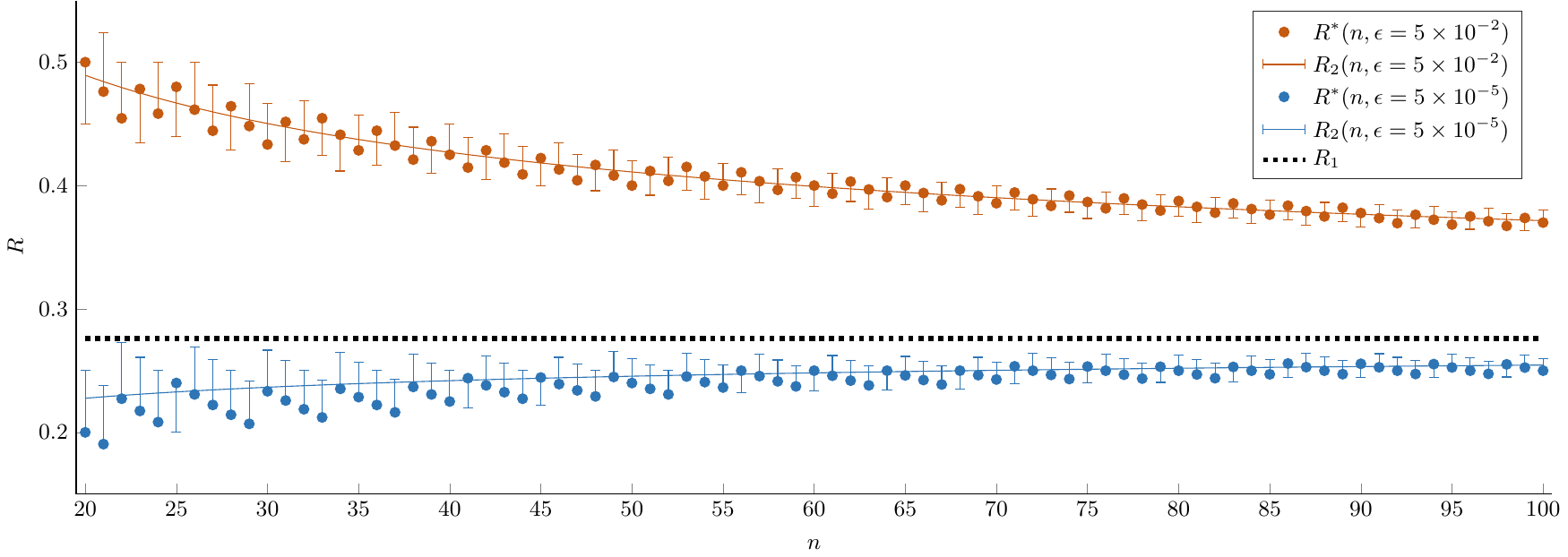}
	\caption{Comparison between the second-order approximation $R_2$ and exact thermal interconversion rates $R^*$, when converting from \mbox{$\rho=\frac{7}{10}\ketbra{0}{0}+\frac{3}{10}\ketbra{1}{1}$} to \mbox{$\sigma=\frac{8}{10}\ketbra{0}{0}+\frac{2}{10}\ketbra{1}{1}$}, with Hamiltonian \mbox{$H=\ketbra{1}{1}$} and access to a thermal bath at temperature \mbox{$1/\beta=3$}. The circles indicate exact conversion rates (c.f.\ \cref{app:numerics}), and the lines the second-order approximation given by Eq.~(\hyperref[eq:intro/general]{1}). As the exact interconversion rate is always a multiple of $1/n$, we have also indicated the rounding of the second-order approximation to the nearest multiples of $1/n$ with error bars. The colours indicate the infidelity tolerance, with $\epsilon=5\times 10^{-2}$ for red and $\epsilon=10^{-5}$ for blue. The dotted line indicates the asymptotic interconversion rate $R_1$. We plot the results for $n\leq20$ in \cref{fig:numerics1}.}
	\label{fig:numerics2}
\end{figure*}

\subsection{Discussion and applications}
\label{sec:discussion}

Although general state interconversion may seem to be a rather abstract problem, we will now show how the formalism can be applied to study more familiar thermodynamic scenarios. Since asymptotic conversion rates allow for reversible interconversion cycles and the results presented in the previous section describe finite-size corrections to these rates, our considerations will mainly revolve around irreversibility. We will first quantify it directly, by calculating the rate at which $n$ copies of a system can be transformed from initial state $\rho$, through $\sigma$, and back to $\rho$. We will then discuss the gap between work of formation and distillable work that opens when one processes finite number $n$ of systems. Finally, we will apply our results to study the performance of heat engines operating with finite-size working bodies.

\subsubsection{Finite-size reversibility}
\label{sec:application_irreversibility}

We start by considering the following thermodynamic process 
\begin{align}
\label{eq:cyclic}
\rho^{\otimes n}\rightarrow\sigma^{\otimes Rn}\rightarrow \rho^{\otimes R'Rn},
\end{align}
with optimal interconversion rates given by
\begin{align}
R=R_\beta^*(n,\epsilon_1;\v{p},\v{q}),\quad R'=R_\beta^*(Rn,\epsilon_2;\v{q},\v{p}),
\end{align}
where $\rho$ and $\sigma$ commute with the Hamiltonian and their eigenvalues are given by $\v{p}$ and $\v{q}$, respectively. Without the second-order asymptotic corrections derived in this work, the reversibility rate \mbox{$R_r^*:=RR'$} is equal to~1, and Eq.~\eqref{eq:cyclic} describes a perfect cyclic process illustrated in Fig.~\ref{fig:irreversibility}a. However, including finite-size corrections, from Eq.~\eqref{eq:interconversion1} we get 
\begin{align}
R_r^*\simeq&\left(1+\sqrt{\dfrac{V(\v{p}||\v{\gamma})}{nD(\v{p}||\v{\gamma})^2}}Z^{-1}_{1/\nu}(\epsilon_1)\right)\nonumber\\&\times \left(1+\sqrt{\dfrac{V(\v{q}||\v{\gamma})}{RnD(\v{q}||\v{\gamma})^2}}Z^{-1}_{\nu}(\epsilon_2)\right),
\end{align}
with the irreversibility parameter $\nu$ given by Eq.~\eqref{eq:nu}. Now, using the duality of Rayleigh-normal distribution, Eq.~\eqref{eq:duality}, and ignoring the terms of order $o(1/\sqrt{n})$ we obtain
\begin{align}
R_r^*\simeq 1+\sqrt{\frac{V(\v{p}||\v{\gamma})}{nD(\v{p}||\v{\gamma})^2}}\left(Z^{-1}_{1/\nu}(\epsilon_1)+Z^{-1}_{1/\nu}(\epsilon_2)\right).
\end{align}

The error is accumulated during both transformations appearing in Eq.~\eqref{eq:cyclic}. However, since the infidelity $\delta$ is not a metric, we cannot simply add the errors. Instead, we note that $\sqrt{\delta}$ \emph{is} a metric, and so it satisfies the triangle inequality. Thus, the total error $\epsilon$, i.e., the infidelity between the final state and the target state $\rho^{\otimes R_r^*n}$, satisfies \mbox{$\sqrt{\epsilon}\leq\sqrt{\epsilon_1}+\sqrt{\epsilon_2}$}. Actually, for $\epsilon_1+\epsilon_2<1$, one can obtain a tighter upper bound~\cite{tomamichel2015quantum},
\begin{align}
\label{eq:error_bound}
\epsilon\leq\left(\sqrt{\epsilon_1(1-\epsilon_2)}+\sqrt{\epsilon_2(1-\epsilon_1)}\right)^2.
\end{align}
Let us now introduce a threshold amount of infidelity,
\begin{align}
\label{eq:threshold}
\epsilon_0:=Z_{1/\nu}(0)=Z_{\nu}(0),
\end{align}
where the equality comes from duality of Rayleigh-normal distribution. Note that, if $\nu=1$, resulting in $\epsilon_0=0$, then for any finite error one can eventually achieve $R_r^*>1$, and a perfect transformation with $R_r^*=1$ and arbitrarily small error can be achieved. Thus, pairs of states satisfying $\nu=1$ are reversibly interconvertible up to second-order asymptotic corrections, analogously to a recent result in entanglement theory~\cite{ito2015asymptotic}. The use of such states in thermodynamic transformations is favourable, as it minimises the dissipation of free energy to the environment.

We will now show that the irreversibility parameter $\nu$ quantifies the incompatibility of two states (in that transformation from one state to the other leads to irreversibility) also beyond the special $\nu=1$ case. Consider a process in which one requires that the number of systems $n$ stays constant at all times. In other words, we require \mbox{$R=R'=1$}, which implies
\begin{align}
\epsilon_1=\epsilon_2=\epsilon_0.
\end{align}
Now, since $Z_{\nu}(0)\leq 1/2$, with equality achieved only for $\nu=0$ and $\nu\rightarrow\infty$, the error rates satisfy Eq.~\eqref{eq:error_bound} and the total error $\epsilon$ can be bounded by
\begin{align}
\epsilon\leq 4\epsilon_0(1-\epsilon_0)=4Z_{\nu}(0)\left(1-Z_{\nu}(0)\right).
\end{align}
We present the above bound as a function of $\nu$ in Fig.~\ref{fig:errors}a. We see that the closer $\nu$ is to~1, the less error will be induced while performing a thermodynamic transformation \mbox{$\rho^{\otimes n}\rightarrow\sigma^{\otimes n}\rightarrow\rho^{\otimes n}$} or, in other words, the more reversible the process will be.

\begin{figure}[t]
	\hspace{-.25cm}\includegraphics[width=\columnwidth]{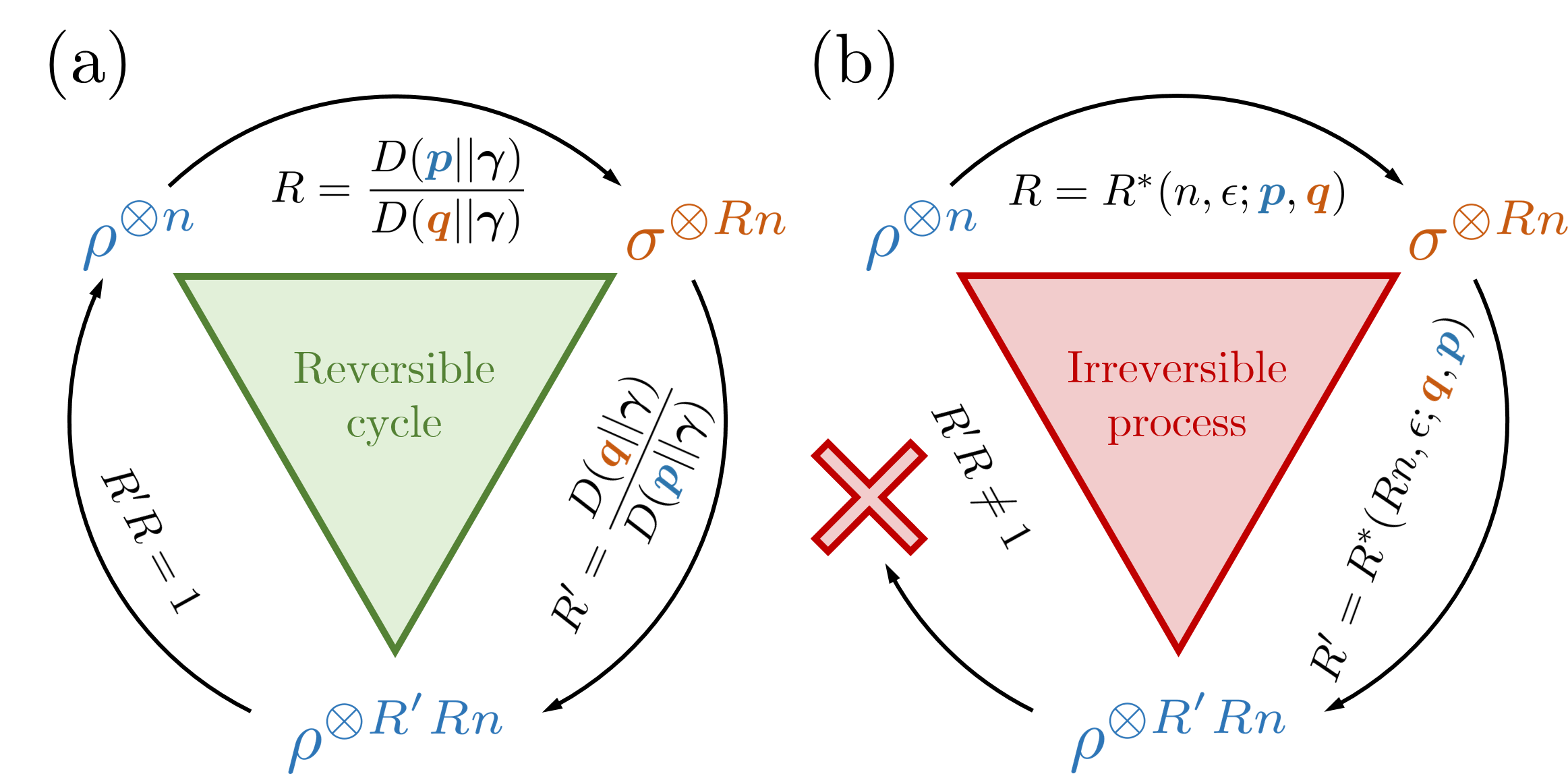}
	\caption{\label{fig:irreversibility} \emph{Finite-size irreversibility.} (a) In the asymptotic limit, $n\rightarrow\infty$, the optimal conversion rate from $\rho$ to $\sigma$ is equal to the inverse of the conversion rate from $\sigma$ to $\rho$. Therefore, reversible cycles can be performed. (b) In general, finite $n$ corrections to conversion rates for near-perfect interconversion are negative, leading to irreversibility with \mbox{$R'R<1$}.}
\end{figure}

\subsubsection{Distillable work and work of formation gap}

One particularly important consequence of irreversibility is the difference between distillable work and work of formation~\cite{horodecki2013fundamental}. These quantify the amount of thermodynamically relevant resources that can be distilled from, or are needed to form, a given state. Similarly to the resource theory of entanglement, where Bell states act as standard units of entanglement resource~\cite{horodecki2009quantum}, also within the resource theory of thermodynamics there are states acting as ``gold standards'' for measuring the amount of resources present in a state. These are given by pure energy eigenstates which, having zero entropy, have a clear energetic interpretation. The transformation requiring a change of an ancillary battery state $\ket{w}$, with energy $w$, into a state $\ket{0}$, with zero energy, is thus interpreted as performing work $w$; and a transformation allowing for an opposite change corresponds to extracting work $w$. Hence, in order to assess the thermodynamic resourcefulness of $n$ copies of a given energy-incoherent state, $\rho^{\otimes n}$, we will now investigate how much the energy of a pure battery system has to decrease per copy of $\rho$ to construct $\rho^{\otimes n}$, and how much can it increase per copy of $\rho$ while transforming $\rho^{\otimes n}$ to a thermal state?

\begin{figure}[t]
	\centering
	\setlength\figheight{4cm}
	\setlength\figwidth{5.5cm}
	\def\figscale{.7}
	\!\!\!\includeTikz{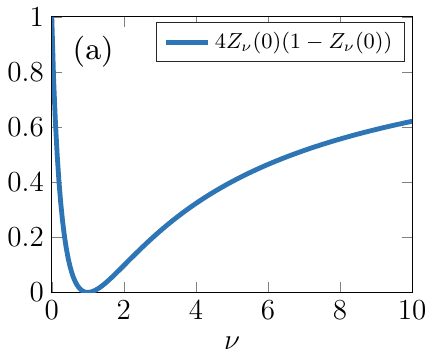}\!\!\includeTikz{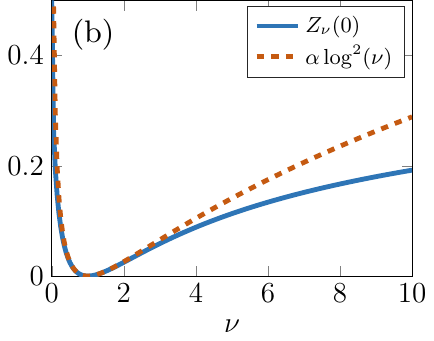}\!\!
	\caption{\emph{Bounds on the infidelity of thermodynamic transformations.} All plots are symmetric under $\nu\to 1/\nu$ transformation. (a) The upper bound on the total error accumulated in a cyclic process \mbox{$\rho^{\otimes n}\rightarrow\sigma^{\otimes n}\rightarrow\rho^{\otimes n}$} as a function of irreversibility parameter $\nu$. (b) Infidelity generated by a heat engine working at the Carnot efficiency during a process that heats up the finite working body from $T_{\mathrm{c}}$ to $T_{\mathrm{c}'}$ as a function of irreversibility parameter $\nu$ (that depends on both $T_{\mathrm{c}}$ and $T_{\mathrm{c}'}$, as well as on the hot bath temperature through Eq.~\eqref{eq:nu}). The optimal achievable infidelity during the process is plotted in solid line, while the bound on the infidelity generated during a continuous process (when the finite working body evolves through thermal states at all intermediate temperatures) is plotted in dashed line.
	\label{fig:errors}}
\end{figure}

More formally, to calculate work of distillation $W_D$ we want to find the maximal value $w$ allowing for the thermodynamic transformation
\begin{equation}
\label{eq:distillation}
\left( \rho\otimes\ketbra{0}{0}\right)^{\otimes n}\rightarrow \left( \gamma\otimes\ketbra{w}{w}\right)^{\otimes n},
\end{equation}
where the second subsystem is a battery described by a Hamiltonian \mbox{$H_B=0\ketbra{0}{0}+w\ketbra{w}{w}$}. Similarly, to calculate work of formation $W_F$ we want to find the minimal value $w$ allowing for the thermodynamic transformation
\begin{equation}
\label{eq:formation}
\left( \gamma\otimes\ketbra{w}{w}\right)^{\otimes n}\rightarrow \left( \rho\otimes\ketbra{0}{0}\right)^{\otimes n}.
\end{equation}
Using Theorem~\ref{thm:sec-ord thermal} we can obtain the optimal rate for transformation described by Eq.~\eqref{eq:distillation} as a function of $w$, set it to 1 and solve for $w$, thus arriving at the approximate expression for the work of distillation:
\begin{equation}
W_D\approx k_B T\left(D(\v{p}||\v{\gamma})+\sqrt{\dfrac{V(\v{p}||\v{\gamma})}{n}}\Phi^{-1}(\epsilon)\right),
\end{equation}
with $\v{p}$ denoting the eigenvalues of $\rho$ and $\Phi^{-1}$ being the inverse of the normal Gaussian distribution. One can obtain the expression for the work of formation in an analogous way, this time looking for the optimal rate for transformation given in Eq.~\eqref{eq:formation}, resulting in
\begin{equation}
W_F\approx k_B T\left( D(\v{p}||\v{\gamma})-\sqrt{\dfrac{V(\v{p}||\v{\gamma})}{n}}\Phi^{-1}(\epsilon)\right).
\end{equation}

\begin{figure}
	\centering
	\includegraphics[width=0.8\columnwidth]{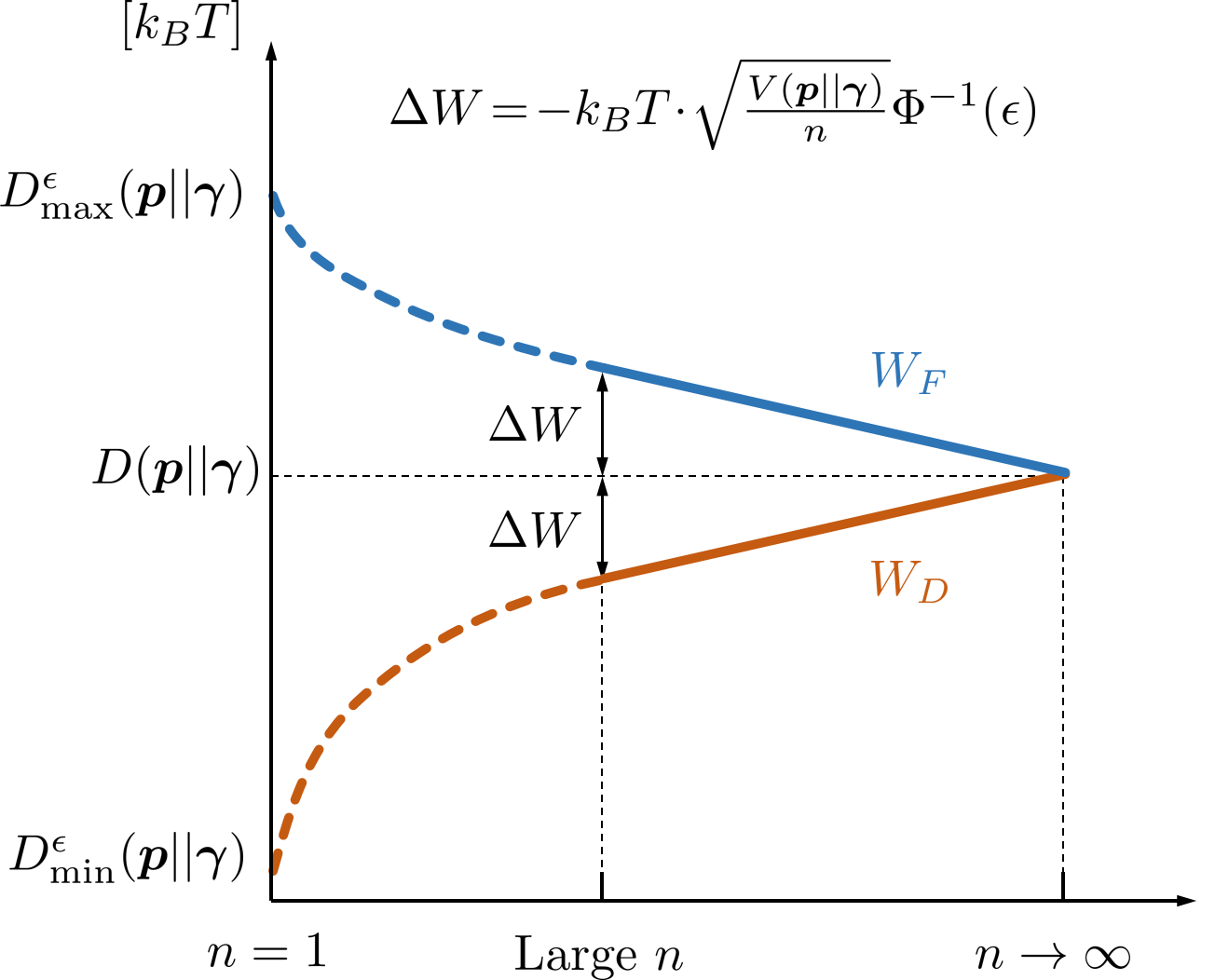}
	\caption{\label{fig:work_gap} \emph{Distillable work and work of formation gap.} The behaviour of distillable work $W_D$ and work of formation $W_F$ varies in different regimes. In single-shot scenarios they are proportional to min- and max-relative entropies~\cite{horodecki2013fundamental}. In the intermediate regime of large but finite $n$ studied in this work, the values of $W_D$ and $W_F$ lie symmetrically around the value achieved in the asymptotic limit, where $W_D$ and $W_F$ coincide and are equal to the non-equilibrium generalisation of free energy. Note that the $y$ axis above is in the units of $k_BT$.}
\end{figure}

First of all, let us briefly comment on the effect that imperfect transformations (characterised by infidelity $\epsilon$) have on the interpretation of distillable work and work of formation derived above. In the case of distillable work the non-zero infidelity means that the final battery state may differ from the pure state $\ketbra{w}{w}^{\otimes n}$, so one may actually distil less than $W_D$ work per particle. Similarly, for the work of formation the final state of the battery may differ from the pure state $\ketbra{0}{0}^{\otimes n}$, so one may actually use more than $W_F$ work per particle (by using the purity of the battery). To overcome such problems, one may employ the idea of $\epsilon$-deterministic work extraction~\cite{aberg2013truly} in the following way. After the distillation process (the argument for the formation process is analogous) one can simply measure the battery in its energy eigenbasis. With probability larger or equal to $1-\epsilon$ the battery state will collapse on $\ketbra{w}{w}^{\otimes n}$ (and so $n\cdot W_D$ work will be distilled), and with probability $\epsilon$ the work gain will differ from the derived value. Additionally, one has to take into account the thermodynamic cost of measuring the battery (erasing memory), proportional to the binary entropy of $\epsilon$ (note however that this cost is constant and so the cost per particle vanishes as $1/n$). Crucially, by choosing $\epsilon$ to be arbitrarily small, one can approach deterministic work distillation arbitrarily well, i.e., distil $n\cdot W_D$ work with probability arbitrarily close to 1.

Secondly, let us note that with the second-order asymptotic correction $W_D$ and $W_F$ lie symmetrically around the asymptotic value \mbox{$W=k_BT\cdot D(\v{p}||\v{\gamma})$},
\begin{align}
W_D\approx W-\Delta W,\quad W_F\approx W+\Delta W,
\end{align}	
with
\begin{align}
\Delta W:=-k_BT\cdot \sqrt{\dfrac{V(\v{p}||\v{\gamma})}{n}}\Phi^{-1}(\epsilon).
\end{align}
Notice that the above correction term is positive for small values of infidelity $\epsilon$, so that the resource cost of near-perfect formation of a state is always larger than the amount of resources than can be distilled from it. This symmetric gap that opens for finite $n$ is illustrated in Fig.~\ref{fig:work_gap}, where we also compare it with the values of $W_D$ and $W_F$ for the single-shot scenario $n=1$ (where $W_D$ and $W_F$ generally lie asymmetrically around the asymptotic value $W$). 

Furthermore, our second-order correction for distillable work exactly coincides with the one derived in Ref.~\cite{aberg2013truly} within an alternative thermodynamic framework, where state transformations are modelled by a sequence of energy level transformations (changes of Hamiltonian eigenvalues interpreted as performing/extracting work) and full thermalisations (replacing a state with the thermal state), rather than by thermal operations. This might have been expected, as the authors of Ref.~\cite{egloff2015measure} (and more recently of Ref.~\cite{perry2015sufficient}) showed that any transformation between energy-incoherent states which can be achieved via a thermal operation, can also be achieved by a sequence of level transformations and partial level thermalisations. 

Finally, let us analyse the special case when the state under scrutiny is itself a thermal equilibrium state $\v{\gamma}'$, at some temperature $T'$ different from the background temperature $T$. In the asymptotic limit $n\rightarrow\infty$, both distillable work $W_D$ and work of formation $W_F$ coincide with the standard thermodynamic result: the maximal (minimal) amount of work that can be extracted (needs to be invested) while changing the temperature of the system from $T'$ to $T$ (from $T$ to $T'$) is given by its free energy change. However, we also obtain the second-order asymptotic correction to $W_D$ and $W_F$, given by
\begin{align}
\Delta W =  -|T - T'| \sqrt{\frac{k_B c_{T'}}{n}} \,
\Phi^{-1}(\epsilon),
\end{align}
where we have used Eq.~\eqref{eq:var_capacity} to relate the relative entropy variance with $c_{T'}$, the heat capacity of the system at temperature $T'$. In order to interpret this correction term, we first note that standard thermodynamic results apply when fluctuations of energy are much smaller than the average energy of the system. Now, to quantify the relative strength of fluctuations, we introduce a fluctuation parameter $f$ as a ratio of the total energy variation and the total energy itself,
\begin{align}
f:=\frac{\sqrt{n\mathrm{Var}_{\v{\gamma}'}(E)}}{n\langle E\rangle_{\v{\gamma}'}}=\sqrt{\frac{k_Bc_{T'}}{n}}\cdot \frac{T'}{\langle E\rangle_{\v{\gamma}'}}.
\end{align}
We then see that the correction term $\Delta W$ can be expressed as
\begin{align}
\Delta W =  - f w  \,
\Phi^{-1}(\epsilon),\quad w:=\langle E\rangle_{\v{\gamma}'} \left|1 - \frac{T}{T'}\right|,
\end{align}
so that $\Delta W$ is directly related to the relative strength of fluctuations $f$, and disappears when the standard thermodynamic assumption, $f=0$, holds. Note that $\Phi^{-1}(\epsilon)$ is negative for $\epsilon<1/2$. Moreover, \mbox{$w$} is the amount of work performed by an engine operating at Carnot efficiency between two heat baths at temperatures $T$ and $T'$, when the amount of heat equal to $\langle E\rangle_{\v{\gamma}'}$ flows in to, or out of, the bath at temperature $T'$ (the former for $T>T'$, the latter for $T'>T$).

\subsubsection{Corrections to efficiency of heat engines}

One of the consequences of studying thermodynamics in the quantum regime is that it may not always be plausible to assume that working bodies and thermal reservoirs are infinite. Thus, recent studies focused on the effects finite-size baths have on standard thermodynamic results like fluctuation theorems~\cite{campisi2009finite}, Landauer's principle~\cite{reeb2014improved}, second~\cite{richens2018finite} and third law of thermodynamics~\cite{scharlau2018quantum}. Moreover, the performance of heat engines operating between finite-size baths was investigated in Refs.~\cite{tajima2017finite,ito2016optimal} (by directly focusing on the behaviour of thermodynamic quantities and not on the interconversion problem) and Ref.~\cite{woods2015maximum} (where the main focus was on the energetic structure of the finite bath and not on its size). Here, we will show how our results can be employed to investigate the performance of heat engines with finite working bodies, by studying the appropriately chosen interconversion problem. As we will discuss systems in equilibrium at different temperatures, we will indicate the (inverse) temperature in the subscript. More precisely, a system at temperature~$T_{\mathrm{x}}$ (at inverse temperature~$\beta_{\mathrm{x}}$) will be denoted by~$\gamma_{\mathrm{x}}$, and the corresponding partition function by $\Z_{\mathrm{x}}$. Also, note that equilibrium states are diagonal in the energy eigenbasis, so our results are applicable.

We consider two infinite baths at temperatures \mbox{$T_{\mathrm{h}}>T_{\mathrm{c}}$}, and a finite working body composed of $n$ particles initially at a cold temperature $T_{\mathrm{c}}$ (analogous considerations hold for the initial temperature being $T_{\mathrm{h}}$). As in the previous subsection, we also include a battery system comprised of $n$ two-level systems, each described by Hamiltonian $H_B=w\ketbra{w}{w}$, initially in a zero energy eigenstate $\ketbra{0}{0}^{\otimes n}$. We now couple the working body at temperature $T_{\mathrm{c}}$ and the battery to the hot bath, allowing us to perform a thermal operation with respect to temperature $T_{\mathrm{h}}$. In particular, we consider the following transformation
\begin{equation}
\label{eq:engine_transformation}
\left( \gamma_{\mathrm{c}}\otimes\ketbra{0}{0}\right)^{\otimes n}\rightarrow \left( \gamma_{\mathrm{c'}}\otimes\ketbra{w}{w}\right)^{\otimes n}.
\end{equation}
This transformation can be understood as a result of heat $Q_{\mathrm{in}}$ flowing from the hotter background bath into an engine; part of it, $Q_{\mathrm{out}}$, then heats up the working body composed of $n$ particles from $T_{\mathrm{c}}$ to $T_{\mathrm{c'}}$, while the remaining energy is used to perform work $n\cdot w$ on $n$ particles comprising the battery. We schematically present this thermodynamic process in Fig.~\ref{fig:heat_engine}.

\begin{figure}
	\centering
	\hspace{-.5cm}\includegraphics[width=\columnwidth]{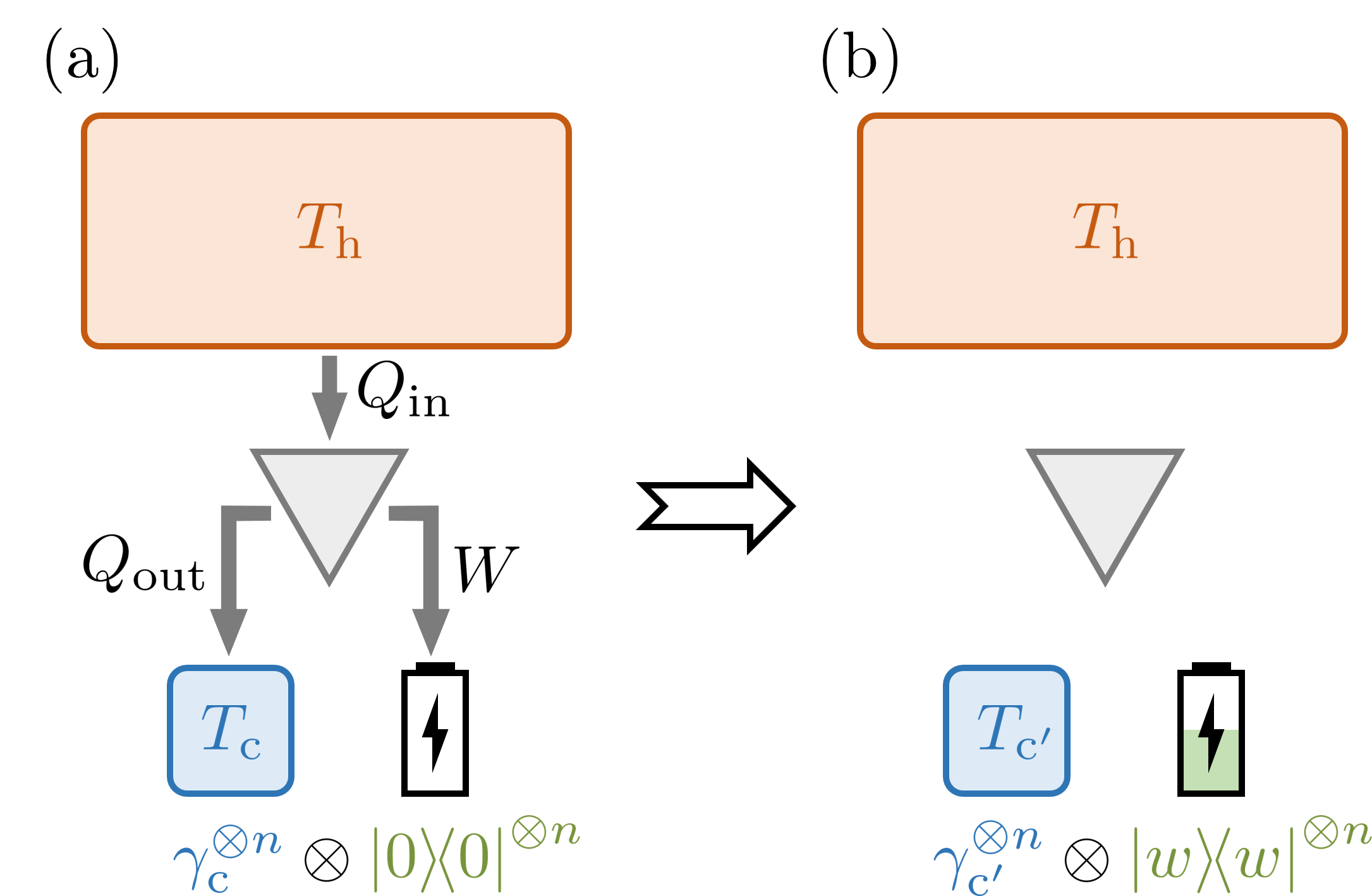}
	\caption{\label{fig:heat_engine} \emph{Performance of a heat engine as an interconversion problem.} 
		(a) The transformation of the working body (initially at cold temperature $T_{\mathrm{c}}$) and the battery (initially in the ground state), $(\gamma_{\mathrm{c}} \otimes \ketbra{0}{0})^{\otimes n} $, can be seen as heat flowing in and out of the engine that performs work on the battery. (b) As a result of such interconversion, i.e., the heat engine performance, the finite working body ends up at the intermediate temperature $T_{\mathrm{c}'}$, while the battery gets transformed to an excited state $\ketbra{w}{w}^{\otimes n}$.}
\end{figure}

The heat $Q_{\mathrm{out}}$ flowing into the working body is given by the change of energy,
\begin{equation}
\label{eq:q_out}
Q_{\mathrm{out}}=n\Delta E,
\end{equation} 
while the optimal amount of performed work $W=n\cdot w$ can be calculated similarly as in the previous subsection (by setting the rate from Eq.~\eqref{eq:interconversion1} for the transformation given by Eq.~\eqref{eq:engine_transformation} to 1 and solving for $w$), yielding
\begin{equation}
w\approx k_BT_{\mathrm{h}}\left(\Delta D + \sqrt{\frac{V(\v{\gamma}_{\mathrm{c}}||\v{\gamma}_{\mathrm{h}})}{n}} Z^{-1}_{1/\nu}(\epsilon) \right),
\label{eq:w}
\end{equation}
where we have introduced the following shorthand notation
\begin{subequations}
	\begin{align}
	\Delta E&:=\langle E\rangle_{\v{\gamma}_{\mathrm{c'}}}-\langle E\rangle_{\v{\gamma}_{\mathrm{c}}},\\
	\Delta D&:=D(\v{\gamma}_{\mathrm{c}}||\v{\gamma}_{\mathrm{h}}) - D(\v{\gamma}_{\mathrm{c'}}||\v{\gamma}_{\mathrm{h}}),
	\end{align}
\end{subequations}
and $\nu=V(\v{\gamma}_{\mathrm{c}}||\v{\gamma}_{\mathrm{h}})/V(\v{\gamma}_{\mathrm{c}'}||\v{\gamma}_{\mathrm{h}})$. Now, using energy conservation,
\begin{align}
Q_{\mathrm{in}}=Q_{\mathrm{out}}+W,
\end{align}
we can calculate the efficiency of the considered process to be
\begin{equation}
\eta=\frac{W}{Q_{\mathrm{in}}}=\left(1+\frac{Q_{\mathrm{out}}}{W}\right)^{-1},
\end{equation}
with $Q_{\mathrm{out}}$ and $W$ given by Eqs.~\eqref{eq:q_out} and~\eqref{eq:w}, respectively.

In order to interpret the obtained expression let us first analyse the limiting case. Ignoring the second-order asymptotic correction (sending \mbox{$n\rightarrow\infty$}), the extracted work is just equal to the change of the free energy of the working body. In Appendix~\ref{app:optimal_engine} we show that this is exactly the amount of work that would be extracted by an engine operating at Carnot efficiency,
\begin{align}
\eta_{\mathrm{C}}(T_{\mathrm{x}}):=1-\frac{T_{\mathrm{x}}}{T_{\mathrm{h}}},
\end{align}
between an infinite bath at fixed temperature $T_{\mathrm{h}}$ and a colder finite bath that heats up during the process from $T_{\mathrm{x}}=T_{\mathrm{c}}$ to $T_{\mathrm{x}}=T_{\mathrm{c'}}$. In other words, without the $1/\sqrt{n}$ correction we obtain an \emph{integrated} Carnot efficiency $\eta^{\mathrm{int}}_{\mathrm{C}}$ that arises from an \emph{instantaneous} Carnot efficiency $\eta_{\mathrm{C}}$ at all times,
\begin{align}
\eta^{\mathrm{int}}_{\mathrm{C}}(T_{\mathrm{c}}\rightarrow T_{\mathrm{c'}})=\left(1+\frac{\Delta E}{k_BT_{\mathrm{h}}\Delta D}\right)^{-1}\!\!\!\!\!\!.
\end{align}
The relation becomes even more evident when we consider the limit $\Delta T\rightarrow 0$. Then 
\begin{subequations}
\begin{align}
\Delta E&\xrightarrow{\Delta T\to 0} c_{T_{\mathrm{c}}} \Delta T,\\
\Delta D&\xrightarrow{\Delta T\to 0} c_{T_{\mathrm{c}}}\Delta T \left(\frac{1}{k_BT_{\mathrm{c}}}-\frac{1}{k_BT_{\mathrm{h}}}\right),
\end{align}
\end{subequations}
with $c_{T_\mathrm{c}}$ denoting the heat capacity of the system at temperature $T_{\mathrm{c}}$, and so
\begin{align}
\eta^{\mathrm{int}}_{\mathrm{C}}(T_{\mathrm{c}}\rightarrow T_{\mathrm{c'}})&\xrightarrow{\Delta T\to 0} \eta_{\mathrm{C}}(T_{\mathrm{c}}).
\end{align} 

Now, the finite-size correction leads to a modified expression for integrated efficiency, given by
\begin{align}
\eta^{\mathrm{int}}\approx \eta^{\mathrm{int}}_{\mathrm{C}}+\frac{k_B (T_{\mathrm{h}}-T_{\mathrm{c}})\Delta E}{(k_B T_{\mathrm{h}}\Delta D+\Delta E)^2}\sqrt{\frac{c_{T_{\mathrm{c}}}}{nk_B}}Z^{-1}_{1/\nu}(\epsilon),
\end{align}
where we have used Eq.~\eqref{eq:var_capacity} again to relate the relative entropy variance with the heat capacity of the system. We first note that there exists a threshold amount of infidelity~$\epsilon_0$, given by Eq.~\eqref{eq:threshold}, below which the correction term is negative. Since the infidelity between final and target states can be interpreted as performing imperfect work, near-perfect work can be performed only with efficiency strictly smaller than $\eta_{\mathrm{C}}^{\mathrm{int}}$. On the other hand, accepting infidelity $\epsilon\geq\epsilon_0$ allows one to achieve and even go beyond the integrated efficiency corresponding to instantaneous Carnot efficiency. This is in accordance with a recent result showing that the Carnot efficiency can be surpassed by extracting imperfect work~\cite{ng2017surpassing}.

As in the asymptotic limit, we also want to investigate the instantaneous efficiency, when $T_{\mathrm{c}'}$ is very close to $T_{\mathrm{c}}$. In particular, we will focus on the quality of performed work when the engine works at instantaneous Carnot efficiency. We thus require that the error $\Delta\epsilon$ accumulated during an infinitesimal step that changes the temperature by $\Delta T\rightarrow 0$ is equal to the threshold error $\epsilon_0$. Since then the correction term vanishes, we have \mbox{$\eta^{\mathrm{int}}\approx \eta^{\mathrm{int}}_{\mathrm{C}}$} and we know that for small $\Delta T$ this yields $\eta_{\mathrm{C}}$. Because the two considered thermal states are close, $\nu$ is close to unity, and as such the infidelity of the process can be expanded as
\begin{align}
\Delta\epsilon=Z_{\nu}(0)=Z_{1+\Delta\nu}(0)\approx\alpha\Delta\nu^2,
\end{align}
where $\alpha\approx 0.0545$ can be numerically evaluated. The expansion of $\Delta\nu$ in terms of $\Delta T$ is given by
\begin{align}
\Delta \nu\approx g(T_{\mathrm{c}})\Delta T,
\end{align}
with
\begin{align}
g(T_{\mathrm{c}}):=\frac{\mathrm d}{\mathrm dT_{\mathrm{x}}}\left[\log V(\v{\gamma}_{\mathrm{x}}||\v{\gamma}_{\mathrm{h}})\right]\Biggr|_{T_{\mathrm{x}}=T_{\mathrm{c}}}.
\end{align}
As discussed in Section~\ref{sec:application_irreversibility}, it is not infidelity, but its square root that satisfies the triangle inequality. We thus have that the instantaneous rate of accumulating square root infidelity is given by
\begin{align}
\frac{\mathrm d\sqrt{\epsilon}}{\mathrm dT_{\mathrm{x}}}\approx \sqrt{\alpha}|g(T_{\mathrm{x}})|,
\end{align}
and so one can achieve the instantaneous Carnot efficiency by paying the price of an instantaneous rate of accumulating error. This can be then translated into the bound on the total accumulated error in the following way
\begin{align}
\epsilon\leq \left(\int_{T_{\mathrm{c}}}^{T_{\mathrm{c}'}}\frac{\mathrm d\sqrt{\epsilon}}{\mathrm dT_{\mathrm{x}}}\,\mathrm dT_{\mathrm{x}}\right)^2 = \alpha\left(\int_{T_{\mathrm{c}}}^{T_{\mathrm{c}'}}|g(T_{\mathrm{x}})|\,\mathrm dT_{\mathrm{x}}\right)^2\!.
\end{align}
It is straightforward to show that this upper bound is larger than \mbox{$\alpha\log^2\nu$}, which in turn is larger than~$Z_{\nu}(0)$. This shows that the error accumulated in a continuous process (with the working body continuously passing through all intermediate temperatures) is in general larger than that of an optimal ``one-step'' process. We illustrate this in Fig.~\ref{fig:errors}b.

Finally, let us comment on a special case when $\nu=1$. For initial and target states being thermal equilibrium states at distinct temperatures (and different from background temperature $T_{\mathrm{h}}$), the value of $\nu$ depends on the Hamiltonian of the investigated system. If it happens that for a given Hamiltonian there exist $T_{\mathrm{x}}$ and $T_{\mathrm{x}'}$ such that $\nu=1$, then it is possible to achieve perfect work extraction at integrated Carnot efficiency $\eta_{\mathrm{C}}^{\mathrm{int}}(T_{\mathrm{x}}\rightarrow T_{\mathrm{x'}})$. Interestingly, for any Hamiltonian there always exist such pairs of temperatures. To see this note that for both $T_{\mathrm{x}}=0$ and $T_{\mathrm{x}}=T_{\mathrm{h}}$ the relative entropy variance vanishes, \mbox{$V(\v{\gamma}_{\mathrm{x}}||\v{\gamma}_{\mathrm{h}})=0$}. Since it is a continuous function of temperature, we get that for any $T_{\mathrm{x}}$ in the interval $(0,T_{\mathrm{h}})$ there exists at least one other temperature $T_{\mathrm{x'}}$ such that \mbox{$V(\v{\gamma}_{\mathrm{x}}||\v{\gamma}_{\mathrm{h}})=V(\v{\gamma}_{\mathrm{x'}}||\v{\gamma}_{\mathrm{h}})$}, resulting in $\nu=1$. This shows that by appropriately choosing the temperatures between which the heat engine operates, one may decrease or even avoid irreversible losses.

\section{Results on approximate majorisation}
\label{sec:aux_results}

We now proceed to the presentation of a few technical lemmas that may be of independent interest. These concern relations between different notions of approximate majorisation and thermomajorisation introduced in Section~\ref{sec:approximate}. We first need the following auxiliary result.
\begin{lem}
	\label{lemma:fidelity}
	For fixed probability vectors $\v{p}$ and $\v{q}$ denote by $\tilde{\v{p}}$ any distribution that majorises $\v{q}$, and by $\tilde{\v{q}}$ any distribution that is majorised by $\v{p}$. Then the maximum fidelity between $\tilde{\v{p}}$ and $\v{p}$ over all such $\tilde{\v{p}}$ is equal to the maximum fidelity between $\v{q}$ and $\tilde{\v{q}}$ over all such $\tilde{\v{q}}$, i.e., 
	\begin{equation}
	\label{eq:opt_fidelity}
	\max_{\tilde{\v{p}}:~\tilde{\v{p}}\succ\v{q}}F(\v{p},\tilde{\v{p}})=\max_{\tilde{\v{q}}:~\v{p}\succ\tilde{\v{q}}}F(\v{q},\tilde{\v{q}}).
	\end{equation}
\end{lem}
\noindent The proof of the above lemma is based on the results first derived in Ref.~\cite{vidal2000approximate} and can be found in Appendix~\ref{app:lemma_fidelity}. Moreover, the proof includes an explicit construction of the state $\tilde{\v{p}}^\star$ maximising the left hand side of Eq.~\eqref{eq:opt_fidelity}, so that one can calculate the value of optimal achievable fidelities appearing in Lemma~\ref{lemma:fidelity}. 

Now we can prove the following crucial result concerning pre- and post-majorisation, i.e., approximate thermomajorisation for $\beta=0$.
\begin{lem}
	\label{lem:pre_post_equiv}
	Pre- and post-majorisation are equivalent, i.e., $\v{p} \prescript{}{\epsilon}{\succ}\v{q}$ if and only if $\v{p} \succ_\epsilon\v{q}$.
\end{lem}
\begin{proof}
	First, assume that $\bm p \prescript{}{\epsilon}{\succ}~\bm q$. This means that there exists $\bm{\tilde p}$ such that $\bm{\tilde p}\succ \bm q$ and $\delta(\bm p,\bm{\tilde p})\leq\epsilon$. By \cref{thm:major} this implies that there exists a bistochastic $\Lambda^0$ such that $\Lambda^0\bm{\tilde p}=\bm q$. Let $\bm{\tilde q}:=\Lambda^0\bm p$, so that $\v{p}\succ\tilde{\v{q}}$. Using the fact that fidelity is non-decreasing under stochastic maps, we then have
	\begin{align}
	\delta(\bm q,\bm{\tilde q})
	=\delta(\Lambda^0\bm{\tilde p},\Lambda^0\bm p)
	\leq\delta(\bm{\tilde p},\bm p)
	\leq \epsilon,
	\end{align}
	which means that $\bm p \prescript{}{\epsilon}{\succ}~\bm q\implies\bm p\succ_\epsilon\bm q$.
	
	Now, assume that $\v{p}\succ_\epsilon\v{q}$. This means that there exists $\tilde{\v{q}}$ such that $\v{p}\succ \tilde{\v{q}}$ and $\delta(\v{q},\tilde{\v{q}})\leq\epsilon$. Let
	\begin{equation}
	\tilde{\v{p}}^{\star}:=\argmax_{\tilde{\v{p}}:~\tilde{\v{p}}\succ\v{q}}F(\v{p},\tilde{\v{p}}).
	\end{equation}
	By definition $\tilde{\v{p}}^{\star}\succ\v{q}$ and, by Lemma~\ref{lemma:fidelity}, we have
	\begin{equation}
	F(\v{p},\tilde{\v{p}}^{\star})=\max_{\v{q}':~\v{p}\succ\v{q}'}F(\v{q},\v{q}')\geq F(\v{q},\tilde{\v{q}}),
	\end{equation}
	so that $\delta(\v{p},\tilde{\v{p}}^{\star})\leq\epsilon$. Thus $\bm p \prescript{}{\epsilon}{\succ}~\bm q\Longleftarrow\bm p\succ_\epsilon\bm q$.
\end{proof}

The next lemma links post-majorisation of embedded vectors with post-thermomajorisation for $\beta\neq 0$.
\begin{lem}
	\label{lem:post_thermo_equiv}
	Post-majorisation between embedded vectors is equivalent to post-thermomajorisation between the original vectors, i.e.,
	\begin{equation}
	\hat{\v{p}}\succ_\epsilon\hat{\v{q}} \quad \Longleftrightarrow\quad \v{p}\succ_\epsilon^\beta\v{q}
	\end{equation}
\end{lem}

\begin{proof}
	First, assume $\hat{\v{p}}\succ_\epsilon \hat{\v{q}}$. This means that there exists a bistochastic matrix $\Lambda^0$ such that $\Lambda^0\hat{\v{p}}=\tilde{\hat{\v{q}}}$ with  \mbox{$\delta(\hat{\v{q}},\tilde{\hat{\v{q}}})\leq\epsilon$}. This, in turn, means that
	\begin{equation}
	\Lambda^\beta\v{p}=\Gamma^{-1}\tilde{\hat{\v{q}}},
	\end{equation}
	with $\Lambda^\beta=\Gamma^{-1}\Lambda^0\Gamma$ being a Gibbs-preserving matrix and $\Gamma$ a shorthand notation of the embedding map $\Gamma^\beta$. We thus conclude that \mbox{$\v{p}\succ^\beta \Gamma^{-1}\tilde{\hat{\v{q}}}$}. It remains to show that \mbox{$\delta(\v{q},\Gamma^{-1}\tilde{\hat{\v{q}}})\leq\epsilon$}. To achieve this we use the facts that embedding is fidelity-preserving, $\Gamma\Gamma^{-1}$ is bistochastic and fidelity is non-decreasing under stochastic maps, so that
	\begin{align}
		F(\v{q},\Gamma^{-1}\tilde{\hat{\v{q}}}) &= F(\Gamma^{-1}\hat{\v{q}},\Gamma^{-1}\tilde{\hat{\v{q}}}) \notag \\
		&= F(\Gamma\Gamma^{-1}\hat{\v{q}},\Gamma\Gamma^{-1}\tilde{\hat{\v{q}}}) \geq F(\hat{\v{q}},\tilde{\hat{\v{q}}}).
	\end{align}
	Therefore \mbox{$\delta(\v{q},\Gamma^{-1}\tilde{\hat{\v{q}}})\leq\epsilon$}, which results in $\v{p}\succ^\beta_\epsilon \v{q}$.
	
	Now, assume that $\v{p}\succ^\beta_\epsilon \v{q}$. This means that there exists a Gibbs-preserving matrix $\Lambda^\beta$ such that $\Lambda^\beta\v{p}=\tilde{\v{q}}$ with $\delta(\v{q},\tilde{\v{q}})\leq\epsilon$. Through embedding this is equivalent to the existence of a bistochastic $\hat{\Lambda}^\beta$ such that $\hat{\Lambda}^\beta\hat{\v{p}}=\hat{\tilde{\v{q}}}$, resulting in $\hat{\v{p}}\succ\hat{\tilde{\v{q}}}$. It remains to show that $\delta(\hat{\v{q}},\hat{\tilde{\v{q}}})\leq\epsilon$. This, however, follows directly from the fact that the embedding map is fidelity-preserving, as  \mbox{$\delta(\hat{\v{q}},\hat{\tilde{\v{q}}})=\delta(\v{q},\tilde{\v{q}})\leq\epsilon$}. We thus conclude that \mbox{$\hat{\v{p}}\succ_\epsilon \hat{\v{q}}$}.
\end{proof}

The statement of Lemma~\ref{lem:post_thermo_equiv} can be rephrased as
\begin{align}
&\exists \tilde{\hat{\v{q}}}:~\hat{\v{p}}\succ\tilde{\hat{\v{q}}},~\delta(\hat{\v{q}},\tilde{\hat{\v{q}}})\leq\epsilon\notag\\
&\qquad\qquad \!\!\Longleftrightarrow\quad\!\!
\exists \tilde{\v{q}}:~\hat{\v{p}}\succ\hat{\tilde{\v{q}}},~\delta(\hat{\v{q}},\hat{\tilde{\v{q}}})\leq\epsilon,
\end{align}
so that it can be interpreted as the fact that embedding (denoted by hat) and smoothing (denoted by tilde) commute when applied to the target state.

\begin{figure}[t!]
	\centering
	\includeTikz{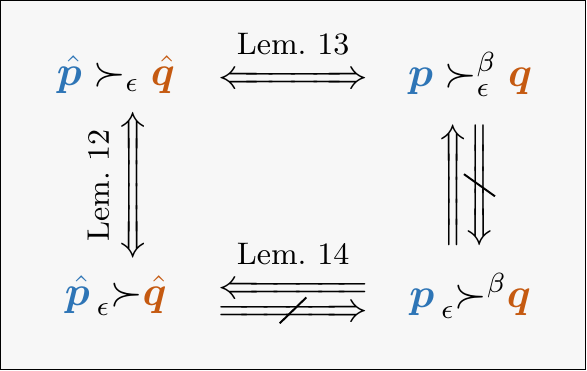}
	\caption{\label{fig:approximate_equiv} \emph{Relations between approximate pre- and post-thermomajorisation relations.} } 
\end{figure}

Finally, we present a result that links pre-majorisation of embedded vectors with pre-thermomajorisation for $\beta\neq 0$.
\begin{lem}
	\label{lem:pre_thermo_equiv}
	Pre-majorisation between embedded vectors is implied by the pre-thermomajorisation between the original vectors, but it does not imply it, i.e.,
	\begin{subequations}
		\begin{align}
		\label{eq:lemma6_1}
		\hat{\v{p}}\prescript{}{\epsilon}{\succ}\hat{\v{q}} \quad \Longleftarrow\quad \v{p}\prescript{}{\epsilon}{\succ^\beta}\v{q},\\
		\hat{\v{p}}\prescript{}{\epsilon}{\succ}\hat{\v{q}} \quad \centernot\Longrightarrow\quad \v{p}\prescript{}{\epsilon}{\succ^\beta}\v{q}.\label{eq:lemma6_2}
		\end{align}
	\end{subequations}
\end{lem}
\begin{proof}
	We first prove Eq.~\eqref{eq:lemma6_1}. Assuming \mbox{$\v{p}\prescript{}{\epsilon}{\succ^\beta}\v{q}$} (and recalling that embedding is fidelity preserving) means that there exists $\hat{\tilde{\v{p}}}$ majorising $\hat{\v{q}}$ and satisfying \mbox{$\delta(\hat{\v{p}},\hat{\tilde{\v{p}}})\leq\epsilon$}. Then, by simply choosing \mbox{$\tilde{\hat{\v{p}}}=\hat{\tilde{\v{p}}}$}, we get $\tilde{\hat{\v{p}}}\succ\hat{\v{q}}$ and \mbox{$\delta(\hat{\v{p}},\tilde{\hat{\v{p}}})\leq\epsilon$}, so that \mbox{$\hat{\v{p}}\prescript{}{\epsilon}{\succ}\hat{\v{q}}$}.
	
	Now, in order to prove Eq.~\eqref{eq:lemma6_2} we will construct a particular counterexample. Consider the distributions
	\begin{align}
		\v{p}=[1,0],\quad\v{q}=[1/2,1/2],\quad\v{\gamma}=[3/4,1/4].
	\end{align}
	Their embedded versions are then given by
	\begin{align}
		\hat{\v{p}}=[1/3,1/3,1/3,0],\quad \hat{\v{q}}=[1/6,1/6,1/6,1/2].
	\end{align} 
	We note that \mbox{$\tilde{\hat{\v{p}}}=[1/2,1/4,1/4,0]$} majorises $\hat{\v{q}}$ and that \mbox{$\delta(\hat{\v{p}},\tilde{\hat{\v{p}}})=(3-2\sqrt{2})/6=:\epsilon_0$}. Therefore, \mbox{$\hat{\v{p}}\prescript{}{\epsilon_0}{\succ}\hat{\v{q}}$}. On the other hand, for a general two-dimensional distribution, \mbox{$\tilde{\v{p}}=[\tilde{p},1-\tilde{p}]$}, we have \mbox{$\delta(\v{p},\tilde{\v{p}})=1-\tilde{p}$}, and the smallest $\tilde{p}$ for which $\tilde{\v{p}}\succ^\beta\v{q}$ is equal to $1/2$. This means that the optimal $\epsilon_1$ for which \mbox{$\tilde{\v{p}}\prescript{}{\epsilon_1}{\succ^\beta}\v{q}$} holds is \mbox{$\epsilon_1=1/2>\epsilon_0$}. We thus conclude that \mbox{$\hat{\v{p}}\prescript{}{\epsilon}{\succ}\hat{\v{q}}$} does not imply \mbox{$\v{p}\prescript{}{\epsilon}{\succ^\beta}\v{q}$}.
\end{proof}

The results of this section are collectively presented in Fig.~\ref{fig:approximate_equiv}. We also would like to make a couple of remarks. First, using the equivalence between \mbox{$\hat{\v{p}}\succ_\epsilon\hat{\v{q}}$} and \mbox{$\v{p}\succ^\beta_\epsilon\v{q}$} one can, in principle, calculate the optimal fidelity (equivalently: minimal distance $\delta$) between the final and target state under thermodynamic interconversion. More precisely, given initial distribution $\v{p}$ and final $\v{q}$, the optimal fidelity $F(\v{q},\tilde{\v{q}})$ among $\tilde{\v{q}}$ that are thermomajorised by $\v{p}$ is equal to the optimal fidelity $F(\hat{\v{q}},\tilde{\hat{\v{q}}})$ among $\tilde{\hat{\v{q}}}$ that are majorised by $\hat{\v{p}}$. This in turn, via Lemma~\ref{lemma:fidelity}, is equal to the optimal fidelity $F(\hat{\v{p}},\tilde{\hat{\v{p}}})$ among $\tilde{\hat{\v{p}}}$ that majorise $\hat{\v{q}}$. But such an optimal state $\tilde{\hat{\v{p}}}^\star$ is given by the explicit construction presented in Appendix~\ref{app:lemma_fidelity}, and thus we can directly calculate
\begin{equation}
\max_{\tilde{\v{q}}:~\v{p}\succ^\beta\tilde{\v{q}}} F(\v{q},\tilde{\v{q}})=F(\hat{\v{p}},\hat{\tilde{\v{p}}}^\star).
\end{equation}
In \cref{app:numerics} we discuss applying these very concepts to numerically compare our approximations of the optimal conversion rate to the true optimum for small system sizes. We give examples of such numerics in \cref{fig:numerics2,fig:numerics1}.

\begin{figure}[t!]
	\centering
	\setlength\figheight{7cm}
	\setlength\figwidth{9cm}
	\def\figscale{.9}
	\hspace{-.25cm}\includeTikzz{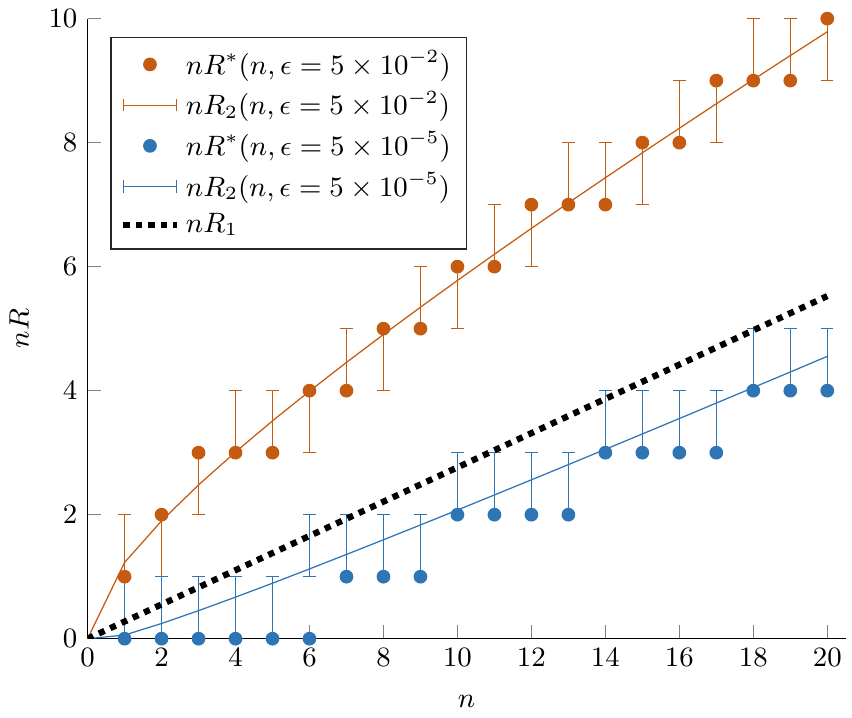}
	\caption{
		Comparison between the second-order approximation $R_2$ and exact thermal interconversion rates $R^*$ for small system sizes, converting from \mbox{$\rho=\frac{7}{10}\ketbra{0}{0}+\frac{3}{10}\ketbra{1}{1}$} to \mbox{$\sigma=\frac{8}{10}\ketbra{0}{0}+\frac{2}{10}\ketbra{1}{1}$}, with Hamiltonian \mbox{$H=\ketbra{1}{1}$} and access to a thermal bath at temperature \mbox{$1/\beta=3$}, as in \cref{fig:numerics2}. The circles indicate number of states produced (c.f.\ \cref{app:numerics}), and the lines those given by the second-order expansion from Eq.~(\hyperref[eq:intro/general]{1}). As the exact number of states produced is always an integer, we have also indicated the rounding of the second-order approximation both up and down with error bars. The colours indicate the infidelity tolerance, with $\epsilon=5\times 10^{-2}$ for red and $\epsilon=10^{-5}$ for blue. The dotted line indicates the number of produced states predicted by the asymptotic interconversion rate $R_1$.
	}
	\label{fig:numerics1}
\end{figure}

Second, we want to point out that Lemmas~\ref{lem:pre_post_equiv}~through~\ref{lem:pre_thermo_equiv} still hold if one applies them to the concept of approximate thermomajorisation based on total variation distance, i.e., if one replaces $\delta(\v{p},\v{q})$ with \mbox{$\frac{1}{2}||\v{p}-\v{q}||$} in Definition~\ref{def:pre_post}. The required modifications of the proofs are rather straightforward (with the exception of Lemma~\ref{lem:pre_post_equiv} which requires some fiddling), and we discuss them in Appendix~\ref{app:total_variation}.

\section{Proofs of the main result}
\label{sec:proof}

We will now present a proof of our main result, \cref{thm:sec-ord thermal}. We will do this by first showing a reduction to special case of bistochastic interconversion, which corresponds to infinite temperature. We recall that as we are considering energy-incoherent initial and target states, $\rho$ and $\sigma$, we only need to consider their eigenvalues, denoted by $\bm p$ and $\bm q$, with the embedded versions of these given by $\bmhat p$ and $\bmhat q$. Also  note that the embedded thermal state $\bmhat \gamma$ simply corresponds to the uniform state $\bm \eta$. We can thus use the equivalence between approximate post-thermomajorisation and embedded majorisation, Lemma~\ref{lem:post_thermo_equiv}, to obtain:
\begin{subequations}
	\begin{align}
	R^*_\beta(n,\epsilon;\bm p,\bm q)&=R^*_0(n,\epsilon;\bmhat p,\bmhat q),\\
	\epsilon^*_\beta(n,R;\bm p,\bm q)&=\epsilon^*_0(n,R;\bmhat p,\bmhat q).
	\end{align} 
\end{subequations}
Conveniently, the second-order expansion only depends on the states through the relative entropy and relative entropy variance, both of which are invariant under embedding, as noted earlier in Eq.~\eqref{eq:embedD}. We can thus solve the problem for embedded distributions and $\beta=0$, and in the final result exchange all $D(\bmhat a||\v{\eta})$ and $V(\bmhat a||\v{\eta})$ with $D(\bm a||\v{\gamma})$ and $V(\bm a||\v{\gamma})$, respectively. For the remainder of this section we will henceforth drop the embedding hats and define $D(\bm a):=D(\bm a\|\bm \eta)$. We also note that \mbox{$V(\bm a\|\bm \eta)=V(\bm a)$}, where 
\begin{align}
&V(\bm a)=\langle (-\log a_i-H(\v{a}))^2\rangle_{\v{a}}
\end{align}
is the entropy variance, and that $V(\bm a)=0$ if and only if $\bm a$ is flat. Specifically, the irreversibility parameter given in Eq.~\eqref{eq:nu} can be expressed as
\begin{align}
	\nu=\frac{V(\bm p)/D(\bm p)}{V(\bm q)/D(\bm q)}.\label{eq:nu_again}
\end{align}

We will split the infinite-temperature proof into four parts, based on whether the relative entropy variances of the initial and target states are non-zero.
\[	
\begin{array}{|c|cc|}
\hline
~ & V(\bm q)=0 & V(\bm q)>0 \\ \hline 
V(\bm p)=0 & \text{Flat-to-flat} & \text{Formation} \\ 
V(\bm p)>0 & \text{Distillation} & \text{Interconversion}\\ \hline
\end{array} 
\]
We start with the case where both initial and target states have zero relative entropy variance. We refer to this as flat-to-flat interconversion, since both the initial and target state are flat. Recalling that states with embedded distributions being flat are proportional to the thermal state on their support, we note that this case contains the conversion between sharp energy states as a special case. We will then consider the cases of distillation and formation, in which either the target or initial states are flat respectively. These are so-named because they contain both the distillation of, and formation from, sharp energy states. Finally we will consider the general interconversion problem, in which neither state is flat. We refer to the non-flat distribution case as general, because it in fact implies the three other results by using the limiting behaviours of the Rayleigh-normal distributions given in Eq.~\eqref{eq:rayleigh_limit}.

\subsection{Central limit theorem}
\label{subsec:clt}

Before we present our proofs, we first formulate the main mathematical tool needed for such a small deviation analysis: a central limit theorem. Specifically, we want to give tail bounds on i.i.d.\ product distributions. Considering the standard central limit theorem, one can derive the following tail bound.
\begin{lem}
	\label{lem:clt-magnitude}
	For any distribution $\bm a$ such that $V(\bm a)>0$,
	\begin{align}
	\lim\limits_{n\to\infty}\sum_{i}\left\lbrace 
	\left(\bm a^{\otimes n}\right)_i		\middle|
	\left(\bm a^{\otimes n}\right)_i\geq 1/k_n(x)		
	\right\rbrace=\Phi\left(x\right),
	\end{align}
	where $k_n(x):=\left\lfloor\exp\left(H(\bm a^{\otimes n})+x\sqrt{V(\bm a^{\otimes n})}\right)\right\rfloor$.
\end{lem}
\noindent For completeness we provide the proof of the above known result in Appendix~\ref{app:lemma_CLT}. We will also rely on an alternate form of central limit theorem. 
\begin{lem}[Lemma~12 of Ref.~\cite{kumagai2017second}]
	\label{lem:clt-index}For any distribution $\bm a$ such that $V(\bm a)>0$,
	\begin{align}
	\lim_{n\to\infty}\sum_{i=1}^{k_n(x)}\left(\bm a^{\otimes n}\right)_i^\downarrow=\Phi(x),
	\end{align}	where $k_n(x):=\left\lfloor\exp\left(H(\bm a^{\otimes n})+x\sqrt{V(\bm a^{\otimes n})}\right)\right\rfloor$.
\end{lem}

We now want to convert the above result into the specific form of a bound we will use. For some rate $R(n)$, we define the \emph{total} initial and target states as
\begin{subequations}
\begin{align}
\bm P^n &:= \bm p^{\otimes n}\otimes\bm \eta ^{\otimes nR(n)},\\
\bm Q^n &:= \bm q^{\otimes nR(n)}\otimes \bm \eta^{\otimes n}.
\end{align}
We note that generally $R$ depends on $n$, but will henceforth omit this explicit dependence.
\end{subequations}
We also introduce a quantity analogous to $k_n(x)$,
\begin{align}
K_n(x):=\left\lfloor \exp\left(H(\bm Q^n)+x\sqrt{V(\bm Q^n)}\right) \right\rfloor,
\end{align}
that will be crucial in all our proofs. The central limit theorem for these distributions is given by the following result.
\begin{lem}[Central limit theorem for $\bm P^n$ and $\bm Q^n$]
	\label{lem:clt}
	If $V(\bm q)>0$ and $R$ is bounded away from zero, then $\bm Q^n$ has the tail bound
	\begin{align}
	\lim\limits_{n\to\infty}\sum_{i=1}^{K_n(x)} Q^{n\downarrow}_i=\Phi(x).
	\end{align}
	Moreover, if we consider a rate of the form
	\begin{align}
	{\label{eq:formation_rate}
		R_\mu(n)=\frac{1}{D(\bm q)}\left[ D(\bm p)+\sqrt{\frac{D(\bm p)}{D(\bm q)}\frac{V(\bm q)}{n}}\,\mu \right],
	}\end{align}
	for some $\mu\in\mathbb{R}$, then $\bm P^n$ also has a corresponding tail bound
	\begin{align}
	\lim\limits_{n\to\infty}\sum_{i=1}^{K_n(x)} P^{n\downarrow}_i= \Phi_{\mu,\nu}(x),
	\end{align}
	with $\nu$ given by Eq.~\eqref{eq:nu_again}.
\end{lem}
\begin{proof}
	We start by noticing that \cref{lem:clt-index} remains true if the product distribution is ``smeared out''. Specifically, for any flat state $\bm f$ we have
	\begin{align}
	\lim_{n\to\infty}\sum_{i=1}^{k_n(x)}\left(\bm a^{\otimes n}\otimes\bm f\right)_i^\downarrow=\Phi(x),
	\end{align}
	where 
	\begin{align}
	\!k_n(x):=\left\lfloor\exp\left(H(\bm a^{\otimes n}\otimes\bm f)+x\sqrt{V(\bm a^{\otimes n}\otimes\bm f)}\right)\right\rfloor\!.\!\!
	\end{align}
	Using this, if we make the substitutions $\bm a^{\otimes n}\leftarrow \bm{q}^{\otimes Rn}$ and $\bm f\leftarrow \bm{\eta}^{\otimes n}$, recalling that $R$ is bounded away from zero, we arrive at
	\begin{align}
	\lim\limits_{n\to\infty}\sum_{i=1}^{K_n(x)} Q^{n\downarrow}_i=\Phi(x).
	\end{align}
	Applying the same argument to $\bm P^n$ gives
	\begin{align}
	\label{eq:limit_Pn}
	\lim\limits_{n\to\infty}\sum_{i=1}^{K^P_n(y)} P^{n\downarrow}_i=\Phi\left(\frac{y}{\sqrt{V(\bm p)}}\right),
	\end{align}	
	where $K_n^P(y):=\left\lfloor\exp\left(H(\bm P^n)+y\sqrt{n}\right)\right\rfloor$. Note that we did not include the variance in the definition of $K_n^P$ (as we did in $K_n$), because we have not assumed $V(\bm p)>0$. Indeed, if we interpret $\Phi(y/V(\bm p))$ as cumulative of the zero-mean Dirac distribution for $V(\bm p)=0$, then this also holds for $V(\bm p)=0$.
	
	We now want to express Eq.~\eqref{eq:limit_Pn} as a summation up to $K_n(x)$ for some $x$. Noticing that \mbox{$R_\mu=D(\bm p)/D(\bm q)+ o(1)$}, our choice of rate $R_\mu$ can be rearranged to give
	\begin{align}
	H(\bm Q^n)\simeq H(\bm P^n)-\sqrt{n\frac{D(\bm p)}{D(\bm q)}V( \bm q)}\mu.
	\end{align}
	Therefore, $K_n^P(y)=K_n(x)$ is equivalent to
	\begin{align}
	\frac y{\sqrt{V(\bm p)}}&=x\sqrt{\frac{V(\bm Q^n)}{V(\bm P^n)}}-\frac \mu{\sqrt \nu}=\frac{x-\mu}{\sqrt{\nu}}+o(1).
	\end{align}
	Finally, using the continuity of $\Phi$ gives the desired tail bound
	\begin{align}
	\lim\limits_{n\to\infty}\sum_{i=1}^{K_n(x)} P^{n\downarrow}_i
	&= \lim\limits_{n\to\infty}\sum_{i=1}^{K_n^P(y)} P^{n\downarrow}_i\nonumber\\
	&= \lim\limits_{n\to\infty}\sum_{i=1}^{K_n^P\left(\frac{x-\mu}{\sqrt{\nu}}\right)} P^{n\downarrow}_i\nonumber\\
	&=\Phi\left(\frac{x-\mu}{\sqrt{\nu}}\right)=\Phi_{\mu,\nu}(x).
	\end{align}
	Recalling that $\nu$ is proportional to $V(\bm p)$, we note that all of the above expressions are still well-defined if $V(\bm p)=0$, where we understand $\Phi_{\mu,0}$ to be the cumulative of the Dirac distribution with mean $\mu$.
\end{proof}

\subsection{Bistochastic flat-to-flat}
\label{subsec:trivial}

We start with the case of flat-to-flat conversion, where both the initial and target states are flat. In this boundary case, which was not considered in Ref.~\cite{kumagai2017second}, we can provide an exact single-shot expression for the optimal rate.
\begin{prop}[Bistochastic flat-to-flat]
	\label{prop:trivial}
	For any initial state $\bm p$ and target state $\bm q$ such that $V(\bm p)=V(\bm q)=0$, and infidelity \mbox{$\epsilon\in[0,1)$}, the optimal interconversion rate is given by
	\begin{align}
		R^*_0(n,\epsilon)=\frac1n \left\lfloor \frac{nD(\bm p)-\log(1-\epsilon)}{D(\bm q)} \right\rfloor\simeq \frac{D(\bm p)}{D(\bm q)}.
	\end{align}
\end{prop}

Before proving this, we first consider the optimal majorising distribution in the case where all distributions involved are flat.

\begin{lem}[Single-shot flat-to-flat]
	\label{lem:1shot-trivial}
	If we let $\bm a$ and $\bm b$ be distributions such that $V(\bm a)=V(\bm b)=0$, then
	\begin{align}
	\label{lem:1shot-trivial_eq}	
	\min\left\lbrace \epsilon \middle| \bm a\succ_\epsilon \bm b \right\rbrace =\max\left\lbrace1-{\frac{\exp H(\bm b)}{\exp H(\bm a)}},0\right\rbrace.
	\end{align}
\end{lem}
\begin{proof}
	First, from Lemma~\ref{lem:pre_post_equiv} we know that 
	\begin{equation}
	\min\left\lbrace \epsilon \middle| \bm a\succ_\epsilon \bm b \right\rbrace =\min\left\lbrace \epsilon \middle| \bm a\prescript{}{\epsilon}{\succ} \bm b \right\rbrace.
	\end{equation}
	Now, the right hand side of the above equation is minimised by a state $\tilde{\v{a}}^\star$, whose explicit construction (found in Ref.~\cite{vidal2000approximate}) we present in Appendix~\ref{app:lemma_fidelity} while proving Lemma~\ref{lemma:fidelity}. Using it one finds that
	\begin{align}
	\tilde{\v{a}}^{\star\downarrow}=\begin{dcases}
	\bm a^\downarrow &\text{if}~~\bm a\succ \bm b,\\
	\bm b^\downarrow &\text{if}~~\bm a\nsucc \bm b.
	\end{dcases}
	\end{align}
	Since $\v{a}$ and $\v{b}$ are flat, in the first case we have
	\begin{equation}
	F(\tilde{\v{a}}^{\star\downarrow},\v{b}^\downarrow)={\frac{\exp H(\bm b)}{\exp H(\bm a)}},
	\end{equation}
	which leads to Eq.~\eqref{lem:1shot-trivial_eq}.
\end{proof}

We can now apply this to give an exact expression for the optimal rate.

\begin{proof}[Proof of \cref{prop:trivial}]
	Applying \cref{lem:1shot-trivial} to \mbox{$\v{a}=\v{P}^{n}$} and \mbox{$\v{b}=\v{Q}^{n}$} we find that $\epsilon_0^*(n,R)$ vanishes for any \mbox{$R\leq D(\bm p)/D(\bm q)$}, and
	\begin{align}
	\frac{1}{n}\log\bigl(1-\epsilon_0^*(n,R)\bigr)=D(\bm p)-RD(\bm q)
	\end{align}
	for any $R\geq D(\bm p)/D(\bm q)$. Converting the expression for the optimal infidelity into an expression for the optimal rate, and recalling that $nR^*(n,\epsilon)$ must be an integer, gives \cref{prop:trivial} as required.
\end{proof}

\subsection{Bistochastic distillation}
\label{subsec:distill}

We now consider distillation, in which the target state is flat. 
\begin{prop}[Bistochastic distillation]
	\label{prop:distill}
	For any initial state $\bm p$ and target state $\bm q$ such that $V(\bm q)=0$, and infidelity \mbox{$\epsilon\in(0,1)$}, the optimal interconversion rate has the second-order expansion
	\begin{align}
	R^*_0(n,\epsilon)
	\simeq
	\frac{1}{D(\bm q)}\Biggl[D(\bm p)
	+\sqrt{\frac{V(\bm p)}{n}}\Phi^{-1}(\epsilon)
	\Biggr].
	\end{align}
\end{prop}

Instead of utilising the techniques of Ref.~\cite{kumagai2017second}, we will, similar to the flat-to-flat case, prove \cref{prop:distill} by first considering a single-shot expression for the optimal error.

\begin{lem}[Single-shot distillation]
	\label{lem:1shot-distill}
	Let $\bm a$ and $\bm b$ be distributions with $V(\bm b)=0$. Then
	\begin{align}
	\min\left\lbrace \epsilon \middle| \bm a\succ_\epsilon \bm b \right\rbrace = \sum_{i>\exp H(\bm b)}a_i^{\downarrow}.
	\end{align}
\end{lem}
\begin{proof}
		First, from Lemma~\ref{lem:pre_post_equiv} we know that 
	\begin{equation}
	\min\left\lbrace \epsilon \middle| \bm a\succ_\epsilon \bm b \right\rbrace =\min\left\lbrace \epsilon \middle| \bm a\prescript{}{\epsilon}{\succ} \bm b \right\rbrace.
	\end{equation}
	Now, to find a state minimising the right hand side of the above equation, consider a distribution $\bmtilde a$ such that $\bmtilde a\succ \bm b$. Since $\v{b}$ is flat, this is equivalent to the statement that $\bmtilde a$ has a support which is no larger than that of $\v{b}$. This condition is clearly necessary; it is sufficient as any distribution with $d$ or fewer non-zero entries majorises the flat distribution over $d$ entries. Using the Schwarz inequality one can then show that the distribution $\bmtilde a$ which contains at most $\exp H(\bm b)$ non-zero elements, and is closest to $\bm a$, is simply the truncated-and-rescaled distribution,
	\begin{align}
	\tilde a_i^\downarrow:=\frac{1}{\sum_{j=1}^{\exp (H(\bm b))}a^\downarrow_j}\begin{dcases}
	a^\downarrow_i  &\text{if}~~i\leq \exp H(\bm b),\\
	0 &\text{if}~~i>\exp H(\bm b).
	\end{dcases}
	\end{align}
	It is a straightforward calculation to show that the infidelity of such a smoothed state is given by mass of the truncated tail
	\begin{align}
	\delta(\bm a,\bmtilde a)=\sum_{i>\exp H(\bm b)}a_i^{\downarrow}.
	\end{align}
\end{proof}

We can now use \cref{lem:clt} to bound this tail, giving a second-order expansion for asymptotic case.

\begin{proof}[Proof of \cref{prop:distill}]
	Applying \cref{lem:1shot-distill} to $\bm a=\bm P^n$ and $\bm b=\bm Q^n$ we find that
	\begin{align}
	\epsilon_0^*(n,R)=\sum\limits_{i>\exp H(\bm Q^n)}P^{n\downarrow}_i.
	\end{align}
	Now consider a rate of the form
	\begin{align}
	r_\mu(n):=\frac{1}{D(\bm q)}\left[D(\bm p)+\sqrt{\frac{V(\bm p)}n}\mu\right],
	\end{align}
	for some $\mu\in\mathbb R$. We then have
	\begin{align}
	H(\bm Q^n)=H(\bm P^n)-\mu\sqrt{V(\bm P^n)},
	\end{align}
	and so if we apply the first bound of \cref{lem:clt} (with $\bm P^n$ in place of $\bm Q^n$), we arrive at
		\begin{align}
	\lim\limits_{n\to\infty}\epsilon_0^*(n,r_\mu)
	&=1-\Phi(-\mu)=\Phi(\mu).
	\end{align}
	Reversing the relationship between infidelity and rate, this implies 
	\begin{align}
	R^*_0(n,\epsilon)\simeq\frac{1}{D(\bm q)}\left[D(\bm p)+\sqrt{\frac{V(\bm p)}n}\Phi^{-1}(\epsilon)\right],
	\end{align}
	as required.
\end{proof}

\subsection{Bistochastic formation}
\label{subsec:form}

For our final special case, we consider flat initial states.

\begin{prop}[Bistochastic formation]
	\label{prop:form}
	For any initial state $\bm p$ and target state $\bm q$ such that $V(\bm p)=0$, and infidelity \mbox{$\epsilon\in(0,1)$}, the optimal interconversion rate has the second-order expansion
	\begin{align}
	R^*_0(n,\epsilon)
	\simeq
	\frac{1}{D(\bm q)}\Biggl[D(\bm p)
	+\sqrt{\frac{D(\bm p)}{D(\bm q)}\frac{V(\bm q)}{n}}\Phi^{-1}(\epsilon)
	\Biggr].
	\end{align}
\end{prop}

This proof is more involved than distillation, and will involve some of the techniques used in the proof of the general interconversion problem, as first developed in Ref.~\cite{kumagai2017second}. As with distillation, we will attempt to bound the rate by bounding the optimal infidelity between the total initial state $\bm P^n$ and a state $\bmtilde P^n$ that majorises the total target state $\bm Q^n$. We will thus fix our rate $R_\mu(n)$ to be given by Eq.~\eqref{eq:formation_rate} from \cref{lem:clt}, and look for bounds on infidelity between $\bmtilde P^n$ and $\bm P^n$. More precisely, our proof is split into two parts: achieveability (upper bound on optimal error/lower bound on optimal rate) and optimality (lower bound on optimal error/upper bound on optimal rate).

\subsubsection{Achieveability}

\paragraph*{Sketch of construction.}

The general idea here is to construct a distribution $\bmtilde{P}^n$ which is close to the total initial state $\bm P^n$ \emph{and} majorises total target state $\bm Q^n$. {By the equivalence of pre- and post-majorisation, Lemma~\ref{lem:pre_post_equiv}, this will prove that there exists a distribution $\bmtilde{Q}^n$ that is majorised by the total initial state $\bm P^n$ and is close to the total target state $\bm Q^n$.} We will start by defining two bins (sets) of indices, $B$ and $B'$. We will then construct a \emph{scaled distribution} $\bm S^n$ such that for indices belonging to $B$ it has the same shape as $\bm P^n$ (i.e., it is flat), but has as much mass as $\bm Q^n$ has over the indices belonging to $B'$. The first property will guarantee that $\bm S^n$ lies close to $\bm P^n$, and the second that it lies close to a \emph{majorising distribution} $\bmtilde P^n\succ \bm Q^n$. We will then analyse $\delta(\bm P^n,\bmtilde{P}^n)$, giving an upper bound on the optimal error.

\paragraph*{Binning.} 

For some small $\zeta>0$, define two bins of indices
\begin{subequations}
\begin{align}
B&:=\lbrace 1,\dots,K_n(\mu+\zeta/2)\rbrace,\\
B'&:=\lbrace K_n(\mu+\zeta),\dots,\infty\rbrace,
\end{align}
\end{subequations}
where $\infty$ is shorthand for the largest index, and $\mu$ fixes the value of rate $R_\mu(n)$. We will consider $B$ as a bin on the indices of $\bm P^n$ and $B'$ on those of $\bm Q^n$. We denote the complements of these bins as $\bar B$ and $\bar{B'}$, respectively.

\paragraph*{Scaled distribution $\bm S^n$.}

For any $j\in B$ define 
\begin{align}
S^n_j:=\frac{\sum_{k\in B'}Q^{n\downarrow}_k}{\sum_{k\in B}P^{n\downarrow}_k}\cdot P_j^{n\downarrow}=\frac{1}{\abs{B}}\sum_{k\in B'}Q^{n\downarrow}_k,
\end{align}
and $S^n_l:=1-\sum_{k\in B'} Q^{n\downarrow}_k$ for some arbitrary $l\notin B$ such that $\bm S^n$ is normalised. 

By construction the mass of $\bm S^n$ on $B$ is equal to that of $\bm Q^{n}$ on $B'$, i.e.\
\begin{align}
\sum_{j\in B}S^n_j:=\sum_{j\in B'}Q^{n\downarrow}_k.
\end{align}
We now want to show that this implies the existence of a nearby majorising distribution $\bmtilde P^n\succ \bm Q^n$.

\paragraph*{Majorising distribution $\bmtilde P^n$.}

Instead of giving an explicit construction of $\bmtilde P^n$, we instead present an existence proof. Specifically we will leverage the following lemma:
\begin{lem}
	\label{lem:determinstic map}
	For non-negative vectors $\bm a$ and $\bm b$ such that $\sum_i a_i=\sum_i b_i$, there exists a vector $\bmtilde a$ such that \mbox{$\sum_k \tilde a_k=\sum_k  a_k$}, $\bmtilde a\succ \bm b$ and $\norm{\bmtilde a-\bm a}_\infty \leq \norm{\bm b}_\infty$.
\end{lem}
\begin{proof}
	Consider a function over indices, $f:\mathbb N\to\mathbb N$, and the vector $\bmtilde{a}$ given by the follow action of $f$ on $\bm b$,
	\begin{align}
	\tilde a_i: =\sum_j\left\lbrace b_j \middle| f(j)=i \right\rbrace.
	\end{align}
	Clearly such a mapping can only concentrate a distribution, and so $\bmtilde a\succ \bm b$. Now, among $\bmtilde a$ of the above form we choose that which is closest to $\bm a$ in {$l_\infty$-norm}. Let $i$ be an index at which the {$l_{\infty}$-norm }of $\bmtilde{a}-\bm a$, denoted by $\Delta$, is achieved,
	\begin{align}
	i\in \mathop{\mathrm{arg\,max}}_j\abs{\tilde a_j-a_j}.
	\end{align}
	We are going to assume that \mbox{$\Delta >\norm{\bm b}_\infty$} and show that this would imply that $\bmtilde a$ cannot be optimal, proving $\norm{\bmtilde a-\bm a}_\infty\leq\norm{\bm b}_\infty$ by way of contradiction. 
	
	There are two cases to consider: $\tilde a_i>a_i$ and $\tilde a_i<a_i$. We start with $\tilde a_i>a_i$. As $\tilde a_i>0$, there must exist some $\alpha$ such that $f(\alpha)=i$ and $b_\alpha> 0$. As $\bm a\neq \bmtilde a$ and \mbox{$\sum_k a_k=\sum_k \tilde a_k$}, there must exist a $k$ such that {$\tilde a_k<a_k$}. Consider changing the map to from $f(\alpha)=i$ to $f(\alpha)=k\neq i$. This has the effect of lowering $\tilde{a}_i$ by $b_\alpha$ and raising $\tilde a_k$ by the same amount. Given that $b_\alpha<\Delta$ by assumption, this means that $\tilde a_i-a_i$ changes from $\Delta$ to $(0,\Delta)$, and $\tilde a_k-a_k$ changes from $[-\Delta,0)$ to $(-\Delta,\Delta)$. As such, $\abs{\tilde a_i-a_i}$ and $\abs{\tilde a_j-a_j}$ are now both strictly smaller than $\Delta$. Since all other entries of $\bmtilde{a}$ are unchanged, we have reduced the number of indices $j$ such that $\abs{\tilde a_j-a_j}=\Delta$ by at least one. Similarly for the case $\tilde a_i<a_i$ one can make an analogous argument by changing $f(\alpha)=k\neq i$ to $f(\alpha)=k$ for some $k$ such that $\tilde a_k> a_k$. Iterating this we can keep decreasing the number of indices at which the norm is achieved, eventually giving us $\abs{\tilde a_j-a_j}<\Delta$ for all $j$, i.e.\ $\abs{\bmtilde{a}-\bm{a}}<\Delta$. This shows that the original choice of $\bmtilde a$ was not optimal as assumed, proving $\Delta \leq \norm{\bm b}_\infty$ by contradiction.
\end{proof}

With the use of the above lemma we can now get the desired majorising distribution $\bmtilde P^n\succ \bm Q^n$.

\begin{lem}
	\label{lem:exist-ptilde-1}
	There exists a distribution $\bmtilde P^n$ such that $\bmtilde P^n\succ \bm Q^n$ and
	\begin{align}
	\abs{\tilde P^n_j-S_j^n}\leq 1/K_n(\mu+\zeta)	
	\end{align}
	for all $j\in B$.
\end{lem}
\begin{proof}
	The idea here is to apply \cref{lem:determinstic map} to the restriction of each distribution to its corresponding bin. Specifically if we take $\bm{a}:=\left.\bm S^n\right|_B$ and $\bm{b}:=\left.\bm Q^{n\downarrow}\right|_{B'}$, then \cref{lem:determinstic map} gives us a vector $\bmtilde a$ such that $\bmtilde a\succ \left.\bm Q^{n\downarrow}\right|_{B'}$ and 
	\begin{align}
	\norm{\bmtilde a-\bm a}_\infty 
	\leq \norm{\bm b}_\infty&=\norm{\left.\bm Q^{n\downarrow}\right|_{B'}}_\infty=\max_{j\in B'} Q^{n\downarrow}_j\nonumber\\
	&=Q^{n\downarrow}_{K_n(\mu+\zeta)}\leq 1/K_n(\mu+\zeta),
	\end{align}
	where the final inequality follows from normalisation of $\bm Q^n$. We now define our majorising distribution within bin $B$ as $ \bmtilde P^n\big|_B:=\bmtilde{a}$, so that $|\tilde P^n_j-S^n_j|\leq1/K_n(\mu+\zeta)$ for any $j\in B$ as desired; and, once again, in order to normalise the distribution we also define $\tilde P^n_l=1-\sum_i \tilde a_i$ for some arbitrary $l\notin B$. 
	
	The fact that $\bmtilde{a}$ majorises the restriction of $\bm Q^{n\downarrow}$ to $B'$, together with the sharpness of $\bmtilde P^n$ outside of $B$, gives
	\begin{align}
	\left.\bmtilde P^n\right|_B\succ \left.\bm Q^{n\downarrow}\right|_{B'}
	\quad\text{and}\quad
	\left.\bmtilde P^n\right|_{\bar B}\succ \left.\bm Q^{n\downarrow}\right|_{\bar B'}.
	\end{align}
	Next, we note that majorisation spreads over direct sum, i.e.\ $\bm \alpha_1\succ \bm \beta_1$ and $\bm \alpha_2\succ \bm \beta_2$ implies $\bm \alpha_1\oplus \bm \alpha_2\succ \bm \beta_1\oplus \bm \beta_2$. This can be seen by using Theorem~\ref{thm:major} (i.e., the equivalence of majorisation relation between two distribution with the existence of a bistochastic map between them), and noticing that bistochasticity is preserved under direct sum. Applying this to $\bmtilde P^n$ gives the desired majorisation property:
	$\bmtilde P^n \succ \bmtilde P^{n}|_{B} \oplus \bmtilde P^{n}|_{\overline B}
	\succ \bm Q^{n\downarrow}|_{B'} \oplus \bm Q^{n\downarrow}|_{\overline {B'}}\succ \bm Q^n$.
\end{proof}

\paragraph*{Infidelity.}

By now we have proven the existence of a majorising distribution $\bmtilde P^n\succ\bm Q^n$ and bounded its distance from the scaled distribution $\bm S^n$ on the restriction to $B$. The final step involves bounding the infidelity $\delta(\bmtilde P^n,\bm P^n)$. To achieve this we will first show that the closeness of $\bm S^n$ and $\bmtilde P^n$ on $B$, as given in \cref{lem:exist-ptilde-1}, allows us to bound $F(\bmtilde P^n,\bm P^n)$ in terms of $\bm S^n$.

\begin{lem}
	\label{lem:partialfid}
	Asymptotically, the fidelity between the majorising distribution $\bmtilde P^n$ and the total initial distribution $\bm P^n$ can be bounded as follows
	\begin{align}
	\liminf\limits_{n\to\infty}{F\left(\bmtilde{P}^n,\bm P^n\right)}\geq \liminf\limits_{n\to\infty}\left(\sum_{j\in B}\sqrt{S^n_jP^n_j}\right)^2\!\!.
	\end{align}
\end{lem}
\begin{proof}
	First, we apply \cref{lem:exist-ptilde-1} and the fact that $\sqrt{x-y}\geq \sqrt{x}-\sqrt{y}$ for all $x\geq y\geq 0$ to break the fidelity into the desired expression and an error term
	\begin{align}
	\sqrt{F(\bmtilde{P}^n,\bm P^{n\downarrow})}
	&\geq \sum_{j\in B}\sqrt{\tilde P^n_jP^{n\downarrow}_j}\nonumber\\
	&\hspace{-1.25cm}\geq \sum_{j\in B}\sqrt{\max\lbrace S^n_j-1/K_n(\mu+\zeta),0\rbrace P^{n\downarrow}_j}\nonumber\\
	&\hspace{-1.25cm}\geq \sum_{j\in B}\left[\sqrt{ S^n_jP^{n\downarrow}_j}-\sqrt{P^{n\downarrow}_j/K_n(\mu+\zeta)}\right].
	\label{eqn:partialfid}
	\end{align}
	We can express the second error term as
	\begin{align}
	\sum_{j\in B}\sqrt{\frac{P^{n\downarrow}_j}{K_n(\mu+\zeta)}}
	&= \sqrt{\frac{\abs{B}}{K_n(\mu+\zeta)}}\nonumber\\
	&= \sqrt{\frac{K_n(\mu+\zeta/2)}{K_n(\mu+\zeta)}}.
	\end{align}
	Given that $\zeta>0$ is a constant and $V(\bm q)>0$, we have that \mbox{$K_n(\mu+\zeta/2)/K_n(\mu+\zeta)$} is decaying exponentially as $n\to\infty$. Taking the limit inferior of Eq.~\eqref{eqn:partialfid} therefore gives the required bound.
\end{proof}

Using the above result on fidelity, we can now prove achieveability.

\begin{proof}[Proof of \cref{prop:form} (achieveability)]
	Substituting the definition of $\bm S^n$ into \cref{lem:partialfid} gives
	\begin{align}{
		\liminf\limits_{n\to\infty}\!{F\left(\bmtilde{P}^n,\bm P^n\right)}
		&\!\geq\! \liminf\limits_{n\to\infty}\!{\sum_{i\in B'}\!Q^{n\downarrow}_i\sum_{j\in B}\!P^{n\downarrow}_j}.
	}\end{align}
	By applying \cref{lem:clt} we then obtain 
	\begin{subequations}
	\begin{align}
	\sum_{i\in B}P^{n\downarrow}_i&\xrightarrow{n\to\infty} \Phi_{\mu,0}(\mu+{\zeta/2})=1,\\
	\sum_{i\in B'}Q^{n\downarrow}_i&\xrightarrow{n\to\infty} 1-\Phi(\mu+{\zeta}),
	\end{align}
	\end{subequations}
	and therefore
	\begin{align}
	\limsup\limits_{n\to\infty}\delta\left(\bmtilde{P}^n,\bm P^n\right)\leq \Phi(\mu+{\zeta}).
	\end{align}
	Due to the equivalence between pre- and post-majorisation, Lemma~\ref{lem:pre_post_equiv}, the above means that there exists a distribution $\bmtilde{Q}^n$ that is majorised by the total initial state $\bm{P}^n$ and such that
	\begin{align}
	\limsup\limits_{n\to\infty}\delta\left(\bmtilde{Q}^n,\bm Q^n\right)\leq \Phi(\mu+{\zeta}).
	\end{align}
	As this is true for any $\zeta>0$ we can take $\zeta\searrow 0$ and conclude that the optimal infidelity is upper bounded
	\begin{align}
	\limsup\limits_{n\to\infty}\epsilon_0^*(n,R_\mu)\leq \Phi(\mu),
	\end{align}
	which implies a corresponding lower bound on the optimal rate
	\begin{align}
	R^*_0(n,\epsilon)
	\gtrsim
	\frac{1}{D(\bm q)}\Biggl[D(\bm p)
	+\sqrt{\frac{D(\bm p)}{D(\bm q)}\frac{V(\bm q)}{n}}\Phi^{-1}(\epsilon)
	\Biggr],
	\end{align}
	as required.
\end{proof}

\subsubsection{Optimality}
\label{subsubsec:form optimality}

We now turn our attention to a corresponding second-order upper bound on the optimal rate. We will make use of the following lemma from Ref.~\cite{kumagai2017second}, which upper bounds the fidelity between a flat state and a majorising distribution.

\begin{lem}[Lemma~6 of Ref.~\cite{kumagai2017second}]
	\label{lem:bipartite}
	Let $\bm a$, $\bmtilde{a}$ and $\bm b$ be any distributions such that $V(\bm a)=0$ and $\bmtilde{a}\succ \bm b$. Also, let $M\leq \exp H(\bm a)$. Then 
	\begin{align}
	\!\!\sqrt{F(\bmtilde a,\bm a)}
	&\leq \sqrt{\sum_{i\leq N}a^{\downarrow}_i\sum_{j\leq N}b_j^\downarrow}+\sqrt{\sum_{i>N}a^{\downarrow}_i \sum_{j>N}b_j^\downarrow},
	\end{align}
	where $N:=\abs{\left\lbrace i \middle| b_i\geq 1/M \right\rbrace}$.
\end{lem}

We can now use the above lemma to obtain an optimality bound which matches that given for achieveability.

\begin{proof}[Proof of \cref{prop:form} (optimality)]
	Consider any distribution $\bmtilde{P}^n\succ \bm Q^n$. Now choose $\bmtilde{a}=\bmtilde{P}^n$, $\bm{a}=\bm{P}^n$ and $\v{b}=\bm Q^n$. Also, notice that $M:=K_n(\mu-\zeta)$ satisfies $M\leq \exp H(\v{a})$. Hence, we can apply \cref{lem:bipartite} to upper bound the fidelity,
	\begin{align}
	\sqrt{F(\bmtilde{P}^n,\bm P^n)}&\leq \sqrt{\sum_{i\leq N}P^{n\downarrow}_i\sum_{j\leq N}Q_j^{n\downarrow}}
	\nonumber\\	\label{eq:fidelity_formation_opt}
	&\qquad
	+\sqrt{\sum_{i>N}P^{n\downarrow}_i \sum_{j>N}Q_j^{n\downarrow}},
	\end{align}
	where $N:=\abs{\lbrace i|Q^n_i\geq 1/K_n(\mu-\zeta)\rbrace}$. By the standard central limit theorem \cref{lem:clt-magnitude}, we have
	\begin{align}
	\!\!\sum_{i=1}^{N}Q^{n\downarrow}_i
	&\!=\!\sum_{i}\left\lbrace Q^{n\vphantom\downarrow}_i \middle| Q^{n}_i\!\geq\! \frac{1}{K_n(\mu-\zeta)} \right\rbrace \!\to\!\Phi(\mu-\zeta).\!
	\end{align}
	The normalisation of $\bm Q^n$ gives us that $N\leq K_n(\mu-\zeta)$, and so we can apply \cref{lem:clt} to obtain
	\begin{align}
	\sum_{i=1}^{N}P^{n\downarrow}_i&\leq \sum_{i=1}^{K_n(\mu-\zeta)}P^{n\downarrow}_i\to \Phi_{\mu,0}(\mu-\zeta)=0.
	\end{align}
	Applying these limits to Eq.~\eqref{eq:fidelity_formation_opt} yields
	\begin{align}
	\liminf\limits_{n\to\infty}\delta\left(\bmtilde{P}^n,\bm P^n\right)\geq \Phi(\mu-\zeta)
	\end{align}
	for any $\bmtilde{P}^n\succ \bm Q^n$. Due to the equivalence between pre- and post-majorisation (c.f.\ Lemma~\ref{lem:pre_post_equiv}) the above means that for any distribution $\bmtilde{Q}^n$ that is majorised by the total initial state $\bm{P}^n$ we have
	\begin{align}
		\liminf\limits_{n\to\infty}\delta\left(\bmtilde{Q}^n,\bm Q^n\right)\geq \Phi(\mu-\zeta)
	\end{align}
	Taking $\zeta\searrow0$ this gives a lower bound on the optimal infidelity
	\begin{align}
	\liminf\limits_{n\to\infty}\epsilon_0^*(n,R_\mu)\geq \Phi(\mu),
	\end{align}
	which implies a corresponding upper bound on the optimal rate 
	\begin{align}
		R^*_0(n,\epsilon)
		\lesssim
		\frac{1}{D(\bm q)}\Biggl[D(\bm p)
		+\sqrt{\frac{D(\bm p)}{D(\bm q)}\frac{V(\bm q)}{n}}\Phi^{-1}(\epsilon)
		\Biggr].
	\end{align}
\end{proof}

\subsection{Bistochastic interconversion}
\label{subsec:gen}

Finally we turn to the general case in which neither relative entropy variance is vanishing.
\begin{prop}[Bistochastic interconversion]
	\label{prop:gen}
	For any initial state $\bm p$ and target state $\bm q$ such that $V(\bm p),V(\bm q)>0$, and infidelity $\epsilon\in(0,1)$, the optimal interconversion rate has the second-order expansion
	\begin{align}{
		R^*_0(n,\epsilon)
		\simeq
		\frac{1}{D(\bm q)}\Biggl[D(\bm p)+\sqrt{\frac{V(\bm p)}{n}}Z^{-1}_{1/\nu}(\epsilon)
		\Biggr] ,
	}\end{align}
	where $\nu$ is given in \ref{eq:nu_again}.
\end{prop}

The proof is similar to that of formation, so it will also utilise many of the ideas inspired by Ref.~\cite{kumagai2017second}. The main complication is that for the general interconversion problem the binning of indices is more elaborate: we now have two \emph{sets} of bins instead of two \emph{individual} bins, and we need to introduce a function $A$ which controls the relative placement of these bins. Once again we will break the proof into both achieveability and optimality bounds.

\subsubsection{Achieveability}

\paragraph*{Sketch of proof.} 

As with formation, the idea here will be to give an explicit construction of $\bmtilde{P}^n\succ \bm Q^n$ which is close to $\bm P^n$. Again, due to the equivalence of pre- and post-majorisation, this will prove that there exists a distribution $\bmtilde{Q}^n$ that is majorised by the total initial state $\bm P^n$ and is close to the total target state $\bm Q^n$. We will start by introducing two sets of bins for each distribution. Using these bins, we will once again construct a \emph{scaled distribution} $\bm S^n$ which reflects the fine-grained features of $\bm P^n$ (same shape within corresponding bins) and coarse-grained features of $\bm Q^n$ (same mass within corresponding bins). We will then show that $\bm S^n$ necessarily lies close to a \emph{majorising distribution} $\bmtilde{P}^n$. Finally we will analyse the infidelity of this distribution with respect to the total initial distribution $\bm{P}^n$ and, by taking the appropriate limits of the parameters in our construction, prove the desired achieveabilty bound of \cref{prop:gen}.

In Section~\ref{subsec:form} our construction was parameterised by a single slack parameter $\zeta>0$. Here, we will have three parameters: $\lambda>0$, $I\in \mathbb{N}$, and a monotone continuously differentiable function $1\geq A\geq \Phi$ pointwise. The parameter $\lambda$ will control the width of our bins, $I$ the number of bins, and $A$ the relative placements of the two sets of bins.

\paragraph*{Binning.}

For $-I\leq i< I$ we define our two sets of bins as
\begin{subequations}
\begin{align}
B_i&:=\lbrace K_n(x_{{i}}),\dots,K_n(x_{i+1})-1\rbrace,\\
B_i'&:=\lbrace K_n(y_{{i}}),\dots,K_n(y_{i+1})-1\rbrace,
\end{align}
\end{subequations}
where the two sequences are defined by
\begin{align}
x_i:=\lambda\frac{i-1}{I}\quad\text{and}\quad
y_i:=\Phi^{-1}\left(A\left(\lambda\frac{i+1}{I}\right)\right)
\end{align}
for $-I\leq i\leq I$. We will consider $B_i$ as bins on the indices of $\bm P^n$ and $B_i'$ on those of $\bm Q^n$. We note that $A\geq \Phi$ implies $y_{i}\geq x_{i+1}+\lambda/I$, resulting in $B_i$ being gapped away from $B_i'$, i.e., all indices belonging to $B_i'$ are much larger than those belonging to $B_i$. This choice plays a role analogous to that of the slack parameter $\zeta$ for the bins in \cref{subsec:form}. For convenience we also define 
\begin{align}
B:=\bigcup_{i=-I}^{I-1}B_i\quad\text{and}\quad
B':=\bigcup_{i=-I}^{I-1}B_i'
\end{align}
to be the union of bins, and $\overline{B}$ and $\overline{B'}$ to be the corresponding complements.

\paragraph*{Scaled distribution $\bm S^n$.}

As in \cref{subsec:form}, we now define $\bm S^n$ in each bin $B_i$ to have the shape of $\bm P^n$ within $B_i$, but the mass of $\bm Q^n$ within $B_i'$. As the bins $B_i$ are all disjoint, for any $j\in B$ there exists a unique $-I\leq i<I$ such that $j\in B_i$. For such indices we define $\bm S^n$ as
\begin{align}
S^n_j:=\frac{\sum_{k\in B_{i}'} Q^{n\downarrow}_k}{\sum_{k\in B_i} P^{n\downarrow}_k}\cdot P^{n\downarrow}_j.
\end{align}
We normalise $\bm S^n$ by taking $S^n_l:=1-\sum_{j\in B}S^n_j$ for some arbitrary $l\notin B$.

\paragraph*{Majorising distribution $\bmtilde P^n$.}

We now want to prove a result analogous to \cref{lem:exist-ptilde-1}: the existence of a distribution that simultaneously majorises the total target distribution $\bm Q^n$ and is close to $\bm S^n$ (within each bin).

\begin{lem}[Existence of a majorising distribution]
	\label{lem:exist-ptilde-2}
	There exists a distribution $\bmtilde P^n$ such that $\bmtilde P^n\succ \bm Q^n$ and
	\begin{align}
	\abs{\tilde P_j^n-S_j^n}\leq 1/K_n(y_{i})
	\end{align}
	for all $j\in B_i$ and $-I\leq i<I$.
\end{lem}
\begin{proof}
	The proof is analogous to that of \cref{lem:exist-ptilde-1}, with an application of \cref{lem:determinstic map} for each pair $\bm S^n|_{B_i}$ and $\bm Q^n|_{B_i'}$, with $-I\leq i<I$. This gives us $\bmtilde{P}^n$ such that for all $-I\leq i<I$ and $j\in B_i$ it is close to $\bm S^n$
	\begin{align}
	\abs{\tilde P_j^n-S_j^n}
	\leq \max_{j\in B_i'}Q_j^{n\downarrow}\leq 1/K_n(y_i),
	\end{align}
	and possesses the majorisation properties
	\begin{align}
	\bmtilde P^n|_{B_i}&\succ \bm Q^{n\downarrow}|_{B_i'}\quad\text{and}\quad
	\bmtilde P^n|_{\bar B}\succ \bm Q^{n\downarrow}|_{\bar{B'}}.
	\end{align}
	Splitting the majorisation across the direct sum, as explained in the proof of \cref{lem:exist-ptilde-1}, gives us the desired overall majorisation
	\begin{align}
	\bmtilde P^n
	&\succ \bigoplus_{i=-I}^{I-1}\bmtilde P^{n}|_{B_i} \oplus \bmtilde P^n|_{\overline B}\nonumber\\
	&\succ \bigoplus_{i=-I}^{I-1}\bm Q^{n\downarrow}|_{B_i'} \oplus \bm Q^{n\downarrow}|_{\overline {B'}}\succ \bm Q^n.
	\end{align}	
\end{proof}

\paragraph*{Infidelity.}

The next step involves bounding the fidelity between the total initial state $\bm P^n$ and majorising distribution $\bmtilde P^n$ given by the above construction. We will start by bounding the fidelity for a fixed set of parameters $A$, $\lambda$ and $I$. 

\begin{lem}
	\label{lem:achieve-finite}
	For any monotone continuously differentiable function \mbox{$1\geq A\geq \Phi$}, $\lambda\geq0$, and $I\in\mathbb{N}$ there exists a sequence of distributions $\bmtilde P^n\succ \bm Q^n$ such that
	\begin{align}
	&\liminf\limits_{n\to \infty} {F\left(\bmtilde{P}^n,\bm{P}^{n\downarrow} \right)}\nonumber\\
	&\qquad\geq\left(\int_{-\lambda}^{\lambda}
	\sqrt{A'\left(x+\frac \lambda I\right)\Phi_{\mu,\nu}'\left(x-\frac \lambda I\right)}\,\mathrm dx\!\right)^{\!2}\!\!,
	\end{align}
	with the prime superscript in $A'$ and $\Phi_{\mu,\nu}'$ denoting a derivative.
\end{lem}
\begin{proof}
	The first part of the proof is analogous to the proof of \cref{lem:partialfid}. More precisely using \cref{lem:exist-ptilde-2} (in place of \cref{lem:exist-ptilde-1}) and employing the fact that $B_i$ is gapped away from $B_i'$, we can apply the argument presented there to obtain
	\begin{align}
	\liminf\limits_{n\to\infty}F\left(\bmtilde{P}^n,\bm P^n\right)\geq \liminf\limits_{n\to\infty}\left(\sum_{j\in B}\sqrt{S^n_jP^n_j}\right)^2.
	\end{align}
	Inserting the definition of the scaled distribution $\bm S^n$ yields
	\begin{align}
	&\liminf_{n\to\infty} \sqrt{F\left(\bmtilde P^n,\bm{P}^{n\downarrow}\right)}\notag\\
	&~\quad\geq \liminf_{n\to\infty}\sum_{i=-I}^{I-1}\sum_{j\in B_i}\sqrt{S^n_jP^{n\downarrow}_j}\nonumber\\
	&~\quad= \liminf_{n\to\infty}\sum_{i=-I}^{I-1}\sum_{j\in B_i}\sqrt{\frac{\sum_{k\in B_{i}'} Q^{n\downarrow}_k}{\sum_{k\in B_i} P^{n\downarrow}_k}\cdot P^{n\downarrow}_j}\sqrt{P^{n\downarrow}_j}\nonumber\\
	&~\quad= \liminf_{n\to\infty}\sum_{i=-I}^{I-1}\sqrt{\sum_{j\in B_{i}'} Q^{n\downarrow}_j\sum_{k\in B_i}P^{n\downarrow}_k}.
	\end{align}
	Recalling that $\Phi(y_i)=A(x_{i+2})$ and applying \cref{lem:clt} gives
	\begin{subequations}
	\begin{align}
	\lim\limits_{n\to\infty} \sum_{j\in B_i}P^{n\downarrow}_j&=\Phi_{\mu,\nu}(x_{i+1})-\Phi_{\mu,\nu}(x_{i}),\\
	\lim\limits_{n\to\infty} \sum_{j\in B_i'}Q^{n\downarrow}_j&=A(x_{i+3})-A(x_{i+2}).
	\end{align}
	\end{subequations}
	Substituting these into our lower bound on fidelity yields
	\begin{align}
	&\liminf_{n\to\infty} \sqrt{F\left(\bmtilde P^n,\bm{P}^{n\downarrow}\right)}\notag\\
	&\qquad\quad\geq \sum_{i=-I}^{I-1}\sqrt{A(x_{i+3})-A(x_{i+2})}\notag\\
	&\qquad\qquad\quad\qquad\times\sqrt{\Phi_{\mu,\nu}(x_{i+1})-\Phi_{\mu,\nu}(x_i)}\\
	&\qquad\quad\geq \sum_{i=-I}^{I-1}
		\sqrt{A\left(\frac{\lambda(i+2)}{I}\right)-A\left(\frac{\lambda (i+1)}{I}\right)}
	\notag\\
	&\qquad\qquad\quad\qquad\times
		\sqrt{\Phi_{\mu,\nu}\left(\frac{\lambda i}{I}\right)-\Phi_{\mu,\nu}\left(\frac{\lambda (i-1)}{I}\right)
	}.\notag
	\end{align}
	
	Using the differentiability of $A$ and $\Phi_{\mu,\nu}$ we can express these finite differences as integrals
	\begin{align}
	&\liminf_{n\to\infty} \sqrt{F\left(\bmtilde P^n,\bm{P}^{n\downarrow}\right)}\notag
	\\&\qquad\qquad
	\geq \sum_{i=-I}^{I-1}
	{
		\sqrt{\int_{\lambda\frac{i}{I}}^{\lambda\frac{i+1}{I}}A'\left(x+\frac\lambda I\right)\,\mathrm dx}
	}\notag\\
	&\qquad\qquad\qquad\qquad\times
	{
		\sqrt{\int_{\lambda \frac iI}^{\lambda \frac{i+1}{I}}\Phi'_{\mu,\nu}\left(x-\frac\lambda I\right)\,\mathrm dx}
	}.
	\end{align}
	Finally, we apply the Schwarz inequality to arrive at the desired bound
	\begin{align}
	&\liminf_{n\to\infty} \sqrt{F\left(\bmtilde P^n,\bm{P}^{n\downarrow}\right)}\notag
	\\&\qquad
	\geq\!\! \sum_{i=-I}^{I-1}\int_{\lambda \frac i I}^{\lambda\frac{i+1}I}\!\!\!\!
	\sqrt{A'\left(x+\frac \lambda I\right)\Phi'_{\mu,\nu}\left(x-\frac \lambda I\right)}\,\mathrm dx
	\nonumber\\&\qquad
	=\int_{-\lambda }^{\lambda}
	\sqrt{A'\left(x+\frac \lambda I\right)	\Phi'_{\mu,\nu}\left(x-\frac \lambda I\right)}\,\mathrm dx.
	\end{align}
\end{proof}

Now, by taking the appropriate limits of our parameters $A$, $\lambda$ and $I$, we will get the desired achieveability bound on the optimal infidelity, and therefore also on the optimal rate. 
\begin{proof}[Proof of \cref{prop:gen} (achieveability)] 
	By \cref{lem:achieve-finite} we know that there exists a family of distributions $\bmtilde{P}^n$ majorising $\v{Q}^n$ and such that $\liminf\limits_{n\to \infty} F\left(\bmtilde{P}^n,\bm{P}^{n\downarrow} \right)$ is lower-bounded by
	\begin{align}
	\left(\int_{-\lambda}^{\lambda}
	\sqrt{A'\left(x+\frac \lambda I\right)\Phi_{\mu,\nu}'\left(x-\frac \lambda I\right)}\,\mathrm dx\right)^2.
	\end{align}
	Due to the equivalence between pre- and post-majorisation, Lemma~\ref{lem:pre_post_equiv}, this means that there exists a family of distributions $\bmtilde{Q}^n$ that is majorised by the total initial state $\bm{P}^n$ and such that their fidelity with the total target state $\bm Q^n$ is also lower bounded by the above expression, which implies a lower bound on the asymptotic ideal fidelity 
	\begin{align}
		&\liminf\limits_{n\to\infty}\sqrt{1-\epsilon^*_0(n,R_\mu)}
		\notag\\
		&~~\quad\geq
		\liminf\limits_{n\to\infty}	\sqrt{F\left({\bmtilde Q^{n}},{\bm Q^{n\downarrow}}\right)}\nonumber\\
		&~~\quad\geq \int_{-\lambda}^{\lambda}
		\sqrt{A'\left(x+\frac \lambda I\right)\Phi_{\mu,\nu}'\left(x-\frac \lambda I\right)}\,\mathrm dx.
	\end{align}
	As the left hand side is independent of $I$, $\lambda$ and $A$, we can now take the desired limits. Note that the order of limits will be important: first we will take $I\to\infty$, then $\lambda\to\infty$, followed by a supremum over $A$. 
		
	Firstly, we take the limit inferior $I\to\infty$. As a consequence of the fact that $\lambda$ is still finite, together with the continuous differentiability of $A$ and $\Phi_{\mu,\nu}$, we have the point-wise limit
	\begin{align}
		\sqrt{A'\left(x+\frac \lambda I\right)}\sqrt{\Phi_{\mu,\nu}'\left(x-\frac \lambda I\right)}
		\xrightarrow{I\to\infty}
		\sqrt{A'(x)}\sqrt{\Phi_{\mu,\nu}'(x)}.
	\end{align}
	Using the compactness of $[-2\lambda,2\lambda]$ we can apply the dominated convergence theorem to move this limit inside the integral, which gives
	\begin{align}
	&\liminf\limits_{n\to\infty}\sqrt{1-\epsilon^*_0(n,R_\mu)}\notag\\
	&\qquad\geq\liminf\limits_{I\to\infty}\!\int_{-\lambda}^{\lambda}
	\sqrt{A'\left(x+\frac \lambda I\right)\Phi_{\mu,\nu}'\left(x-\frac \lambda I\right)}\,\mathrm dx\!\!\nonumber\\
	&\qquad=\int_{-\lambda}^{\lambda} \lim\limits_{I\to\infty}	\sqrt{A'\left(x+\frac \lambda I\right)\Phi_{\mu,\nu}'\left(x-\frac \lambda I\right)}\,\mathrm dx\nonumber\\
	&\qquad=\int_{-\lambda}^{\lambda}
	\sqrt{A'\left(x\right)\Phi_{\mu,\nu}'(x)}\,\mathrm dx.
	\end{align}
	
	Secondly, we want to take $\lambda\to\infty$. The existence of this limit follows from monotone convergence theorem, which we can apply due to the monotonicity of $A$ and $\Phi_{\mu,\nu}$, together with the boundedness of the continuous fidelity. Taking the limit gives us a bound in terms of the continuous fidelity
	\begin{align}
	&\!\!\!\!\liminf\limits_{n\to\infty}\sqrt{1-\epsilon^*_0(n,R_\mu)}\notag\\
	&\geq\lim\limits_{\lambda\to\infty}  \int_{-\lambda}^{\lambda}
	\sqrt{A'\left(x\right)\Phi_{\mu,\nu}'(x)}\,\mathrm dx\nonumber\\
	&\geq\int_{-\infty}^{\infty}
	\sqrt{A'\left(x\right)\Phi_{\mu,\nu}'(x)}\,\mathrm dx=\sqrt{\mathcal F\left({A'},{\Phi_{\mu,\nu}'}\right)}.
	\end{align}
	
	Lastly, we want to take a supremum over all continuously differentiable monotone functions $1\geq A\geq \Phi$, which gives us the Rayleigh-normal distribution
	\begin{align}
	\liminf\limits_{n\to\infty}\epsilon^*_0(n,R_\mu)&\!\leq\! 1\!-\!\sup_{A\geq \Phi}\mathcal F\left({A'},{\Phi_{\mu,\nu}'}\right)	\!=:\!Z_\nu(\mu).
	\end{align}
	
	Using the above and the duality property of Rayleigh-normal distributions, Eq.~\eqref{eq:duality}, we obtain the lower bound on the optimal rate 
	\begin{align}{
		R^*_0(n,\epsilon)
		\gtrsim
		\frac{1}{D(\bm q)}\Biggl[D(\bm p)
		+\sqrt{\frac{V(\bm p)}{n}}Z^{-1}_{1/\nu}(\epsilon)
		\Biggr].
	}\end{align}
\end{proof}

\subsubsection{Optimality}
\label{subsubsec:gen optimality}

We now proceed to the proof of the optimality of \cref{prop:gen}. To this end, we will employ two lemmas originally proved in Ref.~\cite{kumagai2017second}. 
The idea is to start by showing that, after a particular coarse-graining, the fidelity between $\Phi$ and $\Phi_{\mu,\nu}$ is close to the optimal fidelity $\sup_{A\geq \Phi}\mathcal{F}(A',\Phi'_{\mu,\nu})={1-Z_\nu(\mu)}$.

\begin{lem}[Lemma~17 of Ref.~\cite{kumagai2017second}]
	\label{lem:monotonicity}
	For any \mbox{$\zeta>0$}, there exist real numbers \mbox{$s\leq t\leq t'\leq s'$} such that ${\Phi'(x)}/{\Phi'_{\mu,\nu}(x)}$ is strictly monotone decreasing for \mbox{$x\in(s,s')$} and
	\begin{subequations}
	\begin{align}
	\frac{\Phi(t)}{\Phi_{\mu,\nu}(t)}&=\frac{\Phi'(s)}{\Phi'_{\mu,\nu}(s)},\\ \frac{1-\Phi(t')}{1-\Phi_{\mu,\nu}(t')}&=\frac{\Phi'(s')}{\Phi'_{\mu,\nu}(s')}.
	\end{align}
	\end{subequations}	
	Moreover, if we define $\mathcal{F}_{t,t'}(\cdot,\cdot)$ to be the fidelity of distributions which have been coarse-grained on $x\leq t$ and $x\geq t'$, specifically
	\begin{align}
	{\mathcal{F}_{t,t'}(p,q)}:=&
	\Bigg(\sqrt{\int_{-\infty}^{t}p(x)\,\mathrm{d}x\int_{-\infty}^{t}q(x)\,\mathrm{d}x}\notag\\
	&\quad+\int_{t}^{t'}\sqrt{p(x)q(x)}\,\mathrm dx\notag \\
	&\quad+\sqrt{\int_{t'}^{\infty}p(x)\,\mathrm{d}x\int_{t}^{\infty}q(x)\,\mathrm{d}x}\,\,\Bigg)^2\!,
	\end{align}
	then this coarse-grained fidelity has an upper bound
	\begin{align}
		\mathcal{F}_{t,t'}(\Phi',\Phi_{\mu,\nu}')-\zeta&\leq\sup_{A\geq \Phi}\mathcal{F}(A',\Phi_{\mu,\nu}').
	\end{align} 

\end{lem}

Notice that $\Phi$ and $\Phi_{\mu,\nu}$ are exactly the distributions that appear in central limit theorem, \cref{lem:clt}. Thus, we would like to relate the fidelity $F(\bm P^n,\bmtilde{P}^n)$ back to the Rayleigh-normal distribution via this coarse-grained fidelity between Gaussians. To be able to argue this for \emph{any} $\bmtilde P^n\succ \bm Q^n$, we will first give a sufficient condition for a distribution $\bm a$ to have the highest possible fidelity with respect to a second distribution $\bm b$ among all distributions satisfying a majorisation-like condition.

\begin{lem}[Lemma~15 of Ref.~\cite{kumagai2017second}]
	\label{lem:ratio}
	Let $\bm a$ and $\bm b$ be probability distributions such that $a_i/b_i$ is strictly decreasing for all $i$. Then, for any distribution $\bm c$ such that
	\begin{align}
		\sum_{i=1}^{k}a_i\leq \sum_{i=1}^{k}c_i\quad\forall k,
	\end{align}
	we have
	\begin{align}
	F({\bm c},{\bm b})\leq F({\bm a},{\bm b}),
	\end{align}
	with equality if and only if $\bm c=\bm a$.
\end{lem}

We are now ready for the optimality proof.
\begin{proof}[Proof of \cref{prop:gen} (optimality)]
	To prove optimality we need to show that for any $\bmtilde{P}^n\succ \bm Q^n$, the infidelity between $\bmtilde{P}^n$ and $\bm{P}^n$ can be lower bounded by the Rayleigh-normal distribution. This, through Lemma~\ref{lem:pre_post_equiv}, will yield a lower bound on the infidelity between any final state $\bmtilde Q^n$ (i.e., any distribution majorised by the total initial state $\bm P^n$) and the total target state $\bm Q^n$. We will start by using the monotonicity of fidelity under coarse-graining to bound the fidelity between $\bmtilde{P}^n$ and $\bm{P}^n$ by the fidelity between their coarse-grained versions (with coarse-graining over particularly chosen bins of indices). We will then use \cref{lem:ratio} (along with the monotonicity properties of \cref{lem:monotonicity}) to bound the fidelity between coarse-grained versions of $\bm P^n$ and $\bmtilde{P}^n$ by the fidelity between coarse-grained versions of $\bm P^n$ and $\bm Q^n$. Next, by applying \cref{lem:clt}, we will show that this coarse-grained fidelity asymptotes to a fidelity between Gaussians. We will conclude by returning to the last part of \cref{lem:monotonicity}, which will allow us to give a final bound in terms of the Rayleigh-normal distribution.
	
	Fix $\zeta>0$ and $I\in \mathbb{N}$. Let $t,t'\in \mathbb{R}$ be those given by \cref{lem:monotonicity} and introduce 
	\begin{align}
	z_i:=t(1-i/I)+t'(i/I)
	\end{align}
	for \mbox{$0\leq i\leq I$}. Define a set of bins
	\begin{align}
	B_i&:=\lbrace K_n(z_i), \dots, K_n(z_{i+1})-1\rbrace
	\end{align}
	for $0\leq i<I$, as well as two end bins
	\begin{subequations}
	\begin{align}
	B_{-1}&:=\lbrace 1, \dots, K_n(t)-1\rbrace,\\
	B_I&:=\lbrace K_n(t'),\dots,\infty\rbrace,
	\end{align}
	\end{subequations}
	so that $\lbrace B_i\rbrace$ now gives a partition of all indices. Using the monotonicity of fidelity under coarse-graining we get
	\begin{align}
	\sqrt{F(\bmtilde P^n,\bm P^{n})}&=\sum_{i=-1}^I\sum_{j\in B_i} \sqrt{\tilde P^{n\downarrow}_jP^{n\downarrow}_j}\nonumber\\
	&\leq \sum_{i=-1}^I \sqrt{\sum_{j\in B_i}\tilde P^{n\downarrow}_j\sum_{j\in B_i}P^{n\downarrow}_j}=:r_n.\label{eqn:coarse}
	\end{align}

	We now define the limiting coarse-grained versions of total initial and target distributions
	\begin{subequations}
	\begin{align}
	a_i&:=\lim\limits_{n\to\infty}\sum_{j\in B_i}Q^{n\downarrow}_j,\\
	b_i&:=\lim\limits_{n\to\infty}\sum_{j\in B_i}P^{n\downarrow}_j,
	\end{align}
	\end{subequations}
	for $-1\leq i\leq I$. We note that the above limits all exist by \cref{lem:clt}, specifically 
	\begin{subequations}
	\begin{align}
	a_{-1}&=\Phi(t),\\
	a_{i}&=\Phi(z_{i+1})-\Phi(z_{i}),\\
	a_{I}&=1-\Phi(t'),\\
	b_{-1}&=\Phi_{\mu,\nu}(t),\\
	b_{i}&=\Phi_{\mu,\nu}(z_{i+1})-\Phi_{\mu,\nu}(z_{i}),\\
	b_{I}&=1-\Phi_{\mu,\nu}(t').
	\end{align}
	\end{subequations}
	We would also like to analogously define a distribution $\bm c$ that would be a coarse-grained version of $\bmtilde P^n$, but we have no guarantees that the corresponding limits exist. In lieu of this, we will use the vectorial Bolzano-Weirestrass Theorem\footnote{Start with a subset of indices $\lbrace m_l^{(-2)}\rbrace_l\subseteq \mathbb N$ on which $\lbrace r_{n}\rbrace_n$ converges to its limit superior. The scalar version of Bolzano-Weirestrass gives a subset $\lbrace m^{(i)}_l\rbrace_l\subseteq\lbrace m^{(i-1)}_l\rbrace_l$ on which $c_i$ exists. Applying this for each $-1\leq i\leq I$, one obtains a set of indices $\lbrace m_l\rbrace_l:=\lbrace m_l^{(I)}\rbrace_l$ with the desired property.}, which gives that there exists a strictly increasing set of indices $\lbrace m_l\rbrace_l\subset \mathbb{N}$ such that $\lbrace r_n\rbrace$ limits to its limit superior 
	\begin{align}
	\lim\limits_{l\to \infty} r_{m_l}=\limsup\limits_{n\to \infty} r_{n},
	\end{align}	
	\emph{and} that all the limits
	\begin{align}
	c_i&:=\lim\limits_{l\to\infty}\sum_{j\in B_i}\tilde P^{m_l\downarrow}_j
	\end{align}
	exist, for all $-1\leq i\leq I$.	
	
	We now want to apply \cref{lem:ratio} to bound the fidelity $F(\bm c,\bm b)$ with $F(\bm a,\bm b)$, but first we must show that $a_i/b_i$ is strictly decreasing. \cref{lem:monotonicity} allows us to relate this ratio at the ends, $i=-1$ and $i=I$, to the ratio of Gaussian derivatives,
	\begin{subequations}
	\begin{align}
	&\frac{a_{-1}}{b_{-1}}=\frac{\Phi(t)}{\Phi_{\mu,\nu}(t)}=\frac{\Phi'(s)}{\Phi'_{\mu,\nu}(s)},\\
	&\frac{a_{I}}{b_{I}}=\frac{1-\Phi(t')}{1-\Phi_{\mu,\nu}(t')}=\frac{\Phi'(s')}{\Phi'_{\mu,\nu}(s')}.
	\end{align}
	\end{subequations}	
	For $0\leq i<I$ we can apply Cauchy's mean value theorem, which gives that there exists some $s_i\in(z_i,z_{i+1})$ the ratio of finite differences is given by a ratio of derivatives
	\begin{align}
	\frac{a_i}{b_i}=\frac{\Phi(z_{i+1})-\Phi(z_i)}{\Phi_{\mu,\nu}(z_{i+1})-\Phi_{\mu,\nu}(z_i)}=\frac{\Phi'(s_i)}{\Phi'_{\mu,\nu}(s_i)}.
	\end{align}
	Given that $\lbrace s,s_0,\dots,s_{I-1},s'\rbrace$ is a strictly increasing sequence, and that $\Phi'(x)/\Phi_{\mu,\nu}'(x)$ is strictly decreasing on $(s,s')$ by \cref{lem:monotonicity}, we therefore have that $a_i/b_i$ is strictly decreasing as required.
	
	Now that we have shown that $a_i/b_i$ is strictly decreasing, we can apply \cref{lem:ratio}. This gives us 
	\begin{align}
	\limsup\limits_{n\to\infty}F(\bmtilde P^n,\bm P^{n})\leq F(\bm c,\bm b)&\leq F(\bm a,\bm b).
	\end{align}
	Expanding this out, we have
	\begin{align}
	&\limsup\limits_{n\to\infty}\sqrt{F(\bmtilde P^n,\bm P^{n})}\notag\\
	&\quad\leq \sqrt{\Phi(z_{0})\Phi_{\mu,\nu}(z_{0})}\notag\\
	&\qquad+\sum_{i=0}^{I-1}\sqrt{\Phi(z_{i+1})-\Phi(z_{i})}\sqrt{\Phi_{\mu,\nu}(z_{i+1})-\Phi_{\mu,\nu}(z_{i})}\notag \\
	&\qquad+\sqrt{\left(\vphantom{\frac{num}{den}}1-\Phi(z_{I})\right)\left(1-\Phi_{\mu,\nu}(z_{I})\vphantom{\frac{num}{den}}\right)}.
	\end{align}
	If we take $I\to\infty$, these finite differences above approach derivatives, and we get a bound in terms of a coarse-grained fidelity
	\begin{align}
	&\limsup\limits_{n\to\infty}\sqrt{F(\bmtilde P^n,\bm P^{n})}
	\notag\\&\qquad\qquad
	\leq\sqrt{\Phi(t)\Phi_{\mu,\nu}(t)}\notag\\
	&\qquad\qquad~\quad+\int_{t}^{t'}\sqrt{\Phi'(x)\Phi_{\mu,\nu}'(x)}~\mathrm{d}x\notag\\
	&\qquad\qquad~\quad+\sqrt{\left(\vphantom{\frac{num}{den}}1-\Phi(t')\right)\left(1-\Phi_{\mu,\nu}(t')\vphantom{\frac{num}{den}}\right)}.
	\end{align}
	Finally, we apply the last part of \cref{lem:monotonicity}, which allows us to bound this in terms of the Rayleigh-normal distribution, giving
	\begin{align}
	\liminf\limits_{n\to\infty}\delta(\bmtilde P^n,\bm P^{n})\geq {Z_\nu(\mu)}-\zeta.
	\end{align}
	Taking $\zeta\searrow0$, we can once more use the equivalence between pre- and post-majorisation, Lemma~\ref{lem:pre_post_equiv}, to conclude that
	\begin{align}
	\liminf\limits_{n\to\infty}\epsilon^*_0(n,R_\mu)\geq Z_\nu(\mu).
	\end{align}
	Using the above together with the duality property of Rayleigh-normal distributions, Eq.17~\eqref{eq:duality}, we obtain the upper bound on the optimal rate 
	\begin{align}
		R^*_0(n,\epsilon)
		\lesssim
		\frac{1}{D(\bm q)}\Biggl[D(\bm p)
		+\sqrt{\frac{V(\bm p)}{n}}Z^{-1}_{1/\nu}(\epsilon)
		\Biggr].
	\end{align}
\end{proof}

\section{Outlook}
\label{sec:outlook}

In this paper we have derived the exact second-order asymptotics of state interconversion under thermal operations between any two energy-incoherent states. It is then natural to ask whether such a characterisation is also possible for general, not necessarily energy-incoherent, states. Due to the fact that thermal operations are time-translation covariant, such that coherence and athermality form independent resources~\cite{lostaglio2015description,lostaglio2015quantum}, it seems unlikely that the current approach can be easily generalised. Instead, one would need to rely on the full power of Gibbs-preserving maps~\cite{faist2015gibbs,korzekwa2017structure} that form a superset of the thermal operations. For such maps, we believe that a reasonable conjecture is in fact given by Eqs.~\eqref{eq:intro/general} and~\eqref{eq:intro/creation}, with the relative entropy and relative entropy variance replaced by their fully quantum analogues given in Refs.~\cite{umegaki62} and~\cite{li12,tomamichel12}, respectively. 

We also provided a physical interpretation of our main result by considering several thermodynamic scenarios and explaining how our work can be employed to rigorously address the problem of thermodynamic irreversibility. We derived optimal values of distillable work and work of formation, and related them to the infidelity of these processes. This could potentially be used to clarify the notion of imperfect work~\cite{aberg2013truly,woods2015maximum,ng2017surpassing}, and to construct a comparison platform allowing one to continuously distinguish between work-like and heat-like forms of energy. We also discussed thermodynamic processes with finite-size working bodies, focusing particularly on the optimal performance of heat engines. We have shown that there are non-trivial conditions under which an engine can operate at Carnot efficiency and extract perfect work. This opens the possibility of engineering finite heat-baths and working bodies in order to minimise undesirable dissipation of free energy. Moreover, our formalism is general enough to address other interesting problems involving finite-size baths, like fluctuation theorems, Landauers' principle or the third law of thermodynamics~\cite{campisi2009finite,reeb2014improved,scharlau2018quantum}. 

A number of natural technical extensions to our result suggest themselves. We have used the infidelity as our error measure, and conjecture that \cref{thm:sec-ord thermal} will also hold when $\epsilon$ is a bound on the total variational distance. Our second-order expansion falls into a larger class of results known as \emph{small deviation} bounds, in which we consider a fixed error threshold $\epsilon$. Two natural extensions are to the regime of \emph{large deviations}~\cite{hayashi06}, in which a fixed rate is considered, and \emph{moderate deviations}~\cite{Chubb2017,cheng17}, in which the rate approaches its optimum \emph{and} the error vanishes. Last, but not least, we expect that our treatment of approximate majorisation can be extended to cover other distance measures.

\subsection*{Acknowledgements} 

We want to thank Nelly Ng for insightful discussions on the interplay between smoothing and thermomajorisation, as well as for the comments on the manuscript; Masahito Hayashi for discussions and feedback on early technical notes; Matteo Lostaglio for useful discussions on thermodynamic interpretations of second-order rates; and Joe Renes for helpful comments on the manuscript. CC and KK acknowledge support from the ARC via the Centre of Excellence in Engineered Quantum Systems (EQuS), project number CE110001013. CC also acknowledges support from the Australian Institute for Nanoscale Science and Technology Postgraduate Scholarship (John Makepeace Bennett Gift). MT is funded by an ARC Discovery Early Career Researcher Award, project number DE160100821.

\appendix

\section{Work extraction with Carnot efficiency}
\label{app:optimal_engine}

Here we consider a thermodynamic process involving a heat engine with a finite working body. Suppose the working body starts at the cold temperature, $T_c$. Then, the heat flowing from the hot bath will steadily increase the temperature of the working body, from $T_{\mathrm{c}}$ to $T_{\mathrm{c'}}$. We assume that thermalisation happens on a much shorter time-scale than the heat flow, so that the working body is at all times in thermal equilibrium. We are now interested in the amount of work that can be extracted if the engine at each time step operates with the maximal allowed Carnot efficiency (with respect to the background hot temperature and the  instantaneous temperature of the working body).

By definition, the efficiency of an infinitesimal process involving the flow of heat $dQ_{\mathrm{in}}$ from the hot bath to the engine, $dQ_{\mathrm{out}}$ of which flows out to the cold bath, while the remaining energy is converted into work $dW$, is given by
\begin{equation}
\eta:=\frac{\mathrm dW}{\mathrm dQ_{\mathrm{in}}}=\frac{\mathrm dW}{\mathrm dW+\mathrm dQ_{\mathrm{out}}},
\end{equation}
where the equality comes from the conservation of energy. Hence, the work extracted during the infinitesimal flow of heat $dQ_{\mathrm{out}}$ into the cold bath at temperature $T_{\mathrm{x}}$ that heats it up by $dT_{\mathrm{x}}$,
\begin{equation}
\mathrm dQ_{\mathrm{out}}=\frac{\mathrm d\langle E\rangle_{\bm \gamma_{\mathrm{x}}}}{\mathrm dT_{\mathrm{x}}}\,\mathrm dT_{\mathrm{x}},
\end{equation}
is given by
\begin{equation}
\mathrm dW=\frac{\eta}{1-\eta}\,\mathrm dQ_{\mathrm{out}}.
\end{equation}
The Carnot efficiency of an engine acting between the hot bath at fixed temperature $T_{\mathrm{h}}$ and the working body at varying temperature $T_{\mathrm{x}}$ is given by
\begin{equation}
\eta_{\mathrm{c}}(T_{\mathrm{x}})=1-\frac{T_{\mathrm{x}}}{T_{\mathrm{h}}}.
\end{equation}
Thus, the work produced by an engine working at maximum allowed efficiency is given by
\begin{align}
\mathrm dW=\left(\frac{T_{\mathrm{h}}}{T_{\mathrm{x}}}-1\right)\frac{\mathrm d\langle E\rangle_{\v{\gamma}_{\mathrm{x}}}}{\mathrm dT_{\mathrm{x}}}\,\mathrm dT_{\mathrm{x}}.
\end{align}

We now want to calculate the total work $W$ extracted by such an optimal engine while the temperature of the working body changes from $T_{\mathrm{c}}$ to $T_{\mathrm{c'}}$,
\begin{align}
W=\int_{T_{\mathrm{c}}}^{T_{\mathrm{c'}}}\mathrm dW.
\end{align}
By simply integrating by parts we get 
\begin{align}
\!\!W\!=\!\left(\frac{T_{\mathrm{h}}}{T_{\mathrm{x}}}-1\right)\langle E\rangle_{\bm \gamma_{\mathrm{x}}}\Biggr|_{T_{\mathrm{x}}=T_{\mathrm{c}}}^{T_{\mathrm{x}}=T_{\mathrm{c'}}}\!+T_{\mathrm{h}}\int_{T_{\mathrm{c}}}^{T_{\mathrm{c'}}}\frac{\langle E\rangle_{\bm \gamma_{\mathrm{x}}}}{T_{\mathrm{x}}^2}\,\mathrm dT_{\mathrm{x}}.\!
\end{align}
The second term can be calculated by switching from temperature $T_{\mathrm{x}}$ to inverse temperature $\beta_{\mathrm{x}}$ and recalling that the average energy is a negative derivative of $\log\Z_{\mathrm{x}}$ over $\beta_{\mathrm{x}}$:
\begin{align}
\nonumber
T_{\mathrm{h}}\int_{T_{\mathrm{c}}}^{T_{\mathrm{c'}}}\frac{\langle E\rangle_{\bm \gamma_{\mathrm{x}}}}{T_{\mathrm{x}}^2}\,\mathrm dT_{\mathrm{x}}=&\frac{1}{\beta_{\mathrm{h}}}\int_{\beta_{\mathrm{c}}}^{\beta_{\mathrm{c'}}}\frac{\mathrm d\log \Z_{\mathrm{x}}}{\mathrm d\beta_{\mathrm{x}}}\,\mathrm d\beta_{\mathrm{x}}\\
=&\frac{1}{\beta_{\mathrm{h}}}\left(\log\Z_{\mathrm{c'}}-\log\Z_{\mathrm{c}}\right).
\end{align}
By noting that the entropy of a thermal equilibrium state is given by
\begin{align}
H(\v{\gamma}_{\mathrm{x}})=\beta_{\mathrm{x}}\langle E\rangle_{\v{\gamma}_{\mathrm{x}}}+\log \Z_{\mathrm{x}},
\end{align}
we thus have
\begin{align}
W=\left(\langle E\rangle_{\v{\gamma}_{\mathrm{c}}}-\frac{H(\v{\gamma}_{\mathrm{c}})}{\beta_{\mathrm{h}}}\right) - \left( \langle E\rangle_{\v{\gamma}_{\mathrm{c'}}}-\frac{H(\v{\gamma}_{\mathrm{c'}})}{\beta_{\mathrm{h}}}\right).
\end{align}
Finally, comparing the above with Eq.~\eqref{eq:rel_ent_gibbs} we arrive at
\begin{align}
W=k_B T_{\mathrm{h}}\bigl(D(\v{\gamma}_{\mathrm{c}}||\v{\gamma}_{\mathrm{h}})-D(\v{\gamma}_{\mathrm{c'}}||\v{\gamma}_{\mathrm{h}})\bigr),
\end{align}
which is equal to the change of free energy of the finite bath.

\section{Proof of Lemma~\ref{lemma:fidelity}}
\label{app:lemma_fidelity}

\subsection{Preliminaries}

In Ref.~\cite{vidal2000approximate} an explicit construction the solution $\tilde{\v{p}}^{\star}$ to the following maximisation problem,
\begin{equation}
\tilde{\v{p}}^{\star}=\argmax_{\tilde{\v{p}}:~\tilde{\v{p}}\succ\v{q}}F(\v{p},\tilde{\v{p}}),
\end{equation}
was given. We will now describe the construction of this optimal distribution, as it is crucial for our proof, the second part which will very closely follow the reasoning presented in Ref.~\cite{vidal2000approximate}. As explained in Section~\ref{sec:approximate}, without loss of generality we will assume that all the distributions are non-increasingly ordered.

First, for any distribution $\v{a}$ define
\begin{equation}
\label{eq:E_Delta}
E^{\v{a}}_k:=\sum_{i=k}^d a_i,\quad \Delta_{k_1}^{k_2}(\v{a}):=\sum_{i=k_1}^{k_2-1} a_i=E^{\v{a}}_{k_1}-E^{\v{a}}_{k_2}.
\end{equation}
Note that $\v{p}\succ\v{q}$ is equivalent to \mbox{$E^{\v{p}}_k\leq E^{\v{q}}_k$} for all $k$. Now, for a given $\v{p}$ and $\v{q}$ the construction of $\tilde{\v{p}}^\star$ is given by the following iterative procedure. Set $l_0=d+1$ and define
\begin{equation}
\label{eq:l_r}
l_{j}:=\argmin_{k<l_{j-1}}\frac{\Delta_k^{l_{j-1}}(\v{q})}{\Delta_k^{l_{j-1}}(\v{p})},\quad r_j:=\frac{\Delta_{l_j}^{l_{j-1}}(\v{q})}{\Delta_{l_j}^{l_{j-1}}(\v{p})}.
\end{equation}
If the minimisation defining $l_j$ does not have a unique solution then $l_j$ is chosen to be the smallest possible. We will also denote by $N$ an index for which $l_N=1$. The $i$-th entry of the optimal vector for $i\in\{l_{j},\dots,l_{j-1}-1\}$ is then given by
\begin{equation}
\tilde{p}^{\star}_i=r_jp_i.
\end{equation}

It is straightforward to verify that $\tilde{\v{p}}^\star$ is normalised, and $\tilde{\v{p}}^\star\succ\v{q}$ as the construction guarantees that \mbox{$E^{\v{p}}_k\leq E^{\v{q}}_k$} for all $k$. Moreover, the optimal fidelity between $\v{p}$ and a distribution that majorises $\v{q}$ is given by
\begin{eqnarray}
\sqrt{F(\v{p},\tilde{\v{p}}^{\star})}&=&\sum_{i=1}^d\sqrt{p_i\tilde{p}_i^{\star}}=\sum_{j=1}^{N}\sum_{i=l_j}^{l_{j-1}-1} \sqrt{p_i\tilde{p}^\star_i} \nonumber\\
&=&\sum_{j=1}^{N} \left(\frac{\Delta_{l_{j}}^{l_{j-1}}(\v{q})}{\Delta_{l_j}^{l_{j-1}}(\v{p})}\right)^{\frac{1}{2}}\sum_{i=l_j}^{l_{j-1}-1} p_i
\nonumber\\
&=&\sum_{j=1}^{N} \left(\Delta_{l_{j}}^{l_{j-1}}(\v{q})\Delta_{l_j}^{l_{j-1}}(\v{p})\right)^{\frac{1}{2}}.
\end{eqnarray}
The crucial observation in proving the optimality of the above, which we will also need in our proof, is that for all $j$ we have
\begin{equation}
r_j< r_{j+1}.
\end{equation}
This follows from the definition of $l_j$, $r_j$ and the fact that for $a,b,c,d> 0$ one has (see Ref.~\cite{vidal2000approximate} for details)
\begin{equation}
\frac{a}{b}\leq\frac{a+c}{b+d}\quad\Longleftrightarrow\quad \frac{a}{b}<\frac{c}{d}.
\end{equation}

\subsection{Proper proof}

\begin{proof}
	We will prove the equality in Eq.~\eqref{eq:opt_fidelity} by showing that the following two inequalities hold
	\begin{subequations}
		\begin{eqnarray}
		\label{eq:lemma3_1}
		\max_{\tilde{\v{p}}:~\tilde{\v{p}}\succ\v{q}}F(\v{p},\tilde{\v{p}})&\leq&\max_{\tilde{\v{q}}:~\v{p}\succ\tilde{\v{q}}}F(\v{q},\tilde{\v{q}}),\\
		\label{eq:lemma3_2}
		\max_{\tilde{\v{p}}:~\tilde{\v{p}}\succ\v{q}}F(\v{p},\tilde{\v{p}})&\geq&\max_{\tilde{\v{q}}:~\v{p}\succ\tilde{\v{q}}}F(\v{q},\tilde{\v{q}}).
		\end{eqnarray}
	\end{subequations}
	We start with the easier part, Eq.~\eqref{eq:lemma3_1}. It is enough to show that the inequality holds for any $\tilde{\v{p}}$ within the constraints. Let us then take any $\tilde{\v{p}}$ such that it majorises $\v{q}$. Due to Theorem~\ref{thm:major} this is equivalent to the existence of a bistochastic matrix $B_0$ such that $B_0\tilde{\v{p}}=\v{q}$. This implies that
	\begin{equation}
	\max_{\tilde{\v{q}}:~\v{p}\succ\tilde{\v{q}}} F(\v{q},\tilde{\v{q}}) = \max_{B} F(\v{q},B\v{p}),
	\end{equation}
	where the maximisation on the right hand side is over all bistochastic matrices $B$. We now observe that
	\begin{subequations}
	\begin{align}
		\max_{B} F(\v{q},B\v{p})&=\max_{B} F(B_0\tilde{\v{p}},B\v{p})\\
		&\geq F(B_0\tilde{\v{p}},B_0\v{p})\\
		&\geq F(\tilde{\v{p}},\v{p}),
	\end{align}
	\end{subequations}
	where in the last step we used the fact that fidelity obeys data processing inequality. We thus have 
	\begin{equation}
	\max_{\tilde{\v{q}}:~\v{p}\succ\tilde{\v{q}}} F(\v{q},\tilde{\v{q}})\geq F(\tilde{\v{p}},\v{p}),
	\end{equation}
	for any $\tilde{\v{p}}$ majorising $\v{q}$, and so Eq.~\eqref{eq:lemma3_1} holds.

	We now proceed to proving Eq.~\eqref{eq:lemma3_2}. It is again enough to show that the inequality holds for any $\tilde{\v{q}}$ within the constraints. Let us then take any $\tilde{\v{q}}$ such that it is majorised by $\v{p}$. We now have
	\begin{eqnarray}
	\sqrt{F(\v{q},\tilde{\v{q}})}&=&\sum_{i=1}^d\sqrt{q_i\tilde{q}_i}=\sum_{j=1}^{N}\sum_{i=l_j}^{l_{j-1}-1} \sqrt{q_i\tilde{q}_i} \nonumber\\
	&\leq & \sum_{j=1}^{N} \left(\sum_{i=l_j}^{l_{j-1}-1}q_i\right)^{\frac{1}{2}} \left(\sum_{i=l_j}^{l_{j-1}-1}\tilde{q}_i\right)^{\frac{1}{2}}\nonumber\\
	&=&\sum_{j=1}^{N}\left(\Delta_{l_j}^{l_{j-1}}(\v{q})\Delta_{l_j}^{l_{j-1}}(\tilde{\v{q}})\right)^{\frac{1}{2}},		
	\end{eqnarray}
where $l_j$ and $\Delta_{k_1}^{k_2}$ are defined as in Eqs.~\eqref{eq:E_Delta}-\eqref{eq:l_r}. We now introduce
\begin{equation}
x_j:=E_{l_j}^{\tilde{\v{q}}}-E_{l_j}^{\v{p}}\geq 0,
\end{equation}
where the inequality holds for every $j$ because $\v{p}\succ\tilde{\v{q}}$. Observing that
\begin{equation}
\Delta_{l_j}^{l_{j-1}}(\tilde{\v{q}})=\Delta_{l_j}^{l_{j-1}}(\v{p})+x_j-x_{j-1}\geq 0
\end{equation}
we arrive at \mbox{$\sqrt{F(\v{q},\tilde{\v{q}})}\leq f(\v{x})$} with

\begin{equation}
\!\!f(\v{x}):=\sum_{j=1}^{N}\left(\Delta_{l_j}^{l_{j-1}}(\v{q})\left(\Delta_{l_j}^{l_{j-1}}(\v{p})+x_j-x_{j-1}\right)\right)^{\frac{1}{2}}\!\!\!\!.\!\!
\end{equation}

We will now show that $f(\v{x})$ achieves its maximum within the positive orthant $x_j\geq 0$ when $\bm x=\bm 0$, which will finish the proof. This is because then
\begin{equation}
F(\v{q},\tilde{\v{q}})\leq F(\v{p},\tilde{\v{p}}^\star)=\max_{\tilde{\v{p}}:~\tilde{\v{p}}\succ\v{q}}F(\v{p},\tilde{\v{p}})
\end{equation}
for any $\tilde{\v{q}}$ majorised by $\v{p}$, which implies Eq.~\eqref{eq:lemma3_2}. First, by direct calculation one can find a matrix $M$ of second derivatives of $f(\v{x})$, i.e., \mbox{$M_{ij}=\frac{\partial^2 f(\v{x})}{\partial x_{i}\partial x_j}$}. Then, using Gershgorin circle theorem, one can verify that $M$ is negative definite in the allowed region of $\v{x}$, so that there are no local extrema and the maximal value must be obtained at the boundary. Finally,
\begin{equation}
\frac{\partial f(\v{x})}{\partial x_j}\Bigg\vert_{x_j=0}= \frac{1}{2}\left(\sqrt{r_j}-\sqrt{r_{j+1}}\right)<0,
\end{equation} 
which means that the maximal value is obtained for \mbox{$\v{x}=[0,\dots,0]$}, which finishes the proof.
\end{proof}

\section{Efficient algorithm for calculating interconversion infidelities}
\label{app:numerics}

The construction given in \cref{app:lemma_fidelity} gives a natural algorithm for calculating $\max_{\bmtilde{p}:\bmtilde{p}\succ\bm q}F(\bm p,\bmtilde p)$. The run-time of this algorithm is $\mathcal{O}(d^2)$, where $d$ is the size of the input distributions. We now want to argue that this algorithm can be adapted for states described by $\bm p^{\otimes n}$, such that the optimal interconversion rates can be numerically calculated in a time which is efficient in $n$, as is done in \cref{fig:numerics2,fig:numerics1}.

The key property we will leverage is that whilst distributions such as $\bm p^{\otimes n}$ have an exponential number of entries, they only possess a polynomial number of \emph{distinct} entries (in this case $\mathcal{O}(n^{d-1})$). As majorisation is invariant under permutations, it is only the distinct entries (and their degeneracies) that are relevant to our calculation. 

Consider taking as input two distributions $\bm P$ and $\bm Q$ which possess $D=\exp(\mathcal{O}(n))$ total entries, but each only $\mathrm{poly}(n)$ distinct entries. This means there exists indices \mbox{$1=i_1<i_2<\dots<i_t<i_{t+1}=D+1$}, where $t=\mathrm{poly}(n)$, such that $P^{\downarrow}_i$ and $Q^{\downarrow}_i$ are constant on each interval $i\in \lbrace i_s,\dots,i_{s+1}-1\rbrace$ for $s\in\lbrace 1,\dots,t\rbrace$.

The main step in the algorithm, and the bottleneck giving an exponential run-time, is calculating the pivot indices $\lbrace l_j\rbrace_j$ used to construct $\bmtilde{P}$. Specifically these take the form
\begin{align}
l_j:=\argmax_{k<l_{j-1}}\frac{\sum_{i=k}^{l_{j-1}-1}Q^{\downarrow}_i}{\sum_{i=k}^{l_{j-1}-1}P^{\downarrow}_i}.
\end{align}
Using the constancy of $P^\downarrow_i$ and $Q^\downarrow_i$ on each of the intervals we can see that, for a fixed $s$, the function being optimised takes the form
\begin{align}
\frac{\sum_{i=k}^{l_{j-1}-1}Q^{\downarrow}_i}{\sum_{i=k}^{l_{j-1}-1}P^{\downarrow}_i}=\frac{\alpha k+\beta}{\gamma k +\delta},
\end{align}
for any $k\in \lbrace i_s,\dots,i_{s+1}-1\rbrace$. As this function is monotonic as a function of $i$, we conclude that the indices $l_j$ must lie on the edges of these intervals, i.e.\ $\lbrace l_j\rbrace_j\subseteq \lbrace i_s\rbrace_s$. This means that we can restrict our attention only to the these `edge indices' without loss of generality, lowering the algorithms run-time down from $\mathcal{O}(D^2)$ to $\mathcal{O}(t^2)$.

Using the above argument, we now have an efficient algorithm for computing $\epsilon^*_0(n,R;\bm p,\bm q)$. Utilising the idea of embedding, this also allows us to calculate the thermal variant of this, $\epsilon^*_\beta(n,R;\bm p,\bm q)$. Finally, by sweeping over $R$, we can use this to calculate $R^*_\beta(n,R;\bm p,\bm q)$. Examples of this are shown in \cref{fig:numerics2,fig:numerics1}.

\section{Approximate thermomajorisation based on total variation distance}
\label{app:total_variation}

Here we show how to modify the proofs of Lemmas~\ref{lem:pre_post_equiv}, \ref{lem:post_thermo_equiv} and \ref{lem:pre_thermo_equiv}, so that they hold for approximate thermomajorisation defined with total variation distance \mbox{$\delta(\v{p},\v{q})=\frac{1}{2}||\v{p}-\v{q}||$}. In all three proofs one needs to replace the statement ``embedding is fidelity-preserving'' with ``embedding preserves the total variation distance'', and ``fidelity is non-decreasing under stochastic maps'' with ``total variation distance is non-increasing under stochastic maps''. This is already enough to prove modified Lemma~\ref{lem:post_thermo_equiv}. In Lemma~\ref{lem:pre_thermo_equiv} we additionally need to replace the expression $\epsilon_0=(3-2\sqrt{2})/6$ with $\epsilon_0=1/6$. Finally, we replace the second part of Lemma~\ref{lem:pre_post_equiv} with the following reasoning.

Assume that $\bm p\succ_\epsilon\bm q$. This means that there exists a $\bm{\tilde q}$ such that $\bm p\succ \bm{\tilde q}$ and $\delta(\bm q,\bm{\tilde q})\leq\epsilon$. Define $M$ to be the largest integer such that
\begin{align}
\sum_{i=1}^{M}p_i^{\downarrow}\leq 1-\epsilon,
\end{align}
and then define $\bm{\tilde p}$ by cutting off the tail of $\bm p$, and placing all of its mass into the largest element
\begin{align}
\tilde p_i:=
\begin{dcases}
p_i^{\downarrow}+\epsilon & \text{for }i=1,\\
p_i^{\downarrow} & \text{for }1<i<M,\\
p_i^{\downarrow}+1-\epsilon-\sum_{i=1}^{M}p_i^{\downarrow} & \text{for }i=M,\\
0 & \text{for }i>M.
\end{dcases}
\end{align}
By definition of $M$ we can see that $\delta(\bm p,\bm{\tilde p})=\epsilon$. We now need to show that $\bm{\tilde p}\succ \bm q$. For $k\geq M$ we have
\begin{align}
\sum_{i=1}^{k} {\tilde p}_i^{\downarrow}=1\geq \sum_{i=1}^{k} q_i^{\downarrow}.
\end{align}
For $k< M$, we can use $\bm p\succ \bm{\tilde q}$ and $\delta(\bm q,\bm{\tilde q})\leq \epsilon$ to give
\begin{align}
\sum_{i=1}^{k} {\tilde p}_i^{\downarrow}
=\sum_{i=1}^{k} p_i^{\downarrow}+\epsilon
\geq \sum_{i=1}^{k}\tilde q_i^{\downarrow}+\epsilon
\geq \sum_{i=1}^{k} q_i^{\downarrow}.
\end{align}
Thus $\bm p \prescript{}{\epsilon}{\succ}~\bm q\impliedby\bm p\succ_\epsilon\bm q$.

~~\\

\section{Proof of Lemma~\ref{lem:clt-magnitude}}
\label{app:lemma_CLT}
\begin{proof}
	Consider a discrete random variable $L=-\log a$, such that $\langle L\rangle_{\v{a}}=H(\bm a)$ and \mbox{$\Var_{\bm a}(L)=V(\bm a)$}, and therefore
	\begin{align}
	\log k_n(x)=n\langle L\rangle_{\v{a}}+x\sqrt{n\Var_{\v{a}}(L)}.
	\end{align}
	If we let $\lbrace L_j\rbrace_{1\leq j\leq n}$ be i.i.d.\ copies of $L$, then we can write the tail bound of $\bm a^{\otimes n}$ in terms of tail bounds on the average of these variables as
	\begin{align}
	&\sum_{i}\left\lbrace 
	\left(\bm a^{\otimes n}\right)_i \middle|
	\left(\bm a^{\otimes n}\right)_i\geq 1/k_n(x)		
	\right\rbrace\\
	&\quad\quad=\sum_{i_1,\dots,i_n}\left\lbrace \prod_{j=1}^{n}a_{i_j} \middle|\prod_{j=1}^{n}a_{i_j}\geq 1/k_n(x)	\right\rbrace\\
	&\quad\quad=\sum_{i_1,\dots,i_n}\left\lbrace \prod_{j=1}^{n}a_{i_j} \middle|-\sum_{j=1}^{n}\log a_{i_j}\geq \log k_n(x)	\right\rbrace\\
	&\quad\quad=\Pr\left[\sum_{j=1}^n L_j\leq n\langle L\rangle_{\v{a}}+ x\sqrt{n\Var_{\v{a}}(L)}\right].
	\end{align}
	By applying the standard central limit theorem
	\begin{align}
	\lim\limits_{n\to\infty}\!\Pr\left[\sum_{j=1}^n \left(\frac{L_j-\langle L\rangle_{\v{a}}}{\sqrt{\Var_{\v{a}}(L)}}\right)\leq x\sqrt{n}\right]=\Phi\left(x\right)\!,
	\end{align}
	we get the desired bound.
\end{proof}

\end{document}